%% file: main.tex
\documentclass[citeautoscript,aps,prl,reprint,superscriptaddress,longbibliographyaps,nofootinbib]{revtex4-2}

\usepackage{
  amsmath,
  amssymb,
  amsthm,
  bbm,
  bm,
  booktabs,
  braket,
  calc,
  dsfont,
  enumitem,
  float,
  graphicx,
  lipsum,
  mathrsfs,
  mathtools,
  tikz,
  upgreek,
  url,
  xparse
}
\graphicspath{{figures/}}

\usepackage[bookmarks=true,colorlinks=true,linkcolor=orange,citecolor=orange,urlcolor=orange,bookmarks]{hyperref}
\usepackage{xcolor}

\usepackage{iftex}
\ifPDFTeX
  \usepackage[utf8]{inputenc}
\fi

\usepackage[capitalise]{cleveref}

\usepackage[most]{tcolorbox}

\newcommand{\id}{\mathbbm{1}}
\newcommand{\tr}{\operatorname{Tr}}
\newcommand{\rrangle}{\rangle \! \rangle}
\newcommand{\llangle}{\langle \! \langle} 
\newcommand{\kket}[1]{|#1 \rrangle}
\newcommand{\bbra}[1]{\llangle #1 |}

\newcommand{\kketbra}[2]{|#1 \rrangle \! \llangle #2|}
\newcommand{\Span}{\operatorname{span}}

\newcommand{\PT}{{\Upsilon_{k:0}}}

\newcommand{\cPTchoi}{\hat{{\Upsilon}}_{T}}

\newcommand{\T}{\operatorname{T}}
\newcommand{\e}{\operatorname{e}}
\newcommand{\dd}{\operatorname{d}\!}
\newcommand{\Tr}{\operatorname{Tr}}

\newcommand{\cmps}[3]{\Ket{\Phi[#1,#2,#3]}}

\newcommand{\cPT}{|\Upsilon_T \rangle}

\newcommand{\tint}{\int_0^T \mathrm{D}^{n}t}

\definecolor{jonas}{rgb}{0.9,0.17,0.31}

\newcommand{\je}[1]{{\color{cyan} #1}}

\newcount\editmode
\ifcase\editmode

    \newcommand{\add}[1]{{\color{blue}#1}}
    \newcommand{\rem}[1]{{\color{red}[#1]}}

    \newcommand{\grem}[1]{{\color{red}(Greg): [#1]}}

    \newcommand{\jrem}[1]{{\color{red}(Jon{\'a}{\v s}): [#1]}}

    \newcommand{\crem}[1]{{\color{red}(Clara): [#1]}}

\or

    \newcommand{\add}[1]{{\color{blue}#1}}
    \newcommand{\rem}[1]{}

    \newcommand{\grem}[1]{}

    \newcommand{\jrem}[1]{}

    \newcommand{\crem}[1]{}

\or

    \newcommand{\add}[1]{#1}
    \newcommand{\rem}[1]{}

    \newcommand{\grem}[1]{}

    \newcommand{\jrem}[1]{}

    \newcommand{\crem}[1]{}

\fi

\newtheorem{theorem}{Theorem}
\newtheorem{corollary}{Corollary}
\newtheorem{lemma}{Lemma}
\newtheorem{definition}{Definition}

\newtheorem{proposition}{Proposition}
\newcounter{savedcorno}

\makeatletter
\renewcommand\onecolumngrid{
\do@columngrid{one}{\@ne}
\def\set@footnotewidth{\onecolumngrid}
\def\footnoterule{\kern-6pt\hrule width 1.5in\kern6pt}
}
\renewcommand\twocolumngrid{
        \def\footnoterule{
        \dimen@\skip\footins\divide\dimen@\thr@@
        \kern-\dimen@\hrule width.5in\kern\dimen@}
        \do@columngrid{mlt}{\tw@}
}
\makeatother

\makeatletter
\newif\ifeqcontrib@this
\newif\ifeqcontrib@any
\eqcontrib@thisfalse
\eqcontrib@anyfalse

\newcommand{\eqcontrib}{\global\eqcontrib@thistrue\global\eqcontrib@anytrue}

\newcommand{\eqcontribmark}{\textsuperscript{\ensuremath{,\#}}}

\newcommand{\eqcontrib@maybe}{%
  \ifeqcontrib@this
    \global\eqcontrib@thisfalse
    \eqcontribmark
  \fi
}

\newcommand{\printEqContrib}{%
  \ifeqcontrib@any
    \begingroup
      \renewcommand{\thefootnote}{\ensuremath{\#}}%
      \footnotetext{These authors contributed equally to this work.}%
    \endgroup
  \fi
}

\def\doauthor#1#2#3{%
  \ignorespaces#1\unskip\@listcomma
  \begingroup
    #3%
  \@if@empty{#2}{\endgroup{}{}}{\endgroup{\comma@space}{}\frontmatter@footnote{#2}}%
  \eqcontrib@maybe
  \space \@listand
}%
\makeatother

\begin{document}
\title{Dilation theorem for continuum quantum stochastic processes}

\author{Jon{\'a}{\v s} Fuksa\eqcontrib}
\email{jonas.fuksa@fu-berlin.de}
\affiliation{Dahlem Center for Complex Quantum Systems, Freie Universit\"at Berlin, 14195 Berlin, Germany}

\author{Clara Wassner\eqcontrib}
\email{c.wassner@fu-berlin.de}
\affiliation{Dahlem Center for Complex Quantum Systems, Freie Universit\"at Berlin, 14195 Berlin, Germany}

\author{Jens Eisert}
\email{jense@zedat.fu-berlin.de}
\affiliation{Dahlem Center for Complex Quantum Systems, Freie Universit\"at Berlin, 14195 Berlin, Germany}
\affiliation{Helmholtz-Zentrum Berlin f{\"u}r Materialien und Energie, Berlin, Germany}

\author{Gregory A. L. White}
\email{gregory.white@fu-berlin.de}
\affiliation{Dahlem Center for Complex Quantum Systems, Freie Universit\"at Berlin, 14195 Berlin, Germany}

\date{\today}

\begin{abstract}
    Connecting mathematical formalism in open quantum systems to its underlying physics necessitates the notion of a dilation, a way to bridge stochastic dynamics with deterministic Schrödinger evolution on a larger space. Although dilations are well known in the literature for states, channels, and quantum combs, there is a substantial gap when it comes to the fully general setting of non-Markovian dynamics on a continuous interval. In an accompanying work, we introduce a continuous process tensor (cPT) framework wherein these processes are represented by vectors in multi-species bosonic Fock spaces of $L^2$ functions. This accommodates all multi-time statistics of a quantum system on an interval. Here, we present a corresponding dilation theorem and show that all cPTs can be approximated arbitrarily well by an explicit closed description with system and environment coupled by a bounded Hamiltonian. As well as its physical appeal and demonstrating completeness of the framework, we show that such regular representations actually underpin the formalism, defining a dense subspace of all cPTs that is easy to handle. This leads to a continuum-suitable version of Choi duality, and permits the rigorous treatment of continuum controls.
\end{abstract}

\maketitle
\printEqContrib

Open quantum systems represent one of the biggest frontiers within modern quantum mechanics. Noise induced by interactions between a \emph{system} (S) and its \emph{environment} (E) can be complex~\cite{aloisio-complexity, dowlingCapturingLongRangeMemory2024}, and warrants comprehensive understanding before it can be suppressed or corrected~\cite{addisDynamicalDecouplingEfficiency2015,berkExtractingQuantumDynamical2023,pollock2018NonMarkovianQuantumProcesses,gribbenUsingEnvironmentUnderstand2022,rivasQuantumNonMarkovianityCharacterization2014,bylickaNonMarkovianityReservoirMemory2014}. Recent years have seen many advances in understanding temporally correlated, or \emph{non-Markovian}, dynamics, particularly under the operational umbrella of process tensors~\cite{pollock2018NonMarkovianQuantumProcesses,whiteDemonstrationNonMarkovianProcess2020,keelingProcessTensorApproaches2026,whiteNonMarkovianQuantumProcess2022,milzQuantumStochasticProcesses2021}.
Nevertheless, this framework is valid only for discretised dynamics, a setting which neither reflects the underlying laws of physics, nor can be fully consistent with experiment.
In an accompanying paper~\cite{wassnerOperationalContinuumLimit2026}, we remedy this situation by introducing a comprehensive framework of \emph{continuous process tensors} (cPTs). 
Correspondingly, they enable an information-theoretic treatment of non-Markovian dynamics in the continuum, useful for understanding, simulating, and characterising such dynamics.
Operationally, we define the set of physical cPTs purely in reference to all in-principle experimentally measurable quantities. But the broader structure of the theory and connection to underlying physics require elucidation, and are the subject of the present work.

Discrete process tensors admit a dilation in the form of Stinespring isometries~\cite{pollock2018NonMarkovianQuantumProcesses, whiteUnifyingNonMarkovianCharacterization2025}. Every finite process tensor on the cPT interval can be hence individually understood as a collection of unitary operations on SE followed by discarding E.
A \emph{stronger} dilation, is a cPT form we term the \emph{process-canonical} representation (PCR), and encodes an explicit time-dependent SE Hamiltonian within the continuous matrix product state (cMPS) ansatz as a simultaneous dilation.
A natural question is then, to what extent can cPTs be understood within the guise of PCRs? In this work, we answer this question by proving a PCR density theorem (\cref{thrm:pcf existence}). Specifically, we prove the statement: given any cPT $\ket{\Phi}$ and approximation parameter $\varepsilon$, one can always find a PCR representation $\ket{\Upsilon}$ such that $\ket{\Phi}$ and $\ket{\Upsilon}$ are within distance $\varepsilon$ of each other in the Fock space norm (see \cref{fig:main}).
\par 

Apart from simplifying the representation and connecting to underlying physics, this theorem is a key element to further results placing the framework on more mathematically rigorous grounds.
Firstly, it enables the rigorous study of continuous instruments, which represent experiments performed on the system.
For example, some experiments access the process at a particular instant, which is only valid under regularity assumptions on the cPT.
Hence, these instruments are strictly speaking not part of the theory.
The PCR density theorem enables a Gelfand triple construction that allows us to treat these instruments in a rigorous way, delineating the conditions under which they are valid.
Secondly, it is the key ingredient allowing the generalisation of the Choi-Jamio{\l}kowski isomorphism~\cite{choi1975completely, jamiolkowski1972linear,frembsVariationsChoiJamiolkowski2024} to the continuum.

\begin{figure*}[t!]
    \centering\includegraphics[width=0.95\linewidth]{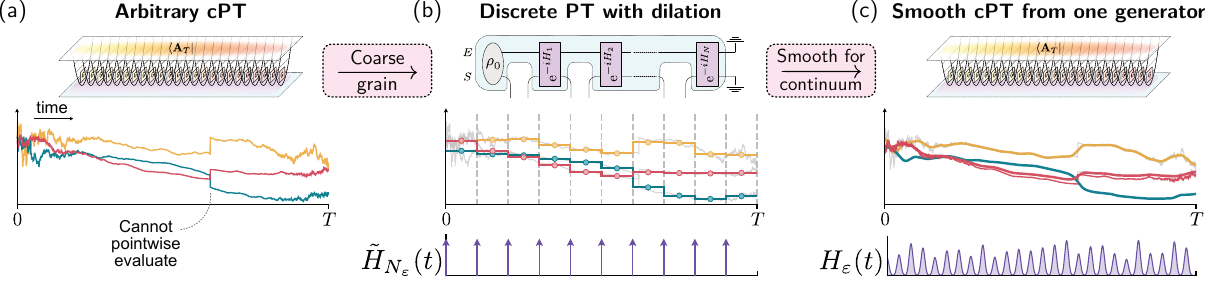}
    \caption{
    Schematic of our main result~(\cref{thrm:pcf existence}). (a) The class $\mathsf{cPT}$ contains processes with discontinuous wavefunctions; these cannot be exactly represented by any single dilation. (b) Nevertheless, we show that these can be arbitrarily approximated by first going to a discrete picture with a single dilation, and then (c) defining from that an appropriately smoothed cPT. 
    }
    \label{fig:main}
\end{figure*}

Dilation theorems of this kind have a lineage in the literature: Stinespring for channels~\cite{kretschmannContinuityTheoremStinesprings2008,nielsenQuantumComputationQuantum2010}, its multi-time analogue for process tensors~\cite{pollock2018NonMarkovianQuantumProcesses,milzQuantumStochasticProcesses2021}, and more recently in lifting one-parameter dynamical curves to analytic Stinespring curves~\cite{burgarthControlQuantumNoise2023,endeFiniteDimensionalStinespringCurves2024}. Our work can be seen as a conceptual continuation of these results, wherein a single time-dependent Hamiltonian must (approximately) reproduce an entire family of process tensors at once.

\emph{Continuous process tensors.---} A controlled experiment on a system $\rm S$ with Hilbert space $\mathcal{H}_{\rm S} = \mathbb C^{d_{\rm S}}$ over a time interval $[0,T]$ consists of a sequence of interventions. In the typical idealised setting, the experiment performs a finite number of instantaneous operations $\{\mathcal{A}_j\}_{j=1}^k \equiv \mathbf{A}$ (described by CP maps) at times $\{t_j\}_{j=1}^k$, and the data returned are a snapshot of a discrete process tensor $\PT$, the specifics of which will be made clear shortly.
We work throughout in the Liouville space of $\rm S$ $\cong \mathcal{H}_{\rm S}\otimes \mathcal{H}_{\rm S}^\ast$ and denote its elements by double-bracket notation $\kket{\cdot}.$ We fix an orthonormal basis of superoperators $\{\mathbb{P}_\nu\}_{\nu = 0}^{d_{\rm S}^4-1}$, with $\mathbb{P}_0$ the identity.
Writing $\nu \equiv (\mu, \mu')$ with $\mathbb{P}_{(\mu,\mu')} = P_\mu \otimes P_{\mu'}^*$ for $\{P_\mu\}_{\mu=0}^{d_{\rm S}^2-1}$ the normalised Paulis on $\rm S$, splits each label into a forward index $\mu$ and a backward index $\mu'$; this fully characterises a discrete process tensor since any control can then be expanded in this basis.
\begin{definition}[Discrete process tensor]\label{def:dpt}
    Let $\bm{t} = (t_1 < \dots < t_k) \in [0,T]^k$.
    A \emph{$k$-step process tensor} is the multilinear functional $\Upsilon_{k:0}$ mapping a sequence of CP maps applied to $\rm S$ at the times $\bm{t}$ to the probability of the associated outcome.
    It is fixed by the entries of its \emph{Choi matrix}, $\Upsilon_{\bm{\nu}}(\bm{t}) \coloneqq \Upsilon_{k:0}[\mathbb{P}_{\nu_k}, \dots, \mathbb{P}_{\nu_1}] = \Upsilon_{(\bm{\mu}, \bm{\mu'})}$ with $\bm{\nu} \in \{0, \dots, d_{\rm S}^4-1\}^k$, and is physical if and only if it is
    \begin{enumerate}[label=(\roman*),leftmargin=*,topsep=2pt,itemsep=0pt]
        \item \underline{Positive:} $\{\Upsilon_{(\bm{\mu},\bm{\mu'})}(\bm{t})\}_{\bm{\mu},\bm{\mu'}} \succcurlyeq 0$ as an operator from $\bm{\mu'}$ to $\bm{\mu}$; and
        \item \underline{Causal:} it returns $1$ on the trivial sequence and, for all $\ell \in [k]$, vanishes whenever a trace-annihilating linear map is applied to the $\ell$-th input and the identity is applied to all the inputs after that.
    \end{enumerate}
\end{definition}
\noindent
Causality is more conveniently phrased in the alternative superoperator basis $\kket{P_{\alpha}}\!\bbra{P_{\alpha'}} = \sum_\nu c^\nu_{\alpha, \alpha'} \mathbb{P}_\nu$: the trace-annihilating maps are exactly the linear combinations of these basis elements with $\alpha \neq 0$.
Although this object describes a discrete quantum stochastic process, it fails to adequately describe the physics of the situation, since the strength of the interventions has no relation to the mathematical structure of the theory.

We have extended the framework of discrete process tensors to the continuum in an accompanying paper~\cite{wassnerOperationalContinuumLimit2026} for finite-dimensional systems and briefly describe its elements here; a continuous-variable version has also recently been introduced~\cite{costaContinuousOperationsNonMarkovian2025}. A continuum framework must supply such an object for every $k$ and every choice of times at once. Rather than positing a single dynamics that underlies different $\{\PT(\bm{t})\}_{k,\bm{t}}$ and deriving the resulting family, we ask for a single object that \emph{contains} all of these finite process tensors.
With a fixed basis, this family can be thought of as a collection of coefficient functions $\Upsilon^{(k)}_{\bm{\nu}} : \Delta_k \to \mathbb{C}$, one for each $k$ and each string $\bm{\nu} \in [q]^k$ of non-identity labels, $[q] = \{1, \dots, q\}$ and $q \coloneqq d_{\rm S}^4 - 1$, where $\Delta_k \subset [0,T]^k$ is the time-ordered simplex.
The natural ansatz we arrive at is that the family as a whole be square-summable. This choice can be made lucid by expanding controlled dynamics as a Dyson series (see SM~\cite{supp} and the companion work~\cite{wassnerOperationalContinuumLimit2026}) and places the cPT in the Fock space
\begin{equation}\label{eq:fock}
    \Gamma_T^{(q)} \coloneqq\bigoplus_{k\geq 0}\operatorname{Sym}^k\!\left(L^2([0,T]) \otimes \mathbb{C}^{q}\right).
\end{equation}
We shall denote this as just `$\Gamma$' when the context is clear.
This square-integrability condition is precisely the mathematical structure that relates to the physicality of interventions, a point we return to in~\cref{cor:gelfand-informal}.
The generating functional we are after is therefore a vector $\ket{\Upsilon}$  in a multi-species bosonic Fock space, with one species per non-identity Pauli $\mathbb{P}_\nu$, vacuum $\ket{\Omega}$ representing idle dynamics, and field operators obeying $[\psi_\nu(t), \psi^\dag_{\nu'}(t')] = \delta_{\nu \nu'} \delta(t - t')$.
The relevant topology is the Fock norm
$\|\Upsilon\|_\Gamma^2 = \sum_{k \ge 0} \sum_{\bm{\nu}} \int_0^T {\rm D}^k \bm{t} | \Upsilon^{(k)}_{\bm{\nu}}(\bm{t})|^2,$
with $\int_0^T {\rm D}^k\bm{t}$ the integral over $\Delta_k$; we see later that it is an operationally meaningful metric on processes.

Coefficients are read off $\ket{\Upsilon}$ by the \emph{intervention operators} $\xi_\nu(t) \coloneqq \psi_\nu(t)$ for $\nu \neq 0$, and $\xi_0(t) = \id$.
These implement an instantaneous $\mathbb{P}_\nu$ on $\rm S$ at time $t$, the case $\nu = 0$ being idling.
Contracting $k$ of them returns exactly the data of~\cref{def:dpt},
\begin{equation}\label{eq:marginal}
    \Upsilon_{\bm{\nu}}(\bm{t}) = \bra{\Omega}\, \xi_{\nu_k}(t_k) \cdots \xi_{\nu_1}(t_1) \ket{\Upsilon}.
\end{equation}
But $\bra{\Omega} \prod_i \xi_{\nu_i}(t_i)$ has distribution-valued amplitudes and is not itself a Fock vector, reflecting the fact that instantaneous control requires unbounded energy.
\Cref{eq:marginal} is therefore meaningful only once smeared against a test function. 
Physicality of $\ket{\Upsilon}$ is then physicality of every finite process tensor contained in it. This leads us to the definition of a \emph{continuous} process tensor. 

\begin{definition}[Continuous process tensor]\label{def:cpt}
    A vector $\ket{\Upsilon_T} \in \Gamma$ with $\braket{\Omega | \Upsilon_T} = 1$ is a \emph{continuous process tensor} (cPT) if for every $k$ and every $f \in L^2(\Delta_k),\ f\ge0$ with $\int_0^T {\rm D}^k\bm{t}|f(\bm{t})| = 1$, the smoothed marginals 
    \begin{equation}
        \Upsilon_{\bm{\nu}}^f \coloneqq \int_0^T {\rm D}^k\bm{t} f(\bm{t}) \Upsilon_{\bm{\nu}}(\bm{t}),
    \end{equation}
    indexed by $\bm{\nu}\in(\{0\}\cup [d_{\rm S}^4 -1])^k$ are the components of a discrete process tensor. We denote the set of all cPTs by $\mathsf{cPT}.$
\end{definition}
The \emph{raison d\textquotesingle{}être} of process tensors is to be an operational theory, mapping events to outcomes. In particular, nothing in~\cref{def:cpt} refers to an underlying dynamics; the conditions are imposed on the generated family of system-level quantities. 
What they do supply, on applying the dilation theorem for discrete process tensors~\cite{pollock2018NonMarkovianQuantumProcesses,milzQuantumStochasticProcesses2021} to each marginal \emph{separately}, is that for every $f$ and $k$ there exists a finite-dimensional environment, a pure state $\kket{\psi_0}$ and unitary channels $\{\mathcal{U}_j\}_{j=0}^{k}$ such that
\begin{equation}\label{eq:physicality}
    \Upsilon^f_{\bm{\nu}} = \bbra{\id} \mathcal{U}_k \mathbb{P}_{\nu_k} \mathcal{U}_{k-1} \cdots \mathbb{P}_{\nu_1} \mathcal{U}_0 \kket{\psi_0}.
\end{equation}
Whether a \emph{single} dilation can generate the whole family at once, a property which we term a \emph{strong} dilation, is the main topic of this work. 

Vectors in $\Gamma$ that \emph{are} generated by a single explicit dilation with a bounded generator take a distinguished form, which we term the process-canonical representation (PCR), and which are represented by a cMPS.
Multi-species cMPS~\cite{verstraeteContinuousMatrixProduct2010,
cMPS2,haegemanCalculusContinuousMatrix2013,tjoaContinuousMatrixproductOperators2026} are parametrised by a drift matrix $Q$, jump matrices $\{R_\nu\}$ and a boundary matrix $B$ each acting on an auxiliary bond space of dimension $\chi$.
Writing $K(t) \coloneqq Q(t) \otimes \id + \sum_\nu R_\nu(t) \otimes \psi_\nu^\dag(t)$, this object is defined as
\begin{equation}\label{eq:cmps}
    \ket{\Phi_T[Q, \{R_\nu\}, B]} \coloneqq \tr\left[B\, \mathcal{T} \e^{\int_0^T \dd t\, K(t)}\right] \ket{\Omega}.
\end{equation}
In the case when $\chi = d_{\rm S}^2 d_{\rm E}^2$, the drift matrix is a valid SE Liouvillian generator of dynamics, the jump matrices are $\mathbb{P}_\nu \otimes \id_{\rm E}$ and $B = \kket{\rho_{\rm SE}}\!\bbra{\id_{\rm SE}}$ with $\rho_{\rm SE}$ a quantum state, the cMPSs provide a natural dilation of cPTs: computing the coefficients of~\cref{eq:cmps} and comparing with~\cref{eq:physicality} identifies the bond space with some SE dynamics. Specifically, this leads to the following. 
\begin{definition}[Process-canonical representation]\label{def:pcr}
    A continuous process tensor on a $d_{\rm S}$-dimensional system for the interval $[0,T]$ is in the \emph{process-canonical representation} (PCR) if it can be written as the cMPS
    \begin{equation}\label{eq:pcr}
        \cPT = \ket{\Phi_T[\mathbb{H}_{\rm SE}, \{\mathbb{P}_\nu\}, \kket{\rho_{\rm SE}}\!\bbra{\id_{\rm SE}}]},
    \end{equation}
    where $\mathbb{H}_{\rm SE}(t) = - \mathrm{i} H(t) \otimes \id + \mathrm{i} \id \otimes H(t)^*$ for $H$ an $\rm SE$ operator-valued function on $[0,T]$ that is (i) self-adjoint and (ii) bounded measurable, with the operator norm $\|H\|_\infty\in L^\infty([0,T])$; and $\rho_{\rm SE}$ is pure.
    We denote the set of cPTs admitting such a representation by $\mathsf{PCR}$.
\end{definition}
We identify this representation as a noteworthy subclass of the cPTs, owing to the facts that they are manifestly physical, have finite-dimensional bond spaces, and have coefficients that may be pointwise evaluated.
It is natural to first ask whether all cPTs admit this explicit form. We show this in the negative, i.e., that it is not possible to find an exact dilation for every cPT, as formalised by the following result.

\begin{proposition}[$\mathsf{PCR} \subsetneq \mathsf{cPT}$]\label{prop:pcr strict}
    Every PCR is a continuous process tensor, but there exist continuous process tensors with no exact process-canonical representative.
\end{proposition}

\begin{proof}[Proof sketch]
    Inclusion is immediate: \cref{eq:pcr} generates every smoothed marginal by unitary dynamics, so~\cref{eq:physicality} holds.
    For strictness, consider Brownian motion $B(t)$ as a stochastic process.
    Fix $\tau \in (0,T)$ and drive $\rm S$ for a given path $\omega$ with the path-dependent qubit unitary $V_t(\omega) = X^{\sigma_t(\omega) \oplus \sigma_0(\omega)}$, where $\sigma_t(\omega)$ is 1 if $B(t,\omega) > B(\tau,\omega)$ and 0 otherwise. Defining process tensors $\Upsilon^{(\omega)}(t,s)$ with respect to the propagator from $s$ to $t$, and then averaging over paths yields a positive, causal vector.
    Self-similarity of $B$ at $\tau$ makes the limit of $\Upsilon^{(k)}_{\bm{\nu}}(\bm{t})$ as $t_i, t_j \to \tau$ depend on the approach: the cones $\{t_i < t_j < \tau\}$ and $\{t_i < \tau < t_j\}$ give different values, and both have positive measure.
    Hence no continuous function agrees with $\Upsilon^{(k)}_{\bm{\nu}}$ almost everywhere, and $\ket{\Upsilon_T}$ is $L^2$-inequivalent to every PCR. We discuss this more explicitly in the SM~\cite{supp}.
\end{proof}
\noindent
Apart from showing a separation between $\mathsf{PCR}$ and $\mathsf{cPT}$, this example shows that probing the process at the discontinuity -- e.g. performing a positive operator-valued measurement (POVM) at $\tau$ -- is ill defined in general.
We will return to this point at a later stage.

\emph{Continuum controls.---} Process tensors, by themselves, are not meaningful objects without a well-defined notion of a dual, provided by the experimental control operations~\cite{milzQuantumStochasticProcesses2021}. 
Here, we introduce the notion of a process statistic: a positive $L^1$ functional on cPTs.

\begin{definition}[Process statistics and their effects]\label{def:instrument}
    A \emph{process statistic} is a pair $(\mathcal{X},\mathcal I_{\mathcal X})$ of a measure space $\mathcal X = (X,\mathcal F, \mu)$ with outcome set $X$, $\sigma$-algebra $\mathcal F$, $\sigma$-finite reference measure $\mu$ together with a linear map $\mathcal{I}_{\mathcal{X}} :  \Gamma \to L^1(\mathcal{X})$ assigning to each process a probability density $p_{\ket{\Upsilon}}(x) = \braket{\mathcal{I}(x) | \Upsilon}$ such that for all $\ket{\Upsilon} \in \mathsf{cPT}$ it is non-negative $\mu$-a.e. and  normalised  $\int_X \mathrm{d}\mu(x) \braket{\mathcal{I}(x) | \Upsilon} = 1$. 
    For any $E \in \mathcal F$ the outcome functionals $\bra{\mathcal{I}(E)}:\mathsf{cPT} \to [0,1]$, defined by $\bra{\mathcal{I}(E)}\Upsilon\rangle:=\int_{E} \mathrm{d}\mu(x) \bra{\mathcal{I}(x)}\Upsilon\rangle$, are called effects of the continuous statistic.
\end{definition}
\noindent 
The function $x \mapsto \braket{\mathcal{I}(x)|\Upsilon}$ is the Radon--Nikodym derivative (see, e.g., the Appendix of Ref.~\cite{barchielliQuantumTrajectoriesMeasurements2009}) of the probability measure $E \mapsto \braket{\mathcal{I}(E)|\Upsilon}$ governing the outcomes of the statistic with respect to $\mu$.
Note that the normalisation condition implies that $\bra{\mathcal{I}(X)}$ is a \emph{deterministic} continuous statistic, having unit overlap with all cPTs.
An example of such a statistic is coherent control, discussed in more detail in the SM~\cite{supp}. 
We introduce these statistics in order to have a means with which process observables can be discussed, but note that to be physically implementable one also requires a \emph{complete} positivity condition. Physical continuous statistics are called continuous instruments; we introduce and study them in our companion paper~\cite{wassnerOperationalContinuumLimit2026}.
Continuous statistics (and thus also instruments) have a natural induced norm, which we now define. 
\begin{definition}[Instrument norm]\label{def:instrument norm}
    The instrument norm is the operator norm of the map $\mathcal{I}_\mathcal{X}$ induced by $\|\cdot\|_\Gamma$ and $\|\cdot\|_{L^1}$,
    \begin{equation}
        \|\mathcal{I}_\mathcal{X}\|_{\Gamma \to L^1(\mathcal{X})} \coloneqq \sup_{\ket{\Phi} \in \Gamma~:~\|\ket{\Phi} \|_\Gamma \leq 1} \|\mathcal{I}_{\mathcal{X}}(\ket{\Phi})\|_{L^1}.
    \end{equation}
\end{definition}
\noindent 
The connection to the operational resources needed to implement the instrument is particularly interesting in the context of \textit{distinguishability} of processes. 
Indeed, for $\ket{\Upsilon_1}, \ket{\Upsilon_2} \in \mathsf{cPT}$ with $\Delta \coloneqq \ket{\Upsilon_1} - \ket{\Upsilon_2}$, twice the total variation distance of their outcome distributions with respect to a fixed statistic $(\mathcal X,\mathcal{I}_{\mathcal X})$ obeys
\begin{equation}
\begin{split}\label{eq:distinguishability}
   \|\mathcal I_{\mathcal X}(\Delta)\|_1\!=\!\int_{\mathcal X}\dd\mu( x)\,|\langle \mathcal I( x)|\Delta\rangle|\leq\|\mathcal I_{\mathcal X}\|_{\Gamma\!\to\!L^1}\|\Delta\|_\Gamma.
\end{split}
\end{equation}
Thus, the $\Gamma$-norm distance governs how distinguishable two processes are by bounded statistics, and provides operational context to the approximation error ultimately required in our main theorem. In the SM~\cite{supp} we bound the norms of various example controls from our companion paper~\cite{wassnerOperationalContinuumLimit2026} in terms of the Frobenius norm of the relevant generator or jump matrix characterising the statistic. We thus make the observation -- to be a subject of future work -- that there is a direct connection between instrument norms and the energetic cost of their implementation.
Notably, we find that the diffusive measurement, which requires an infinite-energy local oscillator, has a diverging instrument norm.

\emph{Dilations for cPTs.---} With our framework in place, we may now assemble the pieces to connect these continuum non-Markovian processes back to the underlying physics.
The following theorem is the main result of the paper. It certifies that the conditions of~\cref{def:cpt} are \emph{complete}: any vector satisfying them is generated, to arbitrary accuracy, by a single SE Hamiltonian.
\noindent
\begin{theorem}[PCR density]
    Suppose that $\ket{\Psi} \in \Gamma$ satisfies the physicality conditions for a cPT.
    Then, for any $\varepsilon>0$, there exists a process-canonical representation $\ket{\Phi[\mathbb{H}_\varepsilon, \{\mathbb{P}_\nu\}, \kket{\psi_\varepsilon}\!\bbra{\id}]}$, such that
    \begin{equation}
        \|\ket{\Phi[\mathbb{H}_\varepsilon, \{\mathbb{P}_\nu\}, \kket{\psi_\varepsilon}\!\bbra{\id}]} - \ket{\Psi}\|_\Gamma < \varepsilon,
    \end{equation}
    where $\mathbb{H}_\varepsilon$ is the Liouvillian generated by $H_\varepsilon$ satisfying $\sup_{t\in[0,T]}\|H_\varepsilon(t)\|_\infty < \infty$, $\psi_\varepsilon$ is a pure state and the PCR has a finite-dimensional environment.
    \label{thrm:pcf existence}
\end{theorem}
\begin{proof}[Proof sketch (see also \cref{fig:main})]
    The core idea is to project $\ket{\Psi}$ onto vectors with wavefunctions that are piecewise constant on the grid $\Lambda_{N_\varepsilon} \coloneqq \{0, \frac{T}{N_\varepsilon}, \dots, T\}$, where $N_\varepsilon$ is sufficiently large.
    This projection can be understood as computing an $N_{\varepsilon}$-step discrete marginal $\Psi^{f_{N_\varepsilon}}_{\bm{\nu}}$ of $\ket{\Psi}$ with respect to a test function $f_{N_\varepsilon}(\bm{t}) = \left(\tfrac{N_\varepsilon}{T}\right)^{N_\varepsilon} \bm{1}_{I_1}(t_1) \dots \mathbf{1}_{I_{N_\varepsilon}}(t_{N_\varepsilon})$, where $\bm{1}_I$ is the indicator function on $I\subset [0, T]$ and $I_k = ((k-1)\frac{T}{N_\varepsilon}, k\frac{T}{N_\varepsilon})$.
    By physicality of $\ket{\Psi}$ we can find a dilation 
    \begin{equation}
        \Psi^{f_{N_\varepsilon}}_{\bm{\nu}} = \bbra{\id}  \mathbb{P}_{\nu_{N_\varepsilon}} \mathcal{U}_{N_\varepsilon-1} \mathbb{P}_{\nu_{N_\varepsilon-1}} \dots \mathbb{P}_{\nu_1} \mathcal{U}_0 \kket{\psi_{N_\varepsilon}}
    \end{equation}
    with a finite environment.
    Defining for each $k$ the Liouvillian $\mathbb{H}_k$ such that $\mathcal{U}_k = \e^{\mathbb{H}_k}$. We let
    \begin{equation}
        \tilde{\mathbb{H}}_{N_\varepsilon}(t) \coloneqq \sum_{k=0}^{N_\varepsilon-1} \mathbb{H}_k  \delta\left(t - \frac{Tk}{N_\varepsilon}\right)
    \end{equation}
    be the Liouvillian generated by the distribution-valued Hamiltonian $\tilde{H}_{N_\varepsilon}$.
    The formal cMPS $\ket{\Phi[\tilde{\mathbb{H}}_{N_\varepsilon}, \{\mathbb{P}_\nu\}, \kket{\psi_{N_\varepsilon}}\!\bbra{\id}]}$ matches exactly the piecewise constant projection of $\ket{\Psi}$.
    The construction allows us to control the $L^2$ error between the piecewise constant cMPS and $\ket{\Psi}$.
    Lastly, we can smooth $\tilde{H}_{N_\varepsilon}$ to $H_\varepsilon$ at controlled $L^2$ cost, to satisfy $\|H_\varepsilon(t)\|_\infty<\infty$.
\end{proof}
\noindent
A few comments about this result are in order.
Firstly, the PCR has a clear connection to the Schrödinger equation and can be implemented by coupling the system to a finite-dimensional environment via $H_\varepsilon$, and evolving the pure initial state $\psi_\varepsilon$.
\Cref{thrm:pcf existence} is a strong dilation theorem in the approximate sense and shows that the claimed physicality conditions really do ensure that the process is physical.
Due to~\cref{eq:distinguishability} and the observed connection between the instrument norm and the corresponding energy requirement, the result states that we can always find a PCR such that distinguishing it from the given $\ket{\Psi}$ would require arbitrarily large energy.

We saw in~\cref{prop:pcr strict} that sharp measurements are undefined for generic cPTs. But frequently in quantum information, it is useful to think of instantaneous and idealised measurements -- similar to position eigenstates in standard quantum mechanics~\cite{madridRoleRiggedHilbert2005}.
As it turns out, \cref{thrm:pcf existence} allows us to make such idealised instruments rigorously meaningful in a rigged Hilbert space construction~\cite{gelfandGeneralizedFunctionsApplications2014,wlokaPartialDifferentialEquations1987}, exploiting the regularity properties of PCRs.
\begin{corollary}[Process Gelfand triple, informal]\label{cor:gelfand-informal}
    Regular processes generated by a bounded-energy environment sit inside the space of all cPTs, which in turn sits inside the space of idealised instruments:
    \begin{equation}\label{eq:process gelfand triple2}
        \underbrace{\Lambda_{\rm reg}}_{\text{PCRs}}
        \;\hookrightarrow\;
        \underbrace{\Lambda}_{\text{all cPTs}}
        \;\hookrightarrow\;
        \underbrace{\Lambda_{\rm reg}^{\ast}}_{\substack{\text{idealised}\\\text{instruments}}} .
    \end{equation}
    An instantaneous measurement at time $t$ gives well-defined statistics on every regular process, but for generic cPTs instantaneous-instrument statistics are undefined. Idealised instruments are therefore made meaningful on the dense set of bounded-energy processes, which we formalise rigorously by a Gelfand triple in the SM~\cite{supp}.
\end{corollary}
\setcounter{savedcorno}{\value{corollary}}
\noindent
Our PCR density theorem hence enables a rigorous formulation of instantaneous POVMs inside our framework. But, as we shall see, it is even more central to the overall formalism by allowing us to connect cPTs with operators and thus formulate a Choi duality in the continuum. 

\emph{Choi duality.---} In the discrete setting, the Choi matrix of a process tensor and its vectorisation are related by the isomorphism between operators on a finite-dimensional space and vectors in the doubled space. Loosely speaking, to each process tensor $\Upsilon$ and effect $\mathbf{A}$ (e.g. a sequence of CP maps) one associates positive operators $\hat{\Upsilon}$ and $\hat{\mathbf{A}}$ such that $\Upsilon[\mathbf{A}] = \Tr[\hat{\Upsilon}\hat{\mathbf{A}}^{\T}]$.
At a formal level, this exploits the relation $\mathcal{B}(\mathbb{C}^d)\cong \mathbb{C}^d\otimes \mathbb{C}^d$; physically this corresponds to the cut between the time-forwards and the time-backwards contour.
In the continuum however, such an isomorphism does not exist, because while the physical cut runs between the two Pauli labels $\mu, \mu'$, the Fock structure $\Gamma_T^{(d_{\rm S}^4-1)}$ does not factorise along this axis.
But on the subspace of $\Gamma$ where the cross-species (neither $P_\mu$ nor $P_{\mu'}$ equal to the identity) are unpopulated, the isomorphism can be restored.
As it turns out, the cross-species amplitudes of physical processes are in fact fully determined by the remaining amplitudes. Consequently, projecting them away retains all relevant information while circumventing the non-factorisability obstacle. \cref{thrm:pcf existence} allows us to make this point rigorous and hence formulate a Choi duality in the continuum. We present here briefly the main ingredients and defer the details to the SM~\cite{supp}.
\par 

We write $\gamma \coloneqq \Gamma^{(d_{\rm S}^2-1)}_T$ for the single-contour Fock space.  
We use a caret to denote the Choi operator of a cPT $\cPTchoi: \gamma \rightarrow \gamma$ and $\check{\mathcal{I}}$ for the object dual to it (\cref{def:instrument}) -- with the explicit maps $\mathcal{P}:\cPT\mapsto\cPTchoi$ and $\mathcal{V} : \bra{\mathcal{I}}\mapsto \check{\mathcal{I}}$, elaborated on in the SM~\cite{supp}. 
The \emph{Choi map} $\mathcal P$ retains only 
wavefunctions not containing the cross species, such that the integral kernel of $\cPTchoi$ becomes:
\begin{equation}\label{eq:choi map}
    (\cPTchoi)^{(k;k')}_{\bm{\mu};\bm{\mu'}}(\bm{t};\bm{t'}) \coloneqq \Upsilon^{(k+k')}_{(\bm{\mu}, \bm{0}), (\bm{0}, \bm{\mu'})}(\bm{t}, \bm{t'}),
\end{equation}
which is a Hilbert--Schmidt (HS) operator by construction. 
On the other hand, the \emph{instrument} map $\mathcal{V}$ retains by necessity all cross-species information. Consequently, the dual object $\check{\mathcal I}$ is in general not an operator on $\gamma$ but instead a sesquilinear form defined via an integral kernel. \par

With this context in hand, we now present a theorem formalising the Choi duality in the continuum, enabled by~\cref{thrm:pcf existence}.  In particular, the separate forward and backward Stinespring curves of a PCR make explicit that the cross species factorise, which allows them to be reconstructed from the Choi operator. 
\begin{theorem}[Choi duality in the continuum]\label{cor:choi}
    There exists a reconstruction map $\mathcal{R}$, such that for every $\cPT \in \mathsf{cPT}$, the Choi operator $\cPTchoi$ satisfies $\mathcal{R}\cPTchoi = \cPT$ and it is a positive-semidefinite trace-class operator.
    Hence, for any $\bra{\mathcal{I}} \in \Gamma^*$ it satisfies
    \begin{equation}\label{eq:duality}
       \braket{\check{\mathcal{I}},\cPTchoi}_{\rm ker} = \braket{\mathcal{I} | \Upsilon_T},
    \end{equation}
    with the left hand side being a HS-like sesquilinear form between integral kernels.
\end{theorem}
\begin{proof}[Proof sketch]
       We introduce the operator $\mathcal{R}$ such that $\mathcal{R}\add{\!}\circ\add{\!}\mathcal{P}(\ket{\Upsilon}) = \ket{\Upsilon}$ directly follows for all $\ket{\Upsilon} \in \mathsf{PCR}$.
       \Cref{thrm:pcf existence} and the bounded linear transformation theorem extend this identity to all cPTs, after which~\cref{eq:duality} then follows by direct computation using the definitions of $\mathcal{R}$ and $\mathcal{V}$.
    Positive semidefiniteness of $\cPTchoi$ follows from the positivity of $\cPT$.
    Lastly, the trace of a PCR can be computed as a constant and hence, using~\cref{thrm:pcf existence} and Fatou's lemma, we can upper bound the trace of $\cPTchoi$.
\end{proof}
\noindent
We remark that for sufficiently well-behaved $\ket{\mathcal I}$, $\check{\mathcal{I}}$ is a bounded operator on $\gamma$ and then $\braket{\check{\mathcal{I}},\cPTchoi}_{\text{ker}} \equiv  \Tr\bigl[\check{\mathcal{I}}^\dag\cPTchoi]$ exactly \textit{is} the trace of the operator product $\check{\mathcal{I}}^\dag\cPTchoi$.
Further, note that since $\mathcal{R}\circ\mathcal{P}$ acts as the identity on all cPTs, the causality condition can be equivalently defined in the Choi operator picture, see the SM~\cite{supp}.

\emph{Discussion.---} 
In this work, we have clarified the structure of continuous non-Markovian processes, linking the operational framework introduced in an accompanying work \cite{wassnerOperationalContinuumLimit2026} to its underlying physical setting, while proving several key results pertaining to mathematical representation.
Most importantly, we proved a dilation theorem for cPTs, connecting the abstract conditions of positivity and causality on Fock vectors to a clear dynamical approximation in the $\Gamma$-norm.
A natural next step would be to formalise this connection and understand whether it is meaningful to go beyond this norm in some scenarios, while linking our formalism to the literature on quantum stochastic processes~\cite{accardiQuantumStochasticProcesses1982,Nurdin_2021}.
While the main focus of this paper has been on the processes, finding a structural theorem analogous to~\cref{thrm:pcf existence} for instruments would solidify the mathematical understanding of the process dual.
Further, we have motivated the $\Gamma$-norm as the natural approximation metric for processes by linking the instrument norm to the energy requirements of instruments by means of examples.
We thus pave the way for information-theoretic treatments which respect the physical constraint of finite energy.

\let\oldaddcontentsline\addcontentsline
\renewcommand{\addcontentsline}[3]{}       

\begin{acknowledgments}
\emph{Acknowledgements.---} This work has been supported by the BMFTR (DAQC, MuniQC-Atoms, QuSol, Hybrid++, PasQuops), Clusters of Excellence (ML4Q, MATH+), the Munich Quantum Valley, Berlin Quantum, the Quantum Flagship (Millenion, Pasquans2), the DFG (CRC 183, SPP 2514), the QuantERA (SDPCode), the European Research Council (DebuQC), and the Alexander-von-Humboldt Foundation. 

\emph{AI statement.---} 
AI tools supported us with fundamental mathematical analysis concepts when carrying out the detailed steps of the proofs and we have used them as a final check.
The main ideas, steps and writing were all done without the use of AI.
\end{acknowledgments}

\bibliography{references}

\let\addcontentsline\oldaddcontentsline  

\clearpage

\input{SM}

\end{document}

%% file: SM.tex
\providecommand{\je}[1]{{\color{cyan} #1}}

\onecolumngrid

\setcounter{secnumdepth}{2}
\appendix

\begin{center}
    {\centering \bfseries \large Dilation theorem for continuum quantum stochastic processes\\ Supplemental Material \par}    
\end{center}

\tableofcontents

\section{Background on Fock spaces, continuous process tensors and notational conventions}

In this appendix, we provide a basic introduction to analysis on Fock spaces, introduce our notation, and show how cPTs naturally become elements of these spaces.
For a more thorough introduction to Fock spaces, many standard texts are available; we recommend the excellent texts~\cite{araiInfiniteDimensionalDiracOperators2022,defariaMathematicalAspectsQuantum2010}.

\subsection{Fock spaces}

As presented in the main text and introduced in detail in the accompanying work \cite{wassnerOperationalContinuumLimit2026}, continuous process tensors on a continuous time-interval $[0,T]$ can be described by vectors $\ket{\Upsilon}$ in a bosonic free scalar quantum field theory. 
Mathematically, the relevant space is a bosonic Fock space.
Bosonic Fock spaces are built from a base Hilbert space $\mathcal{H}$ as follows:
\begin{equation}\label{eq:Fock space def}
\Gamma(\mathcal{H}) \coloneqq \bigoplus_{n=0}^\infty \text{Sym}^{\otimes n} (\mathcal{H}),
\end{equation}
where the notation $\text{Sym}^{\otimes n}(\mathcal{H})$ is the symmetrisation of the $n$-fold tensor product of $\mathcal{H}$.
The subspace labeled by $n$ is called the $n$-particle subspace, while the subspace labeled by $n=0$ is one-dimensional and is spanned by the vacuum.
Given the inner product $\braket{\cdot, \cdot}_\mathcal{H}$ on the base space, we can define the inner product on $\mathcal{H}^{\otimes n}$ by sesquilinearly extending
\begin{equation}
    \braket{f_1 \otimes\dots\otimes f_n, g_1\otimes \dots \otimes g_n}_n \coloneqq \prod_{i=1}^n \braket{f_i, g_i}_{\mathcal H},
\end{equation}
which also defines an inner product on $\text{Sym}^{\otimes n}(\mathcal{H})$.
The inner product on $\Gamma(\mathcal{H})$ becomes
\begin{equation}
    \braket{f, g}_{\Gamma(\mathcal{H})} \coloneqq \sum_{n=0}^\infty \braket{f^{(n)}, g^{(n)}}_n,
\end{equation}
where $f^{(n)} \in \text{Sym}^{\otimes n}(\mathcal{H})$ is the $n$-particle sector component of $f$, which in physics-related discussions is often called the $n$-particle \textit{wavefunction} or \textit{amplitude}.
We adopt this terminology here.
This inner product makes $\Gamma$ into a Hilbert space since direct sums of complete Hilbert spaces are complete.
The inner product also defines a norm and a topology on $\Gamma(\mathcal{H})$ via $\|f \|_{\Gamma(\mathcal{H})} \coloneqq \sqrt{\braket{f, f}}$ for any $f \in \Gamma(\mathcal{H})$.
This topology is also how the infinite direct sum in~\cref{eq:Fock space def} should be understood: It is the closure of the span of all vectors with support in a finite number of sectors under this $\|\cdot\|_{\Gamma(\mathcal{H})}$.

In this work, the base Hilbert spaces are either multiple copies of the square-integrable functions $\mathcal{H} = L^2([0, T]) \otimes \mathbb{C}^q \cong \bigoplus_{i=1}^{q}L^2([0,T])$ or of the Sobolev functions $\mathcal{H} = H^1([0,T])\otimes \mathbb{C}^q$, where the relevant inner products are
\begin{equation}
    \braket{f, g}_{L^2} \coloneqq \sum_{\nu=1}^q \int_0^T f_\nu^*(t) g_{\nu}(t) {\rm d}t, \quad \braket{f, g}_{H^1} \coloneqq \braket{f,g}_{L^2} + \braket{f',g'}_{L^2},
\end{equation}
where $'$ labels the weak derivative.
The Fock space built from $L^2$ is the standard multi-species bosonic Fock space describing a scalar bosonic theory with $q$ particle species.
In our case, to describe cPTs on a $d_{\rm S}$-dimensional system, the space $\Gamma$ has $q = d_{\rm S}^4 - 1$ particle species.

We can introduce a particle annihilation operator $a(u)$ for any element $u \in \mathcal{H}$, which maps the vacuum to zero and the $n$-particle component to an $n-1$-particle component by linearly extending
\begin{equation}
    a(u): \frac1{n!} \sum_{\pi \in \mathsf{S}_n} f_{\pi(1)} \otimes \dots \otimes f_{\pi(n)} \mapsto \frac{\sqrt{n}}{n!}\sum_{\pi \in \mathsf{S}_n} \braket{u, f_{\pi(1)}} f_{\pi(2)} \otimes\dots\otimes f_{\pi(n)}.
\end{equation}
The associated creation operator $a^\dag(u)$ is the adjoint of the annihilation operator.
These operators satisfy the commutation relation
\begin{equation}
    [a(u), a^\dag(v)] = \braket{u,v}_{\mathcal{H}}.
\end{equation}
For the standard multispecies bosonic Fock space $\Gamma$ it is common to introduce the field operators $\psi_\nu(x)$, which satisfy
\begin{equation}
    [\psi_\nu(x), \psi_{\nu'}^\dag(x')] = \delta_{\nu, \nu'} \delta(x - x'),
\end{equation}
which can be thought of as annihilation operators corresponding to the Dirac delta function $\delta_x$.
Hence, they are operator-valued \emph{densities} and they only make sense when integrated over a non-zero measure subset of the domain $[0, T]$.
In fact, for any $u$ in $L^2([0, T])$ we can write
\begin{equation}
    a_\nu(u) = \int_0^T {\rm d}t \ u(t)^* \psi_\nu(t).
\end{equation}

With the field operators in place, we can define the wavefunctions $\Upsilon^{(n)}_{\bm{\nu}}(\bm{t})$ of any $\ket{\Upsilon} \in \Gamma$ via
\begin{equation}
    \ket{\Upsilon} = \sum_{n=0}^\infty \sum_{\bm{\nu}\in[q]^n} \int_0^T {\rm D}^n\bm{t} \ \Upsilon_{\bm{\nu}}^{(n)}(\bm{t}) \prod_{i=1}^n \psi_{\nu_i}^\dag(t_i)\ket{\Omega}.
\end{equation}
The inner product on $\Gamma$ can now be written in terms of the wavefunctions via
\begin{equation}
    \braket{\Phi|\Upsilon} = \sum_{n=0}^\infty \sum_{\bm{\nu}\in[q]^n}\int_0^T {\rm D}^n\bm{t} \left(\ \Phi^{(n)}_{\bm{\nu}}(\bm{t})\right)^*\Upsilon^{(n)}_{\bm{\nu}}(\bm{t}).
\end{equation}
The $n$-particle subspace is hence spanned by square-integrable functions, where the inner product is defined with time ordering.
We can equivalently define the wavefunctions to be symmetric and replace the time-ordered integral with an integral over $[0, T]^n$, obtaining a rescaled inner product between the wavefunctions, which is more standard in describing non-relativistic particles.
For processes, however, it is useful to impose the time-ordering.

\subsection{Continuous process tensors as elements of a bosonic Fock space}\label{ssec:cpts in the Fock space}

In this subsection we motivate physically why the Fock space $\Gamma$ from Equation (1) in the main text is the natural space to host continuum process tensors of finite-dimensional systems.
Consider a system $\rm S$ coupled to its environment $\rm E$ by a CPTP generator $\mathcal{L}_{\rm SE}(t)$ throughout a time interval of interest $[0, T]$.
Suppose that $\rm SE$ starts in an initial state $\rho_{\rm SE}$.
We want to perform an experiment on $\rm S$, yielding an outcome from an outcome set $\mathcal X$ and we assume that for every $x \in \mathcal X$ we know the corresponding generator $\mathcal{A}_x(t)$ of CP dynamics on $\rm S$. That is $\sum_{x \in \mathcal X}\mathcal{A}_x(t)$ is a CPTP map. The associated probability density of obtaining the experimental outcome $x \in \mathcal{X}$ is given by
 \begin{equation}
    p(x | \mathcal{A}) \coloneqq \tr_{\rm SE} \left[\mathcal{T}\left(e^{\int_0^T {\rm d}t\,\bigl[\mathcal{L}_{\rm SE}(t) + \mathcal{A}_x(t)\bigr]}\right)(\rho_{\rm SE})\right],
\end{equation}
where $\mathcal{T}$ is the time-ordering operator. 
The goal of the continuous process tensor framework is to compute the associated probability density as a contraction between the \emph{continuous process tensor}, which depends on $\rho_{\rm SE}$ and $\mathcal{L}_{\rm SE}$, and the \emph{continuous instrument}, which depends on $\mathcal{A}_x(t)$.
To achieve this separation, we switch to the Liouville space and use the orthonormal generalised Pauli basis $\{\mathbb{P}_\nu\}_{\nu = 0, \dots, d_{\rm S}^4 - 1}$ for operators on its system part, such that $\mathcal{A}_x(t) = \sum_{\nu = 0}^{d_{\rm S}^4 - 1} c_\nu^x(t) \mathbb{P}_\nu$, see the accompanying work~\cite{wassnerOperationalContinuumLimit2026}.
Now, the separation can be achieved through the Dyson series as
\begin{align}
    \label{eq:dyson}
    p(x|\mathcal{A}) &= \sum_{n=0}^\infty \int_0^T {\rm D}^n \bm{t} \bbra{\id_{\rm SE}} \Lambda(T, t_n) \mathcal{A}_x(t_n) \Lambda(t_n, t_{n-1}) \dots \mathcal{A}_x(t_1) \Lambda(t_1, 0) \kket{\rho_{\rm SE}} \nonumber\\ 
\begin{split}
    &= \sum_{n=0}^\infty \int_0^T {\rm D}^n \bm{t} \sum_{\bm{\nu} \in \{1, \dots, d_{\rm S}^4-1\}^n} e^{\int_0^T {\rm d}t\,c_0^x(t)} c_{\nu_n}^x(t_n) \dots c_{\nu_1}^x(t_1) \bbra{\id_{\rm SE}} \Lambda(T, t_n) \id_{\rm E} \otimes \mathbb{P}_{\nu_n} \Lambda(t_n, t_{n-1}) \dots \\
    &\hspace{3cm}\dots \Lambda(t_2, t_1) \id_{\rm E} \otimes \mathbb{P}_{\nu_1} \Lambda(t_1, 0)\kket{\rho_{\rm SE}}
\end{split}
\end{align}
where $\int_0^T {\rm D}^n\bm{t}$ is the integral over $\bm{t} \in [0,T]^n$ with the time-ordering restriction $t_1<\dots<t_n$, $\Lambda(t',t)=\mathcal T e^{\int_t^{t'}{\rm d}s\,\mathcal L_{\rm SE}(s)}$, and we label vectors in the Liouville space with the double bra-ket notation.
In the following we will suppress $\id_{\rm E}$ when it is clear from the context.
Note that we use the notation $\mathcal{L}_{\rm SE}$ both for the GKSL operator and for its vectorisation, but which one is meant should be clear from the context.

Notice that the expression~\cref{eq:dyson} is an inner product in the space $\Gamma_T^{(d_{\rm S}^4 - 1)}$, the $(d_{\rm S}^4 - 1)$-species bosonic Fock space of the $L^2$ functions on $[0, T]$, where one of the vectors becomes the cPT and the other represents the instrument (see Definition 4 in the main text).
Focusing on the cPTs here, this approach encompasses all processes that have a generator.
In fact, these are exactly the processes that are described by a PCR vector.
But, as we have seen in the main text, the space $\Span(\mathsf{PCR})$ is not complete under the Fock space norm.
Hence, to obtain a Hilbert space for the cPTs, we complete the space under $\|\cdot \|_\Gamma$, which is the starting point of the theory.
The Dyson series illustrates that the $L^2$ inner product is \emph{the} natural inner product on cPTs and hence completing under the $L^2$ Fock norm is \emph{the} natural completion of the space.
However, we can complete the space under other norms, such as the Sobolev norms, which is, for instance, useful to obtain the Gelfand triple construction described in the main text, allowing us to make certain useful instruments rigorous.

\section{Bounds on norms of cPT wavefunctions}
\label{sec:bounds_on_norms_cPT_wavefunctions}

The $\Gamma$ norm can be written as a sum of time-ordered $L^2$ norms of the wavefunctions as
\begin{equation}
    \|\cPT\|_\Gamma = \sqrt{\sum_{n=0}^\infty \int_0^T {\rm D}^n\bm{t} \sum_{\bm{\nu}} | \Upsilon^{(n)}_{\bm{\nu}}(\bm{t})|^2} = \sqrt{\sum_{n=0}^\infty \sum_{\bm{\nu}} \|\Upsilon^{(n)}_{\bm{\nu}}\|^2_2}
\end{equation}
In this appendix we will obtain a general upper bound on each $\|\Upsilon^{(n)}_{\bm{\nu}}\|_2$.

Using the variational definition of the time-ordered $L^2$ norm we have
\begin{equation}
    \|\Upsilon^{(n)}_{\bm{\nu}}\|_2 =  \sup_{\substack{f\in L^2(\Delta_n): \\ \|f\|_2 = 1}} \left|\int_0^T {\rm D}^n\bm{t} \ f(\bm{t}) \Upsilon^{(n)}_{\bm{\nu}}(\bm{t})\right|.
\end{equation}
We can write any function $f\in L^2(\Delta_n)$ as
\begin{equation}
    f(\bm{t}) = \frac1\pi \int_0^{2\pi} {\rm d}\theta \ \e^{i\theta} |f(\bm{t})|\left[1 + \cos(\theta - \phi(\bm{t}))\right] \eqqcolon \int_0^{2\pi} {\rm d}\theta \ \e^{i\theta} f^\theta(\bm{t}),
\end{equation}
where $\phi(\bm{t})$ is the complex argument of $f(\bm{t})$.
This identity can be verified by a direct calculation using the orthogonality properties of $\sin$ and $\cos$.
The crucial observation is that $f^\theta(\bm{t}) \ge 0$ for all $\bm{t} \in \Delta_n$.
Therefore, we can write
\begin{equation}
    \|\Upsilon_{\bm{\nu}}^{(n)}\|_2 = \sup_{\substack{f\in L^2(\Delta_n):\\\|f\|_2 = 1}} \left|\int_0^{2\pi} {\rm d}\theta \ \e^{i\theta}\int_0^T {\rm D}^n\bm{t} \ f^\theta(\bm{t}) \Upsilon^{(n)}_{\bm{\nu}}(\bm{t})\right|.
\end{equation}
Since $f^\theta(\bm{t}) \ge 0$, whenever $\int_{\Delta_n} {\rm d} \bm{t} f^\theta(\bm{t})>0$, the definition of a cPT implies that
\begin{equation}
    \Upsilon^{f^\theta}_{\bm{\nu}} \coloneqq \frac1{\int_0^T {\rm D}^n\bm{t'} \ f^\theta(\bm{t'})}\int_0^T {\rm D}^n\bm{t} \ f^\theta(\bm{t})\Upsilon_{\bm{\nu}}^{(n)}(\bm{t})
\end{equation}
is an element of a valid $n$-step process tensor.
Therefore, we can dilate this object~\cite{pollock2018NonMarkovianQuantumProcesses} and write
\begin{equation}
    \|\Upsilon_{\bm{\nu}}^{(n)}\|_2 = \sup_{\substack{f\in L^2(\Delta_n):\\\|f\|_2 = 1}} \left|\int_0^{2\pi}{\rm d}\theta \ e^{i\theta} \bbra{\id_{\rm SE}} \mathbb{P}_{\nu_n} U_{n-1}\dots U_1 \mathbb{P}_{\nu_1} \kket{\rho} \int_0^T {\rm D}^n\bm{t} f^\theta(\bm{t})\right|,
\end{equation}
where the $U_i$ and $\rho$ depend in general on $f$ and $\theta$, and the $U_i$ are vectorised unitary channels on a $d_{\rm S}^{2n}$ dimensional quantum system and $\rho$ is a quantum state.
Using the definition of the operator norm we can upper bound 
\begin{equation}
    |\bbra{\id_{\rm SE}} \mathbb{P}_{\nu_n} U_{n-1}\dots U_1 \mathbb{P}_{\nu_1} \kket{\rho}| \le \|\bbra{\id_{\rm SE}}\| \|\mathbb{P}_{\nu_1}\|_\infty \dots \|\mathbb{P}_{\nu_n}\|_\infty \|\kket{\rho}\| \le d_{\rm S}^{n}\frac1{d_{\rm S}^n} = 1.
\end{equation}
Using the Cauchy--Schwarz inequality, we get
\begin{equation}
    \int_0^T {\rm D}^n\bm{t} \ f^\theta(\bm{t}) \le \frac2\pi \int_0^T {\rm D}^n\bm{t} |f(\bm{t})| \le \frac2\pi \|1\|_2 \|f\|_2 = \frac2\pi \sqrt{\frac{T^n}{n!}}
\end{equation}
and hence
\begin{align}
    \|\Upsilon^{(n)}_{\bm{\nu}}\|_2 \le 4 \sqrt{\frac{T^n}{n!}}
    \label{eq:wavefunction_bound_cPT}.
\end{align}
From the normalisation condition $\langle \Omega \cPT = 1$ we hence get
\begin{equation}
    1 \le \|\cPT\|_\Gamma = \sqrt{\sum_{n=0}^\infty \sum_{\bm{\nu}} \|\Upsilon_{\bm{\nu}}^{(n)}\|_2^2 } \le \sqrt{\sum_{n=0}^\infty 16 (d_{\rm S}^4 - 1)^n \frac{T^n}{n!}} = 4 \e^{T(d_{\rm S}^4 - 1)/2}.
\end{equation}

\section{Example of a cPT not exactly representable as a PCR}
In the main text, we briefly described a motivating example of a cPT $\cPT$ satisfying all the criteria of physicality, but for which no $\ket{\Phi}\in\mathsf{PCR}$ could be exactly $L^2$-equivalent. The purpose of this example is twofold: first, we wish to emphasise that the approximation error in our main theorem is necessary; it is not through lack of imagination that a PCR cannot be found for every cPT, the former is indeed strictly contained within the latter. Second, this example is geared towards the essential distinction between these two sets, which is (non-)continuity of the wavefunction amplitudes. To this effect, we exploit properties of Brownian motion, whose sign-transformed increments exhibit the type of discontinuity that we are after. Although the example is somewhat contrived, this is a matter of exposition and the stochastic dynamics from which it stems appear as a valid mathematical model in many physically relevant settings~\cite{benedetti2013dynamics, udupa2026performance}.

To begin, let $B(t)$ be the continuous-time self-similar stochastic process Brownian motion. This has zero expectation for all $t\in[0,T]$ and the covariance function over paths $\omega$
\begin{equation}\label{eq:fbm-cov}
    \operatorname{cov}[B(t), B(s)] = \frac12(|t|+ |s| - |t-s|).
\end{equation}
This process has stationary increments, which is to say that the distribution of $B(t+s) - B(s)$ depends only on $t$.

We define the cPT $\cPT$ as follows. On the interval $[0,T]$ pick a fixed $\tau\in (0,T)$. Draw a path $\omega\sim B$ and let $B(t,\omega)$ be its value at some time $t$. Next, let $\sigma_t(\omega)$ be an indicator function for each $t$, relative to $\tau$. We define it as
\begin{equation}
    \sigma_t(\omega) = \begin{cases}
        1 & B(t,\omega) > B(\tau,\omega),\\
        0 & \rm{otherwise}.
    \end{cases}
\end{equation}
From this indicator we define the process as follows: for every $t$, check if $\sigma_t(\omega)$ has changed from its initial value $\sigma_0(\omega)$; if yes, apply an $X$ gate to the system, and if no apply the identity. So the path-dependent propagator $\mathcal{U}_\omega(t,s) = U_\omega(t,s) \otimes U^\ast _\omega(t,s)$ with $U_\omega(t,s) = X^{\sigma_t(\omega) \oplus \sigma_s(\omega)}$. That is, conjugate the state by $X$ if and only if the signal changes from $s$ to $t$. Hence, every finite process tensor element is defined by this propagator averaged over paths:
\begin{equation}\label{eq:fbm-wavefunctions}
    \Upsilon_{\bm \nu}^{(k)}(\bm{t}) = \mathbb{E}_\omega \Tr[\mathcal{U}_{\omega}(T,t_k)\circ \mathbb{P}_{\nu_k}\circ \mathcal{U}_{\omega}(t_k,t_{k-1})\circ \mathbb{P}_{\nu_{k-1}}\circ \cdots \circ \mathbb{P}_{\nu_1}\circ \mathcal{U}_{\omega}(t_1,0)[\rho_0]].
\end{equation}
For simplicity, we let $\rho_0=\id/2$ so that the initial propagator can be safely ignored.

Now, to show separation from $\mathsf{PCR}$, it suffices to show some discontinuity in the wavefunctions which can be made visible across sets of positive measure. In the simplest instance, this can be seen in the two-step process tensors. Concretely, for $\Upsilon_{\nu_1,\nu_2}^{(2)}(t_1,t_2)$, let $\mathbb{P}_{\nu_1}=\mathbb{P}_{\nu_2} = \frac12(Z\otimes \id)$. Then~\cref{eq:fbm-wavefunctions} reduces to 
\begin{equation}\label{eq:fbm-2-step}
\begin{split}
    \Upsilon_{(3,3)}^{(2)}(t_1,t_2) &= \frac18 \mathbb{E}_\omega\Tr[ZX^{\sigma_{t_2}(\omega)\oplus \sigma_{t_1}(\omega)}ZX^{\sigma_{t_2}(\omega)\oplus \sigma_{t_1}(\omega)}]
    = \frac14\mathbb{E}_\omega\left[(-1)^{\sigma_{t_1}(\omega)\oplus\sigma_{t_2}(\omega)}\right].
    \end{split}
\end{equation}
In evaluating this expression, we shift coordinates to be relative to $\tau$, introducing $u = t_1-\tau$ and $v=t_2-\tau$. Since $\sigma_t(\omega)$ is just related to the sign of the increments, we define the random variable $\Delta_u := B(u+\tau)-B(\tau)$ and observe that $(-1)^{\sigma_{t_1}\oplus \sigma_{t_2}} = \rm{sgn}(\Delta_u)\rm{sgn}(\Delta_v)$. Then~\cref{eq:fbm-2-step} becomes
\begin{equation}\label{eq:fbm-2-step-sign}
    \Upsilon_{(3,3)}^{(2)}(t_1,t_2) = \frac14 \mathbb{E}_\omega[\rm{sgn}(\Delta_u)\rm{sgn}(\Delta_v)].
\end{equation}
Both $\Delta_u$ and $\Delta_v$ have zero mean and are jointly Gaussian with covariance given by~\cref{eq:fbm-cov}. Thus, their correlation coefficient $\varrho(u,v):=\rm{cov}[\Delta_u,\Delta_v]/(\sigma_{\Delta_u}\sigma_{\Delta_v})$ is equal to $\frac12(|u| + |v| - |u-v|)/(\sqrt{|u||v|})$. Applying a standard Gaussian identity to the above, it then follows that for $u,v\neq 0$ we have
\begin{equation}
    \Upsilon_{(3,3)}^{(2)}(t_1,t_2) = \frac{1}{2\pi}\arcsin(\varrho(u,v)) = \frac{1}{2\pi}\arcsin(\varrho(t_1-\tau,t_2-\tau)).
\end{equation}
The scale invariance of Brownian motion increments means that $\varrho(\lambda u,\lambda v)= \varrho(u,v)$ for all $\lambda > 0$. Thus, for any fixed $(u_0,v_0)$, the ray $\{(\lambda u_0,\lambda v_0)\mid\lambda>0\}$ has constant correlation $\varrho(u_0,v_0)$ terminating at $(u,v)=(0,0)$ or, equivalently, $(t_1,t_2)=(\tau,\tau)$. For example, $u_0=-1$ and $v_0=\pm\tfrac12$ give two constant-valued rays whose respective two-step process-tensor coefficients differ, namely
\begin{equation}
    a_- \coloneqq\lim_{\lambda\to 0} \Upsilon_{(3,3)}^{(2)}\left(\tau - \lambda,\tau - \frac{\lambda}{2}\right) \neq \lim_{\lambda\to 0} \Upsilon_{(3,3)}^{(2)}\left(\tau - \lambda, \tau + \frac{\lambda}{2}\right)\coloneqq a_+,\qquad a_-=\tfrac18,\quad a_+=0.
\end{equation}
To see that there can be no PCR representative in the $L^2$ class of $\cPT$, one can grow these constant-valued rays to narrow cones which have positive two-dimensional Lebesgue measure. Consider $C_+$ and $C_-$ to be sufficiently narrow cones around the above rays. For any sufficiently small $\epsilon>0$, the cones may be chosen such that 
\begin{equation}
    |\Upsilon^{(2)}_{(3,3)}(\tau+u,\tau+v) - a_\pm| < \epsilon~\forall~(u,v)\in C_\pm,\quad \text{with}~~|a_+-a_-| > 4\epsilon.
\end{equation}
Next, suppose that there were a function $g(u,v)$, continuous up to $(0,0)$ on the closed ordered domain, which agreed with $\Upsilon_{(3,3)}^{(2)}(\tau+u,\tau+v)$ for almost every $(u,v)$. By continuity of $g$ at the origin, there exists $r>0$ such that $|g(u,v) - g(0,0)| < \epsilon$ everywhere in $B_r(0)$ -- the ball of radius $r$ centred around $(0,0)$. Since $g(u,v) = \Upsilon_{(3,3)}^{(2)}(\tau+u,\tau+v)$ almost everywhere, and since each $C_\pm\cap B_r(0)$ has positive measure, there exist points in both cones at which the two functions agree. It follows that $|g(0,0) - a_\pm| < 2\epsilon$. By the triangle inequality, this implies however that $|a_+ - a_-| < 4\epsilon$, contradicting our initial choice of $\epsilon$. We conclude therefore that $\cPT$ admits no representative in its $L^2$ equivalence class that extends continuously to $(t_1,t_2)=(\tau,\tau)$, and hence that $\cPT \not\in \mathsf{PCR}$. 

\section{Examples of continuous instruments and their norms}\label{sec:instruments}

In the companion work \cite{wassnerOperationalContinuumLimit2026}, several examples of continuous instruments were presented, in particular for the scenarios of coherent control, weak measurements in the jump and diffusive regime and controlled coherent operations.
The effects of these instruments  were derived by computing the overlaps with PCRs, which are easy to interpret as quantum dynamics.
Density of the set of PCRs in the set of cPTs, the main result of this work, implies that this consideration is enough to guarantee the correct action of these effects on all cPTs, provided the corresponding instrument norm is finite.
In this appendix we formalise these instruments in terms of Definition 4 from the main text: we specify the associated measure space $\mathcal X = (X, \mathcal{F}, \mu)$ and compute an upper bound on the associated instrument norm.
We will see that the norm is larger when the instrument requires more energy to implement. 
The diffusive measurement is a special case: as the idealised limit which requires an infinite-energy local oscillator, it has an infinite instrument norm. However, it is a well-defined map from the cPTs to $L^1$ functions on the associated measure space. 

\subsection{Coherent control}

The coherent control instrument $\mathcal{I}_\text{coherent}$ has a single effect $\bra{\Phi_T[c_0, \{c_\nu\}, \kket{\rho_{\rm A}}\!\bbra{\id_{\rm A}}]}$, where $\rm A$ is an auxiliary system, $\rho_{\rm A}$ is a quantum state and $c_\nu = \tr_{\rm S}[\mathbb{P}_\nu \mathcal{L}_{\rm SA}]^*$, where $\mathcal{L}_{\rm SA}$ is the generator of the control dynamics implemented on $\rm SA$.
This is an example of a deterministic instrument with no associated outcome.
Hence, it is associated to the trivial measure space $\mathcal{X}$ with $X = \{0\}$, $\mathcal{F} = \{\emptyset,X\}$ and $\mu(\emptyset) = 0$ while $\mu(X) = 1$.
As discussed in the companion work~\cite{wassnerOperationalContinuumLimit2026}, this instrument satisfies $\braket{\mathcal{I}_\text{coherent}|\Upsilon} = 1$ for all $\ket{\Upsilon} \in \mathsf{PCR}$, and since PCRs are dense in the cPTs, this is in fact the case for all $\ket{\Upsilon} \in \mathsf{cPT}$, showing that the coherent control satisfies the definition of an instrument in the main text.
We will now upper bound its instrument norm.
Consider
\begin{align}
    \|\mathcal{I}_\text{coherent}\|_{\Gamma\rightarrow L^1(\mathcal{X})}^2 
    &= \|\ket{\Phi_T[c_0, \{c_\nu\}, \kket{\rho_{\rm A}}\!\bbra{\id_{\rm A}}]}\|_\Gamma^2 \nonumber \\
    &= \bbra{\id_{\rm A}} \otimes \bbra{\id_{\rm A'}} \mathcal{T} \e^{\int_0^T {\rm d}t\, c_0 \otimes \id + \id \otimes c_0^* + \sum_\nu c_\nu \otimes c_\nu^*} \kket{\rho_{\rm A}} \otimes \kket{\rho_{\rm A'}} \nonumber \\
    &\le d_{\rm A} \e^{T \sup_{t\in[0, T]} \|c_0 \otimes \id + \id \otimes c_0^* + \sum_\nu c_\nu \otimes c_\nu^*\|_\infty},
\end{align}
where we label the backward contour using a prime $'$.
Here, we have used
\begin{equation}
    \|\mathcal{T} e^{\int_0^T A(t) {\rm d}t}\|_\infty = \|\sum_{n=0}^\infty \int {\rm D}^n\bm{t} A(t_n)\dots A(t_1)\|_\infty \le \sum_{n=0}^\infty \frac{T^n}{n!} \left(\sup_{t\in[0,T]}\|A(t)\|_\infty\right)^n = \e^{T \sup_{t \in [0,T]} \|A(t)\|_\infty}
\end{equation}
and the cMPS overlap formula~\cite{haegemanCalculusContinuousMatrix2013}.
Now notice
\begin{equation}\label{eq:sum of paulis}
    \sum_{\nu} \mathbb{P}_\nu \otimes \mathbb{P}_\nu^* = \left(\ket{\id_{\rm SS'}}\!\bra{\id_{\rm SS'}} - \frac1{d_{\rm S}^2} \id_{\rm S} \otimes \id_{\rm S'}\right) \otimes \id_{\rm A} \otimes \id_{\rm A'},
\end{equation}
which implies
\begin{align}
    \left\|\sum_{\nu} c_{\nu} \otimes c_{\nu}^* \right\|_\infty &= \left\|\bra{\id_{\rm SS'}}\mathcal{L}_{\rm SA} \otimes \mathcal{L}_{\rm S'A'} \ket{\id_{\rm SS'}} - \frac1{d_{\rm S}^2} c_0\otimes c_0^*\right\|_\infty \nonumber\\
                                                                &\le (d_{\rm S} + 1) \|\mathcal{L}_{\rm SA}\|_\infty^2.
\end{align}
Finally, we get
\begin{equation}
    \sup_{t\in[0, T]} \|c_0 \otimes \id + \id \otimes c_0^* + \sum_\nu c_\nu \otimes c_\nu^*\|_\infty \le 2 \|\mathcal{L}_{\rm SA}\|_\infty + (d_{\rm S} + 1) \|\mathcal{L}_{\rm SA}\|_\infty^2
\end{equation}
and hence
\begin{equation}
    \|\mathcal{I}_{\text{coherent}}\|_{\Gamma \rightarrow L^1(\mathcal{X})} \le \sqrt{d_{\rm A}} \e^{T\left((d_{\rm S} + 1)\|\mathcal{L}_{\rm SA}\|_\infty^2 + 2\|\mathcal{L}_{\rm SA}\|_\infty \right)},
\end{equation}
where $\|\mathcal{L}_{\rm SA}\|_\infty \coloneqq \sup_{t\in[0,T]} \|\mathcal{L}_{\rm SA}(t)\|_\infty$.
Hence, the upper bound on the instrument norm is determined by the operator norm of the vectorised Lindbladian, which depends on the energy needed to implement it.

It may seem that the units in the exponent do not align.
This is due to the fact that we are working in natural units $\hbar=1$. In fact, $\frac\hbar{c_0}$ and $\frac{\hbar}{c_\nu^2}$ both have units of time.

\subsection{Jump measurement}

The outcomes of jump measurements are sequences of times when a particular outcome, a ``click'', has been observed.
A prototypical example is a photon-counting experiment, where the click is the arrival of a photon.
For an introduction to jump measurements see~\cite{wisemanQuantumMeasurementControl2014,guilminParametersEstimationFitting2024a,rosal2025deterministicequationsfeedbackcontrol,gardiner2004quantum}.
Often, a click is associated with a type (e.g. polarisation of the photon) from some type alphabet $\Omega$, such that the outcome of the measurement is a sequence of times $t_i$ along with a list of types $x_i\in\Omega$ of clicks, such that a click of type $x_i$ has been observed at the time $t_i$.

To define the measure space for the jump measurement, consider for each non-negative integer $n$ the space
\begin{equation}
    X_n \coloneqq \Delta_n \times \Omega^n,
\end{equation}
and let $X \coloneqq \bigsqcup_{n=0}^{\infty} X_n$.
The $\sigma$-algebra $\mathcal F_n$ on $X_n$ is the product of the Borel $\sigma$-algebra on $\Delta_n$ and the discrete $\sigma$-algebra on $\Omega^n$. The $\sigma$-algebra on $X$ is $\mathcal F\coloneqq\{E\subseteq X:E\cap X_n\in\mathcal F_n\text{ for every }n\}$.
The reference measure is the Poisson measure, induced by the Lebesgue measure on $\Delta_n$.

Each jump type $x \in \Omega$ is associated with a jump operator $L_x$ acting on $\mathcal H_{\rm S}$.
The collection $\{L_x\}_{x\in\Omega}$ defines the measurement.
The effects of the associated instrument are the sums over $n$ time-ordered Lebesgue integrals of
\begin{equation}
    \sum_{\bm{\nu}\in(\{0\}\cup[q])^n} \bra{\Phi_T[c_0, \{c_\nu\}, 1]} \prod_{j\in [n]} d_{\nu_j}^{x_j} \xi_{\nu_j}(t_j),
\end{equation}
where $c_\nu \coloneqq -\frac12 \tr\left[\mathbb{P}_\nu \left(\sum_{x\in\Omega} L_x^\dag L_x \otimes \id + \id \otimes L_x^T L_x^*\right)\right]^\ast$ and $d_\nu^x \coloneqq \tr[\mathbb{P}_\nu L_x \otimes L_x^*]^\ast$.
We can compute the associated deterministic instrument
\begin{align}
    \bra{\mathcal{I}_\text{jump}(X)} 
    &= \sum_{n=0}^\infty \int_0^T {\rm D}^n \bm{t} \sum_{\bm{x}\in\Omega^n} \sum_{\bm{\nu}\in(\{0\}\cup[q])^n} \bra{\Phi_T[c_0, \{c_\nu\}, 1]} \prod_{j\in [n]} d_{\nu_j}^{x_j} \xi_{\nu_j}(t_j) \nonumber\\
    &= \sum_{n=0}^\infty \int_0^T {\rm D}^n \bm{t} \ \e^{T\sum_{x\in\Omega}d_0^x} \sum_{\bm{x} \in \Omega^n}\sum_{\bm{\nu}\in[q]^n} \bra{\Phi_T[c_0, \{c_\nu\}, 1]} \prod_{j\in[n]} d_{\nu_j}^{x_j} \psi_{\nu_j}(t_j) \nonumber\\
    &= \Bra{\Phi_T\left[c_0 + \sum_{x\in\Omega} d_0^x, \{c_\nu + \sum_{x\in\Omega} d_\nu^x\}, 1\right]} \label{eq:deter_instr_jump},
\end{align}
as can be verified by computing the wavefunctions of the resulting cMPS~\cite{haegemanCalculusContinuousMatrix2013}.
We recognise this cMPS from the previous section as the effect of the coherent control with no auxiliary system, which is implementing the Lindbladian
\begin{equation}
    \mathcal{L}_{\rm S} = \sum_{x \in \Omega} \Bigl[L_x \otimes L_x^* - \frac12\left(L_x^\dag L_x \otimes \id_{\rm S'} + \id_{\rm S}\otimes L_x^TL_x^*\right)\Bigr],
\end{equation}
which is the usual backaction associated to this jump measurement.
We have seen in the previous section that coherent control is an example of a deterministic instrument, having unit overlap with all cPTs.
This verifies that the jump instrument as defined above is correctly normalised.

We will now upper bound the instrument norm for the jump process.
Consider
\begin{align}
    \|\mathcal{I}_{\text{jump}} \|_{\Gamma\rightarrow L^1(\mathcal{X})} &= \sup_{\substack{\ket{\Upsilon} \in \Gamma:\\\|\ket{\Upsilon} \|_\Gamma \leq 1}}
    \sum_{n=0}^\infty \int_0^T {\rm D}^n \bm{t} \sum_{\bm{x} \in \Omega^n} 
    \left| \sum_{\bm{\nu}\in(\{0\}\cup[q])^n} \braket{\Phi_T[c_0, \{c_\nu\}, 1] | \prod_{j \in[n]} d^{x_j}_{\nu_j} \xi_{\nu_j}(t_j) | \Upsilon} \right| \nonumber\\
    &= \sup_{f: \|f\|_\infty = 1}\sup_{\substack{\ket{\Upsilon}\in\Gamma: \\\|\ket{\Upsilon}\|_\Gamma = 1} }\Bigl|
    \sum_{n=0}^\infty \int_0^T {\rm D}^n \bm{t} \sum_{\bm{x} \in \Omega^n} 
    f^{(n)}_{\bm{x}}(\bm{t})\sum_{\bm{\nu}\in(\{0\}\cup[q])^n} \braket{\Phi_T[c_0, \{c_\nu\}, 1] | \prod_{j \in[n]} d^{x_j}_{\nu_j} \xi_{\nu_j}(t_j) | \Upsilon}\Bigr| \nonumber\\
    &= \sup_{f: \|f\|_\infty = 1}\left\| \sum_{n=0}^\infty \int_0^T {\rm D}^n \bm{t} \sum_{\bm{x} \in \Omega^n} 
    f^{(n)}_{\bm{x}}(\bm{t})\sum_{\bm{\nu}\in(\{0\}\cup[q])^n} 
    \bra{\Phi_T[c_0, \{c_\nu\}, 1]} \prod_{j \in[n]} d^{x_j}_{\nu_j} \xi_{\nu_j}(t_j) \right\|_\Gamma \nonumber\\
    &\le \sup_{f: \|f\|_\infty = 1} \sum_{n=0}^\infty \sum_{m=0}^\infty
    \left\| \int_0^T {\rm D}^n\bm{t} \sum_{\bm{\nu}\in(\{0\}\cup[q])^n} \sum_{\bm{x}\in\Omega^n} \left(f_{\bm{x}}^{(n)}(\bm{t})\right)^*
    \prod_{j\in[n]} \left(d_{\nu_j}^{x_j}\right)^* \xi_{\nu_j}^\dag(t_j) \ket{\Phi_m}\right\|_\Gamma,
\end{align}
where $\ket{\Phi_m}$ is the $m$-particle component of $\ket{\Phi[c_0, \{c_\nu\}, 1]}$.
Direct computation shows that 
\begin{equation}
    \|\ket{\Phi_m}\|_\Gamma = \sqrt{\frac{T^m}{m!}} \e^{T\Re(c_0)} \|\bm c_\perp\|_2^m,
\end{equation}
where $\|\bm c_\perp \|_2^2 \coloneqq \sum_{\nu\in[q]} |c_\nu|^2$.
We now define
\begin{equation}
    \ket{\Phi_{n,m}} \coloneqq 
    \int_0^T {\rm D}^n\bm{t} \sum_{\bm{\nu}\in(\{0\}\cup[q])^n} \sum_{\bm{x}\in\Omega^n} \left(f_{\bm{x}}^{(n)}(\bm{t})\right)^*
    \prod_{j\in[n]} \left(d_{\nu_j}^{x_j}\right)^* \xi_{\nu_j}^\dag(t_j) \ket{\Phi_m}.
\end{equation}
The vector $\ket{\Phi_{n,m}}$ has support only in sectors with $m+k$ particles for $0\le k \le n$.
Let $\ket{\Phi_{n,m}^k}$ be the $m+k$-particle component of $\ket{\Phi_{n,m}}$ and write
\begin{equation}
    \|\mathcal{I}_\text{jump}\|_{\Gamma \rightarrow L^1(\mathcal{X})} \le 
    \sup_{f:\|f\|_\infty = 1} \sum_{n=0}^\infty \sum_{m=0}^\infty \|\ket{\Phi_{n,m}}\|_\Gamma \le
    \sup_{f:\|f\|_\infty = 1} \sum_{n=0}^\infty \sum_{m=0}^\infty \sum_{k=0}^n \|\ket{\Phi_{n,m}^k}\|_\Gamma.
\end{equation}
Note that the vectors $\ket{\Phi_{n,m}}$ depend on the function $f$.
We will now upper bound the norm of $\ket{\Phi_{n,m}^k}$.

Without loss of generality we can restrict each $f^{(n)}_{\bm{x}}(\bm{t})$ to be symmetric under simultaneous permutations of the time--outcome pairs $(t_i,x_i)$.
Hence, labelling $\Phi^{(m)}$ the wavefunction of $\ket{\Phi_{m}}$, we have
\begin{align}
    \ket{\Phi_{n,m}^k} 
    &= \Pi_{k+m} \int_0^T {\rm D}^n\bm{t} \sum_{\bm{\nu}\in(\{0\}\cup[q])^n} \sum_{\bm{x}\in\Omega^n} \left(f_{\bm{x}}^{(n)}(\bm{t})\right)^*
    \prod_{j\in[n]} \left(d_{\nu_j}^{x_j}\right)^* \xi_{\nu_j}^\dag(t_j)\ket{\Phi_{m}} \nonumber\\
    &= \Pi_{k+m} \frac1{n!} \int_0^T {\rm d}^n\bm{t} \sum_{\bm{\nu} \in (\{0\}\cup [q])^n} \sum_{\bm{x}\in\Omega^n} \left(f_{\bm{x}}^{(n)}(\bm{t})\right)^*
    \prod_{j\in[n]} \left(d_{\nu_j}^{x_j}\right)^* \xi_{\nu_j}^\dag(t_j)\ket{\Phi_{m}} \nonumber\\
    &= \frac1{n!} \binom{n}{k} \int_0^T {\rm d}^{k}\bm{t} \sum_{\bm{\nu}\in[q]^{k}} \sum_{\bm{x}\in\Omega^{k}}
    (n-k)! g_{\bm{x}}^{(n,k)}(\bm{t}) \prod_{j\in[k]} \left(d_{\nu_j}^{x_j}\right)^* \psi_{\nu_j}^\dag(t_j)\ket{\Phi_{m}} \nonumber\\
    &= \int_0^T {\rm D}^k \bm{t} \sum_{\bm{\nu}\in[q]^{k}} \sum_{\bm{x}\in\Omega^{k}}
    g_{\bm{x}}^{(n,k)}(\bm{t}) \prod_{j\in[k]} \left(d_{\nu_j}^{x_j}\right)^* \psi_{\nu_j}^\dag(t_j)
    \int_0^T {\rm D}^m \bm{t'} \sum_{\bm{\nu'}\in[q]^m} \Phi^{(m)}_{\bm{\nu'}}(\bm{t'}) \prod_{i\in[m]} \psi^\dag_{\nu'_i}(t'_i)\ket{\Omega} \nonumber\\
    &= \frac1{(k+m)!} \int_0^T {\rm d}^{k+m} \bm{t} \sum_{\substack{\bm{s}\in\{0,1\}^{k+m} \\ \|\bm{s}\| = k}} \sum_{\bm{\nu}\in[q]^{k+m}} \sum_{\bm{x}\in\Omega^{k}}
    g_{\bm{x}}^{(n,k)}(\bm{t}_{\bm{s}=1}) \Phi^{(m)}_{\bm{\nu}_{\bm{s}=0}}(\bm{t}_{\bm{s}=0}) 
    \prod_{j\in[k]} \left(d_{(\nu_{\bm{s}=1})_j}^{x_j}\right)^* \prod_{i\in[k+m]} \psi_{\nu_i}^\dag(t_i)\ket{\Omega} \nonumber\\
    &= \int_0^T {\rm D}^{k+m} \bm{t} \sum_{\bm{\nu}\in[q]^{k+m}} \sum_{\substack{\bm{s}\in\{0,1\}^{k+m} \\ \|\bm{s}\| = k}} \sum_{\bm{x}\in\Omega^{k}}
    g_{\bm{x}}^{(n,k)}(\bm{t}_{\bm{s}=1}) \Phi^{(m)}_{\bm{\nu}_{\bm{s}=0}}(\bm{t}_{\bm{s}=0}) 
    \prod_{j\in[k]} \left(d_{(\nu_{\bm{s}=1})_j}^{x_j}\right)^* \prod_{i\in[k+m]} \psi_{\nu_i}^\dag(t_i)\ket{\Omega}
\end{align}
where $\Pi_{k+m}$ is the projector to the $k+m$-particle subspace, indexing a vector with $\bm{s}=0, 1$ throws away all the elements $j$ where $s_j=1,0$ respectively, and for $\bm{t}\in[0,T]^k$ and $\bm{x}\in\Omega^k$ we define the symmetric function
\begin{equation}
    g^{(n,k)}_{\bm{x}}(\bm{t}) \coloneqq \frac1{(n-k)!}\sum_{\bm{x'} \in \Omega^{n-k}} \int_0^T {\rm d}^{n-k}\bm{t'} \left(f^{(n)}_{\bm{x}, \bm{x'}}(\bm{t}, \bm{t'})\right)^* \prod_{j\in[n-k]} \left(d_0^{x'_j}\right)^*.
\end{equation}
For $\|f\|_\infty \le 1$ we have $\|g^{(n,k)}_{\bm{x}}\|_\infty \le \frac{T^{n-k}}{(n-k)!} |\sum_{x\in\Omega} d_0^x|^{n-k}$.
The norm of $\ket{\Phi^k_{n,m}}$ is hence upper bounded by
\begin{align}
    \|\ket{\Phi^k_{n,m}}\|_\Gamma 
    &= \sqrt{\sum_{\bm{\nu}\in[q]^{k+m}} \int_0^T {\rm D}^{k+m}\bm{t} \ \Big|\sum_{\substack{\bm{s}\in\{0,1\}^{k+m} \\ \|\bm{s}\| = k}} 
    \sum_{\bm{x}\in\Omega^{k}} g_{\bm{x}}^{(n,k)}(\bm{t}_{\bm{s}=1}) \Phi^{(m)}_{\bm{\nu}_{\bm{s}=0}}(\bm{t}_{\bm{s}=0}) 
    \prod_{j\in[k]} \left(d_{(\nu_{\bm{s}=1})_j}^{x_j}\right)^*\Big|^2} \nonumber\\
    &= \sqrt{\frac1{(k+m)!} \sum_{\bm{\nu}\in[q]^{k+m}} \int_0^T {\rm d}^{k+m}\bm{t} \ \Big|\sum_{\substack{\bm{s}\in\{0,1\}^{k+m} \\ \|\bm{s}\| = k}} 
    \sum_{\bm{x}\in\Omega^{k}} g_{\bm{x}}^{(n,k)}(\bm{t}_{\bm{s}=1}) \Phi^{(m)}_{\bm{\nu}_{\bm{s}=0}}(\bm{t}_{\bm{s}=0}) 
    \prod_{j\in[k]} \left(d_{(\nu_{\bm{s}=1})_j}^{x_j}\right)^*\Big|^2
    } \nonumber\\
    &\le \binom{k+m}{k} \sqrt{ \frac1{(k+m)!} \sum_{\bm{\nu}\in[q]^{k}}\sum_{\bm{\nu'}\in[q]^{m}} \int_0^T 
    {\rm d}^{k}\bm{t} \ {\rm d}^m\bm{t'} \ \left|\sum_{\bm{x}\in\Omega^k} g_{\bm{x}}^{(n,k)}(\bm{t}) 
    \prod_{j\in[k]} \left(d_{\nu_j}^{x_j}\right)^* \Phi^{(m)}_{\bm{\nu'}}(\bm{t'})\right|^2} \nonumber\\
    &\le \binom{k+m}{k} \frac{T^{n-k}}{(n-k)!} \left|\sum_{x\in\Omega}d_0^x\right|^{n-k} \bigl(\sum_{x\in\Omega}\|\bm d_\perp^x\|_2\bigr)^k
    \sqrt{ \frac{m!}{(k+m)!} T^k \sum_{\bm{\nu'}\in[q]^m} \int_0^T {\rm D}^m\bm{t'} \ |\Phi^{(m)}_{\bm{\nu'}}(\bm{t'})|^2} \nonumber\\
    &= \binom{k+m}{k} \frac{T^{n-\frac{k}2}}{(n-k)!} \sqrt{\frac{m!}{(k+m)!}} 
    \left|\sum_{x\in\Omega}d_0^x\right|^{n-k} \bigl(\sum_{x\in\Omega}\|\bm d_\perp^x\|_2\bigr)^k\|\ket{\Phi_m}\|_\Gamma \nonumber\\
    &\le \e^{T\Re(c_0)} \frac{T^{n+\frac{m-k}2}\sqrt{(k+m)!}}{k!m!(n-k)!} 
    \left|\sum_{x\in\Omega}d_0^x\right|^{n-k} \bigl(\sum_{x\in\Omega}\|\bm d_\perp^x\|_2\bigr)^k \|\bm c_\perp\|_2^m.
\end{align}
Finally, we obtain a bound on the instrument norm
\begin{align}
    \|\mathcal{I}_\text{jump}\|_{\Gamma \rightarrow L^1(\mathcal{X})} 
    &\le \sup_{f:\|f\|_\infty = 1} \sum_{n=0}^\infty \sum_{m=0}^\infty \sum_{k=0}^n \|\ket{\Phi_{n,m}^k}\|_\Gamma \nonumber\\
    &\le \sum_{n=0}^\infty \sum_{m=0}^\infty \sum_{k=0}^n \e^{T\Re(c_0)} \frac{T^{n+\frac{m-k}2}\sqrt{(k+m)!}}{k!m!(n-k)!} 
        \left|\sum_{x\in\Omega}d_0^x\right|^{n-k} \bigl(\sum_{x\in\Omega}\|\bm d_\perp^x\|_2\bigr)^k \|\bm c_\perp\|_2^m \nonumber\\
    &\le \e^{T\Re(c_0)} \sum_{n=0}^\infty \sum_{m=0}^\infty \frac{T^{n+\frac{m}2}\sqrt{(n+m)!}}{n!m!} \|\bm c_\perp\|_2^m 
        \left(\left|\sum_{x\in\Omega}d_0^x\right| + \frac{\bigl(\sum_{x\in\Omega}\|\bm d_\perp^x\|_2\bigr)}{\sqrt{T}}\right)^n \nonumber\\
    &= \e^{T\Re(c_0)} \sum_{p=0}^\infty \sum_{q=0}^p \frac{T^{\frac{p+q}2}}{\sqrt{p!}} \binom{p}{q} \|\bm c_\perp\|_2^{p-q}
        \left(\left|\sum_{x\in\Omega}d_0^x\right| + \frac{\bigl(\sum_{x\in\Omega}\|\bm d_\perp^x\|_2\bigr)}{\sqrt{T}}\right)^q \nonumber\\
    &= \e^{T\Re(c_0)} \sum_{p=0}^\infty \frac{T^{\frac{p}2}}{\sqrt{p!}} 
        \left(\|\bm c_\perp\|_2 + \bigl(\sum_{x\in\Omega}\|\bm d_\perp^x\|_2\bigr) + \sqrt{T}\left|\sum_{x\in\Omega}d_0^x\right|\right)^p \nonumber\\
    &\le \sqrt{2} \exp\left[T\left(\Re(c_0) + \left(\|\bm c_\perp\|_2 + \bigl(\sum_{x\in\Omega}\|\bm d_\perp^x\|_2\bigr) + \sqrt{T}\left|\sum_{x\in\Omega}d_0^x\right|\right)^2\right)\right],
\end{align}
where in the final line we have used the Cauchy--Schwarz inequality: for any $a\ge0$
\begin{equation}
    \sum_{p=0}^\infty \frac{a^p}{\sqrt{p!}} 
    \le \sqrt{\sum_{p=0}^\infty \frac1{2^p}} \sqrt{\sum_{q=0}^\infty \frac{(2a^2)^q}{q!}} = \sqrt{2} \e^{a^2}.
\end{equation}
We can express the coefficients in terms of the jump operators using~\cref{eq:sum of paulis} to obtain
\begin{align}
\label{eq:c_d_intermsof_Loperators}
    c_0 &= - \tr \sum_{x\in\Omega} L_x^\dag L_x\\
    \|\bm c_\perp\|_2^2 &= \frac12 \left(d_{\rm S}\tr \left[\sum_{x\in\Omega} L_x^\dag L_x \sum_{x'\in\Omega} L_{x'}^\dag L_{x'}\right] - \left| \tr\sum_{x\in\Omega} L_x^\dag L_x\right|^2 \right)\\
    \left|\sum_{x\in\Omega} d_0^x\right| &= \frac1{d_{\rm S}} \sum_{x\in\Omega} |\tr L_x|^2\\
    \sum_{x\in\Omega} \|\bm d_\perp^x\|_2^2 &= \sum_{x\in\Omega} \|L_x\|_F^4 - \frac1{d_{\rm S}^2} |\tr L_x|^4
\end{align}

\subsection{Diffusive measurement}

Diffusive measurements output continuous functions $C([0, T])$, which can be real or complex. 
An example of this measurement is the measurement of photocurrents in homodyne and heterodyne detection \cite{barchielliQuantumTrajectoriesMeasurements2009,wisemanQuantumMeasurementControl2014}. 
Just like the jump measurements, they are associated with measurement operators $\{L_x\}_{x\in\Omega}$ on $\mathcal H_{\rm S}$ for a finite alphabet $\Omega$. 
As in the companion work \cite{wassnerOperationalContinuumLimit2026} we will consider here only the case of independent heterodyne detection in each mode corresponding to complex-valued outcome functions with independent complex Wiener reference noises.  
Hence, we have $X = C([0, T],\mathbb C)^{|\Omega|}$ and the $\sigma$-algebra is generated by all the cylinder sets
\begin{equation}
    \left\{c\in X:(c(t_1),\ldots,c(t_n))\in A\right\},
\end{equation}
where $n\ge1$, $0\le t_1<\cdots<t_n\le T$, and $A$ is a Borel subset of $(\mathbb C^{|\Omega|})^n$.
The measure $\mu_W$ is the product measure of $|\Omega|=:k$ complex Wiener measures. 
We denote the derivative in a distributional sense of the outcome functions $(Y_x )_x  \in X$ by $\vec{J} = (J_x )_x $. Under the complex Wiener measure $\mu$ the \textit{white noises} $J_x ,J_x ^\ast$, $x \in [k]$ are centred, jointly Gaussian, and obey 
\begin{equation}
    \label{eq:white_noise_prop}
    \mathbb E[J_x (t)J_y^\ast(s)]= \delta_{xy}\delta(t-s), \qquad  \mathbb E[J_x^\ast(t)J_y^\ast(s)] =  \mathbb E[J_x (t)J_y(s)]= 0.
\end{equation}

The diffusive measurement is an idealised limit of a physical measurement protocol, where $\rm S$ interacts continuously with an optical mode, which is then mixed with a local oscillator and heterodyne measurement is performed.
The diffusive limit is obtained by letting the strength of the local oscillator diverge.
Hence, strictly speaking, the diffusive measurement requires infinite energy and cannot be actually performed in the lab, despite being a useful approximation in many scenarios.
We will see in this section that while our framework is well capable of accommodating this instrument, its instrument norm diverges.
In other words, the infinite energy required to perform it gives this instrument unphysical distinguishing power in the sense of Eq. (8) in the main text.

The associated instrument effects are built from particle-sector integrals of \cite{wassnerOperationalContinuumLimit2026} 
\begin{equation}
\label{eq:diff_instrument_density}
    \bra{\mathcal I_{\text{diffusive}}(\vec{J})} = \bra{\Phi_T[c_0 +d_0+ \chi_0^{\vec{J}},\{c_{\nu} +d_\nu+ \chi^{\vec{J}}_{\nu}\}_{\nu},1]}
\end{equation}
with
\begin{align}
     c_\nu &\coloneqq - \tfrac{1}{2} \tr\Bigl[\mathbb{P}_\nu\bigl(\sum_{x \in [k]} L_x ^\dag L_x  \otimes \id + \id \otimes L_x ^T L_x ^*\bigr)\Bigr]^\ast \nonumber\\
     d_\nu&\coloneqq \sum_{x \in [k]} d_{\nu}^x = \tr\Bigl[\mathbb{P}_\nu\bigl(\sum_{x \in [k]}  L_x  \otimes L_x ^*\bigr)\Bigr]^\ast \nonumber \\
        \chi^{\vec{J}}_\nu(t) &\coloneqq \tr\left[\mathbb{P}_\nu \sum_{x \in [k]} \bigl(J^\ast_x (t) L_x \otimes \id + \id \otimes J_x (t)L^\ast_{x}\bigr)\right]^\ast.
\end{align}
We note that $ \bra{\mathcal I_{\text{diffusive}}(\vec{J})} = \bra{\Phi_T[c_0 +d_0+ \chi_0^{\vec{J}},\{c_{\nu} +d_\nu+ \chi^{\vec{J}}_{\nu}\}_{\nu},1]}$ is of cMPS form with linear functionals of the noise as its wavefunctions, so it is not a vector in $\Gamma$. All stochastic integrals encountered are of Itô type \cite{gardinerStochasticMethodsHandbook2009,meyerQuantumProbabilityProbabilists1995}. 

It is helpful to think of Itô stochastic integrals in terms of the \textit{Wick product} \cite{jansonGaussianHilbertSpaces1997,meyerQuantumProbabilityProbabilists1995}. The Wick product $:Y_1Y_2\cdots Y_n:$ of centred, jointly Gaussian random variables  $Y_1,Y_2,\cdots ,Y_n$ (here linear functionals of the noise) is the polynomial obtained from the ordinary product by deleting all self-contractions:
\begin{equation}
    \colon Y_1Y_2\cdots Y_n \colon  = \sum_P (-1)^{|P|}\prod_{\{i,j\}\in P}\mathbb E[Y_iY_j ]\prod_{i \notin \cup P} Y_i,
\end{equation}
where the sum runs over all sets $P$ of disjoint pairs in $[n]$ (including $P = \emptyset$). For $X, Y_1, \cdots,Y_n$ centred jointly Gaussian random variables and constants $a,a_1,\cdots,a_n \in \mathbb C$, the Wick product satisfies 
\begin{align}
 \colon (a + a_1Y_1 + \cdots + a_nY_n) \colon &= a + a_1\colon Y_1 \colon +\cdots +  a_n\colon Y_n \colon ,\label{eq:linearity_wick} \\
    \colon \e^{X}\prod_{i=1}^n(a_i+Y_i) \colon &= \e^{X-\tfrac{1}{2}\mathbb E [X^2]}\colon \prod_{i=1}^n(a_i-\mathbb E[XY_i]+Y_i) \colon ,\label{eq:wick_prod_exp_formula} \\
    \label{eq:wick_prod_expect_value} 
\mathbb E [\colon Y_1 \cdots Y_n \colon] &= 0\quad(n\ge1), \; \text{ and} \\
\label{eq:wick_prod_expec_of_exp} \colon \e^{X} \colon &= \e^{X- \tfrac 12 \mathbb E[X^2]} .
\end{align}
 
Using the Wick product, the contraction of the instrument density with $\ket{\Phi}\in \Gamma$ is  expressible as the formal series
\begin{align}
\label{eq:formal_series_diff_1}
    \braket{\mathcal I_{\text{diffusive}}(\vec J)|\Phi}&= \sum_{n=0}^\infty \int_{0}^T \mathrm{D}^n\bm{t} \sum_{\bm{\nu}} \phi_{\bm{\nu}}^{(n)}(\bm{t}) \e^{T(c_0+d_0)} \colon \e^{\int_0 ^T \mathrm ds \chi_0^{\vec J}(s)^\ast}\prod_{i=1}^n(c_{\nu_i}+d_{\nu_i}+\chi_{\nu_i}^{\vec J}(t_i))^\ast \colon.
\end{align}
To obtain the corresponding deterministic instrument or analyse the instrument norm we have to consider expectation values with respect to $\mu_W$ of the above expression and its absolute value, respectively. For this it is convenient to collect all white noises in the vector $\xi =(\xi_\alpha)_{\alpha \in [2k]}$ with $\xi_x := J^\ast_x $ and $\xi_{k+x}:= J_x $ such that
\begin{equation}
    \chi_{\nu}^{\vec J}(t)^\ast = \sum_{\alpha \in [2k]} C_{\nu,\alpha}\xi_\alpha(t), \qquad C_{\nu ,x}:= \tr[\mathbb P_\nu (L_x \otimes \id)], \; C_{\nu, x+k}:= \tr[\mathbb P_\nu ( \id \otimes L^\ast_x )],
    \label{eq:matrix_C_coefficient_change}
\end{equation}
where $\nu = 0,1,\cdots, d_{\rm S}^{4}-1$ and $C \in \mathbb C^{d_{\rm S}^4 \times 2k}$.  Then, we can express the contraction of two random coefficients as
\begin{align}
    \mathbb E[  \chi_{\nu}^{\vec J}(t)^\ast  \chi_{\nu'}^{\vec J}(s)^\ast] = \underbrace{(C \Sigma C^T)}_{=:\Xi}\;_{\nu,\nu'} \delta(t-s),  \qquad \text{with }\Sigma =  \begin{pmatrix} 0&\id_k\\\id_k&0\end{pmatrix}.
\end{align}
With this we start manipulating the formal series in \cref{eq:formal_series_diff_1}, term by term in $n$. First we use \cref{eq:wick_prod_exp_formula} with $X = \int_0 ^T \mathrm ds \chi_0^{\vec J}(s)^\ast$ and $Y_i = \chi_{\nu_i}^{\vec J}(t_i)^\ast$ such that we get
\begin{multline}
\label{eq:formal_series_diff_2}
    \braket{\mathcal I_{\text{diffusive}}(\vec J)|\Phi}= \\ \e^{T(c_0+d_0)}\e^{\int_0 ^T \mathrm ds \chi_0^{\vec J}(s)^\ast - \tfrac 1 2 T \Xi_{0,0}}\sum_{n=0}^\infty \int_{0}^T \mathrm{D}^n\bm{t} \sum_{\bm{\nu}} \phi_{\bm{\nu}}^{(n)}(\bm{t}) \colon \prod_{i=1}^n\bigl((c_{\nu_i}+d_{\nu_i})^\ast-\Xi_{0,\nu_i}+\sum_\alpha C_{\nu_i\alpha}\xi_\alpha(t_i)\bigr) \colon.
\end{multline}

Since $X$ is real and Gaussian,
the expression
\begin{equation}
\e^{X-\tfrac1 2 \mathbb E[X^2]}=\e^{\int_0 ^T \mathrm ds \chi_0^{\vec J}(s)^\ast - \tfrac 1 2 T \Xi_{0,0}}
\end{equation}
can be interpreted as a probability density with respect to $\mu_W$ (the Cameron-Martin-Girsanov density \cite{meyerQuantumProbabilityProbabilists1995,barchielliQuantumTrajectoriesMeasurements2009}). This means that under $\e^{X-\tfrac1 2 \mathbb E[X^2]}\mathrm d\mu$ the noise $\xi$ has the law of $\xi + \lambda$ under $\mu$, where $\lambda \in \mathbb C^{2k}$ is a constant drift $\lambda_\alpha = \mathbb E[\xi_\alpha(t)X]$ which evaluates to $\lambda_x  = \tr[L_x ]^*$, $\lambda_{x+k} = \tr[L_x ]$. The Wick polynomial in $\xi$ evaluated on deterministically shifted noise is again a Wick polynomial in $\xi$ with the shift entering additively as a constant. 
Because $(C\lambda)_\nu=\Xi_{0,\nu}$, the deterministic shift cancels inside the Wick polynomial after the change of measure. Define
\begin{equation}
\label{eq:diff_F_definition}
    F_\Phi(\xi)\coloneqq\sum_{n=0}^\infty \int_0^T \mathrm{D}^n\bm{t} \sum_{\bm{\nu}} \phi_{\bm{\nu}}^{(n)}(\bm{t}) \colon \prod_{i=1}^n\bigl((c_{\nu_i}+d_{\nu_i})^\ast+\sum_\alpha C_{\nu_i\alpha}\xi_\alpha(t_i)\bigr) \colon.
\end{equation}
Then, whenever the expressions are integrable, $\mathbb E_{\mu_W}\langle\mathcal I_{\mathrm{diffusive}}|\Phi\rangle=\e^{T(c_0+d_0)}\mathbb E_{\mu_W}F_\Phi$ and $\mathbb E_{\mu_W}|\langle\mathcal I_{\mathrm{diffusive}}|\Phi\rangle|=\e^{T(c_0+d_0)}\mathbb E_{\mu_W}|F_\Phi|$. By the linearity of the Wick product, \cref{eq:linearity_wick}, it holds that
\begin{align}
\label{eq:diff_F_expression_1}
\begin{split}
    F_\Phi(\xi) &= \sum_{n=0}^\infty \int_0^T \mathrm{D}^n\bm{t} \sum_{\bm{\nu}} \phi_{\bm{\nu}}^{(n)}(\bm{t}) \colon \prod_{i=1}^n(c_{\nu_i}+d_{\nu_i}+ \chi^{\vec J}_{\nu_i}(t_i))^\ast \colon \\
    &= \sum_{n=0}^\infty \int_0^T \mathrm{D}^n\bm{t}  \sum_{\bm{\nu}} \phi_{\bm{\nu}}^{(n)}(\bm{t}) \sum_{S \subset [n]}\sum_{\bm \alpha \in [2k]^{|S|}} \prod_{i \notin S }(c_{\nu_i}+d_{\nu_i})^\ast \bigl(\prod_{i \in S}C_{\nu_i\alpha_i}\bigr)\colon\prod_{i\in S}\xi_{\alpha_i}(t_i)\colon.
    \end{split}
\end{align}
where the sum over subsets $S \subset [n]$ also includes the empty set.

Now, to obtain the deterministic instrument associated to the diffusive measurement, we perform the integration of \cref{eq:diff_F_definition} with respect to the Wiener measure over the full space $X$. Using \cref{eq:diff_F_expression_1} and \cref{eq:wick_prod_expect_value} we get
\begin{align}
\label{eq:diff_calc_det_instr}
   \int \mathrm d\mu_W \braket{\mathcal I_{\text{diffusive}}(\vec J)|\Phi} &= \int \mathrm d\mu_W  \e^{T (c_0+d_0)} \sum_{n=0}^\infty \int_0^T \mathrm{D}^n\bm{t} \sum_{\bm{\nu}} \phi_{\bm{\nu}}^{(n)}(\bm{t}) \prod_{i=1}^n(c_{\nu_i}+d_{\nu_i})^\ast 
\end{align}
We recognise this as $\bra{\mathcal I_{\text{diffusive}}(X)}\Phi\rangle$ with 
\begin{align}
  \ket{\mathcal I_{\text{diffusive}}(X)} 
  &= \ket{\Phi_T[c_0 +d_0,\{c_{\nu} +d_\nu\}_{\nu},1]}.
\end{align}
We see that the deterministic instrument of a diffusive measurement coincides with the one from the jump measurement case in \cref{eq:deter_instr_jump}. This is expected because the diffusive and the jump measurement constitute two different unravellings of the same master equation \cite{wisemanQuantumMeasurementControl2014}.

We now turn to the examination of the instrument norm. We use the variational characterisation of the $L^1$ norm to write
\begin{align}
      \int \mathrm{d}\mu_W |
    \braket{\mathcal I_{\text{diffusive}}(\vec J)|\Phi}| &= 
  \e^{T( c_0+d_0)} \int \mathrm{d}\mu_W
    \left|F_{\Phi}(\xi) \right| \nonumber\\
    &=  \e^{T(c_0+d_0)} \sup_{f \in L^\infty(\mu): \|f\|_\infty = 1}\; 
  \Bigl|\int \mathrm{d}\mu_W
   f(\xi) F_{\Phi}(\xi)\Bigr|. \label{eq:var_char_diff_inst_norm}
\end{align}
 We remark that on a finite measure space $L^\infty(\mu) \subset L^2(\mu)$. To proceed we want to make use of the \textit{Wiener-Chaos decomposition} \cite{jansonGaussianHilbertSpaces1997,meyerQuantumProbabilityProbabilists1995}. This is an explicit isomorphism between the space $L^2(\mu)$ and the symmetric Fock space over $L^2([0,T])\otimes \mathbb C^{2k}$. We will denote it as $\Gamma_{\text{diff}}$ in the following. For a symmetric kernel $h \in L^2([0,T]^m)\otimes (\mathbb C^{2k})^{\otimes m}$, the multiple Wiener-Ito integral is defined as 
\begin{equation}
    I_m(h) := \int_{[0,T]^m}\mathrm d \bm{t} \sum_{\bm \alpha}h_{\bm{\alpha}}(\bm t) \colon \prod_{i=1}^m \xi_{\alpha_i}(t_i) \colon . 
\end{equation}
It holds that 
\begin{equation}
    \mathbb E[I_m(h)I_{m'}(h')^\ast] = \delta_{mm'}m! \underbrace{\int_{[0,T]^m}\mathrm d \bm{t} \sum_{\bm \alpha}h'_{\bm{\alpha}}(\bm t)^\ast h_{\bm{\alpha}}(\bm t)}_{=:\langle h',h\rangle}.
    \label{eq:wiener_chaos_decomp_rules}
\end{equation}
The Wiener-Chaos decomposition associates with any random variable $f\in L^2(\mu)$ an element of $\Gamma_\text{diff}$ characterised by a series of kernels $(h^{(m)})_m$, such that
\begin{equation}
    f = \sum_{m=0}^\infty I_m(h^{(m)}),
\end{equation}
with $\|f\|^2_{L^2(\mu)}= \sum_{m=0}^\infty m! \|h^{(m)}\|^2$, and $ \|h^{(m)}\|^2 = \braket{h^{(m)},h^{(m)}}$.
We will now show that as a formal series 
\begin{equation}
    F_\Phi = \sum_{m=0}^\infty \frac{1}{m!}I_m(g^{(m)})
    \label{eq:chaos_decomp_F_phi}
\end{equation}
with 
\begin{align}
\label{eq:diff_g_psi_definitions}
\begin{split}
    g^{(m)}_{\bm \alpha}(\bm t) &:= \sum_{\bm \nu}\psi^{(m)}_{\bm \nu}(\bm t)\prod_{i=1}^m C_{\nu_i,\alpha_i}, \quad \text{and}\\
   \psi^{(m)}_{\bm \nu}(\bm t) &:= \sum_{r=0}^\infty \frac{1}{r!} \int_{[0,T]^r}\mathrm d \bm s \sum_{\bm \nu'}\phi^{(m+r)}_{\bm{\nu}\bm{\nu'}}(\bm t, \bm s)(c_{\nu'_1}+d_{\nu'_1})^\ast\cdots (c_{\nu'_r}+d_{\nu'_r})^\ast.
   \end{split}
\end{align}
Starting from \cref{eq:diff_F_expression_1} we use the symmetry of $\phi_{\bm{\nu}}^{(n)}(\bm t)$ under the joint permutations of $\bm{t}$ and $\bm{\nu}$ to write 
\begin{align}
    F_\Phi 
    &= \sum_{n=0}^\infty \frac{1}{n!} \int_{[0,T]^n} \mathrm{d}\bm{t} \sum_{\bm{\nu}} \phi_{\bm{\nu}}^{(n)}(\bm{t}) \sum_{m = 0}^{n} \binom{n}{m} \prod_{i=m+1 }^{n}(c_{\nu_i}+d_{\nu_i})^\ast \colon \prod_{i=1}^m(\chi^{\vec J}_{\nu_i}(t_i))^\ast \colon.
\end{align}
 Denoting with $\bm \mu= (\nu_{m+1},\cdots,\nu_{n})$ and $\bm \nu'=(\nu_{1},\cdots,\nu_{m})$ and changing variables $r = n-m$ we get as a formal manipulation that
\begin{align}
    F_\Phi 
    &= \sum_{m=0}^\infty  \sum_{r = 0}^{\infty} \frac{1}{m!} \sum_{\bm{\nu'}} \int_{[0,T]^{m}}  \mathrm{d}\bm{t} \frac{1}{r!} \int_{[0,T]^{r}}\mathrm{d}\bm{s} \sum_{\bm \mu} \phi_{\bm{\nu'}\bm \mu}^{(r+m)}(\bm{t},\bm{s}) \prod_{i=1 }^{r}(c_{\mu_i}+d_{\mu_i})^\ast \colon \prod_{i=1}^{m}(\chi^{\vec J}_{\nu'_i}(t_i))^\ast \colon , 
\end{align}
which shows the claimed \cref{eq:chaos_decomp_F_phi}. As a consequence, we get together with using \cref{eq:wiener_chaos_decomp_rules} that
\begin{align}
    \int \mathrm d\mu_W f(\xi) F_\Phi(\xi)
    &=  \sum_{m=0}^\infty \int_{[0,T]^m}\mathrm d \bm{t} \sum_{\bm \alpha} \bigl(h_{\bm \alpha}^{(m)}(\bm t)\bigr)^*g_{\bm \alpha}^{(m)}(\bm t) 
\end{align}
Note that the boundedness of $f$ implies $\|f\|^2_{L^2(\mu)}= \sum_{m=0}^\infty m! \|h^{(m)}\|^2 \leq 1$. 
This together with two applications of Cauchy--Schwarz gives 
\begin{align}
   \Bigl|\int \mathrm d\mu_W f(\xi) F_\Phi(\xi)\Bigr|
    &\leq  \sum_{m=0}^\infty \|h^{(m)}\| \; \|g^{(m)}\| \nonumber \\
    & \leq \Bigl(\sum_{m=0}^\infty \frac{1}{m!} \; \|g^{(m)}\|^2\Bigr)^{1/2} .
\end{align}
The matrix $C$ defined in \cref{eq:matrix_C_coefficient_change} informally converts from the $\bm \alpha$ basis to the $\bm \nu$ basis. It contains the information about the measured operators $L_x $, $x \in  [k].$ We can proceed by bounding this expression from above by considering its restriction to the submatrix $C_\perp$ where $\nu \neq 0$. 
Then
\begin{align}
 \Bigl|\int \mathrm d\mu_W f(\xi) F_\Phi(\xi)\Bigr|
    &\leq \Bigl(\sum_{m=0}^\infty \frac{1}{m!} \| C_\perp \|^{2m} \|\psi^{(m)}\|^2\Bigr)^{1/2} .
\end{align}
Here, $\| C_\perp \|$ denotes the operator norm, or largest singular value, of $C_\perp$. To bound $\|\psi^{(m)}\|^2$ we apply Cauchy--Schwarz to get 
\begin{align}
   \|\psi^{(m)}\|^2 &= \int_{[0,T]^m}\mathrm d \bm{t} \sum_{\bm \nu \in [d_{\rm S}^4-1]^m}  |\sum_{r=0}^\infty \frac{1}{r!}  \sum_{\bm \nu' \in [d_{\rm S}^4-1]^r}\int_{[0,T]^r}\mathrm d \bm s\phi^{(m+r)}_{\bm{\nu}\bm{\nu'}}(\bm t, \bm s)(c_{\nu'_1}+d_{\nu'_1})^\ast\cdots (c_{\nu'_r}+d_{\nu'_r})^\ast|^2 \nonumber \\
   &\leq \int_{[0,T]^m}\mathrm d \bm{t} \sum_{\bm \nu \in [d_{\rm S}^4-1]^m} \Bigl|\sum_{r=0}^\infty \frac{1}{r!} \sqrt{\sum_{\bm \nu'\in [d_{\rm S}^4-1]^r}\int_{[0,T]^r}\mathrm d \bm s |\phi^{(m+r)}_{\bm{\nu}\bm{\nu'}}(\bm t, \bm s)|^2}\sqrt{T^r\sum_{\bm \nu'\in [d_{\rm S}^4-1]^r}\prod_{i}|c_{\nu'_i}+d_{\nu'_i}|^2}\Bigr|^2 \nonumber \\ 
   &\leq \int_{[0,T]^m}\mathrm d \bm{t} \sum_{\bm \nu \in [d_{\rm S}^4-1]^m}  \Bigl|\sum_{r=0}^\infty \frac{1}{r!} \sum_{\bm \nu'\in [d_{\rm S}^4-1]^r}\int_{[0,T]^r}\mathrm d \bm s |\phi^{(m+r)}_{\bm{\nu}\bm{\nu'}}(\bm t, \bm s)|^2 \Bigr| \Bigl|\sum_{r=0}^\infty \frac{T^r}{r!}(\sum_{\nu \neq 0}|c_{\nu}+d_{\nu}|^2)^r\Bigr| \nonumber \\
 & \leq \e^{T\sum_{\nu \neq 0 }|c_{\nu}+d_{\nu}|^2}  \sum_{r=0}^\infty \frac{1}{r!} \int_{[0,T]^m}\mathrm d \bm{t}\int_{[0,T]^r}\mathrm d \bm s \sum_{\bm \nu \in [d_{\rm S}^4-1]^m} \sum_{\bm \nu'\in [d_{\rm S}^4-1]^r} |\phi^{(m+r)}_{\bm{\nu}\bm{\nu'}}(\bm t, \bm s)|^2.
\end{align}
Putting things together, changing the summation variables and using the binomial theorem we get
\begin{align}
   \|\mathcal{I}_{\text{diffusive}} \|_{\Gamma\rightarrow L^1(\mathcal{X})} 
    &\leq  \e^{T(c_0+d_0+\|\bm c_\perp + \bm d_\perp\|_2^2}/2) \sup_{\substack{\ket{\Phi} \in \Gamma:\\ \|\ket{\Phi}\|_\Gamma \leq 1}}  \Bigl(\sum_{n=0}^\infty \sum_{r=0}^{n} \binom{n}{r}\| C_\perp \|^{2(n-r)} 1^{r} \int_0^T\mathrm D^n \bm{t}   \sum_{\bm \nu} |\phi^{(n)}_{\bm{\nu}}(\bm t)|^2\Bigr)^{1/2} \nonumber \\
    &= \e^{T(c_0+d_0+\|\bm c_\perp + \bm d_\perp\|_2^2/2)} \sup_{\substack{\ket{\Phi} \in \Gamma:\\ \|\ket{\Phi}\|_\Gamma \leq 1}} \Bigl(\sum_{n=0}^\infty (1+\| C_\perp \|^{2})^n \int_0^T\mathrm D^n \bm{t}  \sum_{\bm \nu} |\phi^{(n)}_{\bm{\nu}}(\bm t)|^2\Bigr)^{1/2}, \label{eq:instr_norm_diff_infinite}
\end{align}
where we abbreviated $\|\bm c_\perp + \bm d_\perp\|_2^2 := \sum_{\nu \neq 0 }|c_{\nu}+d_{\nu}|^2$. 

After these manipulations we are ready to deduce some properties of the diffusive instrument. The series in \cref{eq:instr_norm_diff_infinite} is finite precisely when $\sum_{n\ge0}(1+\|C_\perp\|^2)^n\|\Phi^{(n)}\|_\Gamma^2<\infty$. This gives a sufficient condition on the particle wavefunctions for the displayed bound to establish integrability; super-exponential decay suffices but decay merely proportional to $(1+\|C_\perp\|^2)^{-n}$ does not. This excludes perfectly admissible vectors in $\Gamma$: for $\int_0^T\mathrm D^n \bm{t}\sum_{\bm \nu}|\phi^{(n)}_{\bm{\nu}}(\bm t)|^2 \propto 1/n^2$, for example, the series diverges.
Our upper bound can therefore only show that the diffusive instrument for general measurement operators $L_x $, $x \in  [k]$ is well defined on a dense subspace of $\Gamma$, e.g. on vectors with finite particle number or with sufficiently rapidly decaying particle-number distribution. 
By virtue of the bounds obtained in~\cref{sec:bounds_on_norms_cPT_wavefunctions}, we know that the domain of the instrument contains all the cPTs, which are hence mapped to functions with uniformly bounded $L^1$ norm.
However, we will now show that for a general Fock vector the bound inside the supremum~\cref{eq:instr_norm_diff_infinite} is attainable and hence the instrument norm of the diffusive instrument diverges.

We consider for any $z \in \mathbb C$ the exponential vector $\ket{\Phi_z} \in \Gamma$ with $n$-particle wavefunctions 
\begin{equation}
    \phi_{\bm \nu}^{(n)}(\bm t) = \e^{-|z|^2/2}z^n \prod_{i=1}^n u(t_i) v_{\nu_i},
\end{equation} 
where $u \in L^2([0,T])$ satisfies 
\begin{align}
    \int_{[0,T]}\mathrm dt\;  |u(t)|^2 &=1 \nonumber  \\
    \int_{[0,T]}\mathrm dt \; u(t) & =\int_{[0,T]}\mathrm dt \; u^2(t)=0, \label{eq:averaging_prop_satur_example}
\end{align}
and the normalised vectors $v \in \mathbb C^{d_{\rm S}^4-1}, \tilde v \in \mathbb C^{2k}$ are such that $C^T_\perp v = \|C_\perp\| \tilde v $. Then $\ket{\Phi_z}$ lies in the domain of the diffusive instrument map and $\|\ket{\Phi_z} \|_\Gamma^{2} = \e^{-|z|^2}\sum_n |z|^{2n}/n! =1$.  
Furthermore, with these choices $\psi^{(m)}_{\bm \nu} =  \phi^{(m)}_{\bm \nu}$ in \cref{eq:diff_g_psi_definitions} and the corresponding $g^{(m)}_{\bm \alpha}$ is
\begin{equation}
    g^{(m)}_{\bm \alpha}(\bm t) = \|C_\perp\|^m \e^{-|z|^2/2}z^m \prod_{i=1}^m u(t_i) \tilde v_{\alpha_i}
\end{equation}
such that we get
\begin{align}
    F_{\Phi_z} &= \sum_{m=0}^\infty \frac{1}{m!}I_m(g^{(m)}) \nonumber \\
    &= \e^{-|z|^2/2} \sum_{m=0}^\infty \frac{\|C_\perp\|^m z^m }{m!} \int_{[0,T]^m}\mathrm d \bm{t} \sum_{\bm \alpha}\Bigl(\prod_{i=1}^m u(t_i)\tilde v_{\alpha_i}\Bigr)\colon\prod_{i=1}^m\xi_{\alpha_i}(t_i)\colon \nonumber \\
    &= \e^{-|z|^2/2} \sum_{m=0}^\infty \frac{\|C_\perp\|^m z^m }{m!}\colon\Bigl(\underbrace{ \int_{[0,T]}\mathrm d t \sum_{\alpha} u(t) \tilde v_{\alpha} \xi_{\alpha}(t)}_{=: X } \Bigr)^m\colon.
\end{align}
We recognise $X$ as being a centred complex Gaussian variable. By the rules of the Wick product, \cref{eq:wick_prod_expec_of_exp}, we have 
\begin{align}
  \int \mathrm d\mu_W \bigl|  F_{\Phi_z} \bigr| &= \e^{-|z|^2/2} \int \mathrm d\mu_W \bigl| \colon \e^{\|C_\perp\| z X} \colon \bigr| \nonumber \\
     &= \e^{-|z|^2/2}\int \mathrm d\mu_W \bigl|  \e^{\|C_\perp\| z X-\tfrac1 2 \|C_\perp\|^2 z^2 \mathbb E[X^2]}  \bigr| \nonumber \\
     &= \e^{-|z|^2/2}\int \mathrm d\mu_W \bigl|  \e^{\|C_\perp\| z X-\tfrac1 2 \|C_\perp\|^2 z^2 \bigl(2\sum_{x\in[k]}\tilde v_x\tilde v_{x+k}\bigr)(\int_0^T \mathrm d t u^2(t))} \bigr| \nonumber  \\
     &= \e^{-|z|^2/2}\int \mathrm d\mu_W \bigl|  \e^{\|C_\perp\| z X} \bigr|,
\end{align}
where in the second-to-last line we used that $\mathbb E[\xi_{\alpha}(t)\xi_{\beta}(s)] = \Sigma_{\alpha \beta}\delta(t-s)$ and the last line follows from \cref{eq:averaging_prop_satur_example}. Then using that for linear functionals $Y$ of jointly Gaussian random variables $\mathbb E[\e^{Y}]= \e^{\tfrac 1 2 \mathbb E [Y^2]}$, we get
\begin{align}
  \int \mathrm{d}\mu_W |
    \braket{\mathcal I_{\text{diffusive}}(\vec J)|\Phi_z}| &= \e^{T(c_0+d_0)}\e^{-|z|^2/2}\int \mathrm d\mu_W  \e^{ \|C_\perp\|(z X + z^\ast X^\ast)/2} \nonumber \\
    &= \e^{T(c_0+d_0)}\e^{-|z|^2/2} \e^{ \|C_\perp\|^2 \mathbb E[(z X + z^\ast X^\ast)^2]/8} \nonumber \\
    &= \e^{T(c_0+d_0)}\e^{-|z|^2/2} \e^{\|C_\perp\|^2 |z|^2 (\sum_{\alpha} |\tilde v_\alpha|^2)(\int_0^T \mathrm d t |u|^2(t))/4} \nonumber \\
    &=  \e^{T(c_0+d_0)} \e^{\tfrac{1}{4}|z|^2(\|C_\perp\|^2-2)}. \label{eq:diff_saturs_bound_exam}
\end{align}
Since $\|\ket{\Phi_z}\|=1 $ for all $z \in \mathbb C$, letting $|z| \rightarrow \infty$ shows that the diffusive instrument cannot be uniformly bounded on the dense domain of vectors with exponentially decaying particle distribution whenever $\|C_\perp\|^2 > 2$. To see how $\|C_\perp\|$ depends on the measurement operators we use that $\{\mathbb P_\nu\}_\nu$ is Hermitian and orthonormal with $\mathbb P_0=\id/d_{\rm S}$,
so that Parseval's identity gives, for any operators $X,Y$ on the doubled space $\mathcal H_{\rm S}\otimes \mathcal H_{\rm S}$,
$\sum_{\nu\neq 0}\tr[\mathbb P_\nu X]^*\tr[\mathbb P_\nu Y]
=\tr[X^\dagger Y]-\tr[X]^*\tr[Y]/d_{\rm S}^2$.
Applied to the columns $L_x \otimes\id$ and $\id\otimes L_x ^*$ of $C_\perp$ (see \cref{eq:matrix_C_coefficient_change}), the mixed
contributions cancel and the Gram matrix becomes block diagonal,
\begin{equation}
  C_\perp^\dagger C_\perp=\begin{pmatrix}A&0\\0&A^T\end{pmatrix},
  \qquad
  A_{xx'}:=d_{\rm S}\tr\big[\tilde L _x ^\dagger \tilde L _{x'}\big],
  \qquad
  \tilde L _x :=L_x -\frac{\tr[L_x ]}{d_{\rm S}}\id,
\end{equation}
where $A\geq0$ is $d_{\rm S}$ times the Gram matrix of the traceless parts of the measurement
operators. Since $A$ and $A^T$ have the same spectrum, we obtain
\begin{equation}
  \|C_\perp\|^2=\|A\|\leq \tr A=d_{\rm S}\sum_{x \in [k]}\big\|\tilde L_x \big\|_{\text{HS}}^2.
\end{equation}
If the $\tilde L_x $ are Hilbert--Schmidt orthogonal it exactly holds that $\|C_\perp\|^2 = d_{\rm S} \max_x  \|\tilde L_x \|^2_{\text{HS}}$. We see that the quantity controlling the instrument norm is therefore the Hilbert--Schmidt norm of
the traceless part of the measurement operators, which quantifies the strength of the
system-meter coupling. This is a further instance of the connection between the
distinguishing power of an instrument and the cost of implementing it.

As the form of the diffusive instrument has been obtained by considering its action only on the set of PCRs, the unbounded instrument norm implies that we cannot use density of PCRs in cPTs to extend its action to all cPTs.
Intuitively, one can understand this as the consequence of the fact that every run of the diffusive instrument probes the process over the whole interval $[0,T]$, making it sensitive to pointwise values of the cPT, which is not allowed due to the $L^2$-like equivalence on $\Gamma$.
Furthermore, the diverging instrument norm implies that the diffusive instrument has infinite distinguishing power: using the bounds obtained in~\cref{sec:bounds_on_norms_cPT_wavefunctions} in~\cref{eq:instr_norm_diff_infinite} we can evaluate an upper bound on the total variation distance between the diffusive instrument statistics evaluated on two cPTs $\cPT$ and $\ket{\Upsilon'_T}$
\begin{equation}
    \text{TV}(\cPT, \ket{\Upsilon'_T}) \le \frac{\e^{T(c_0 + d_0 + \frac12\|\bm{c} + \bm{d}\|_2^2)}}{\|\cPT - \ket{\Upsilon'_T}\|} \left(\sum_{n=0}^\infty (1 + \|C_\perp\|^2)^n 32\frac{T^n}{n!}\right)^\frac12 = \frac{4\sqrt{2}\e^{T\left(c_0 + d_0 + \frac12 (\|\bm{c}+\bm{d}\|_2^2 + \|C_\perp\|^2 + 1)\right)}}{\|\cPT - \ket{\Upsilon'_T}\|},
\end{equation}
where the difference in the denominator can be arbitrarily small.
This is the case even if the cPTs are in the PCR.
While it is not completely clear that this bound is attained by physical processes, since we have seen that it is tight for general Fock vectors this indicates unbounded distinguishing power of the diffusive instrument.

The infinite instrument norm is directly related to the idealised unphysical nature of the diffusive limit: the diffusive measurement
can be understood as the, suitably rescaled, jump measurement in the limit of infinite rate. For homodyne detection this corresponds to interfering the output field with a local oscillator of diverging amplitude~\cite{wisemanQuantumMeasurementControl2014}. In this limit the measurement noise has a fixed intensity per unit time, independent of how many excitations have already been emitted. Since in the Fock space picture every particle corresponds to one emission time, this translates into a constant weight per particle, namely the factor  $(\|C_\perp\|^2-2)$ in \cref{eq:diff_saturs_bound_exam}. The diffusive instrument thus probes a Fock state vector in any particle sector with the same sensitivity and no sector is suppressed. 

Any measurement that is actually implemented in the lab has a finite bandwidth and hence cannot access the currents $J_x$ directly \cite{guilminTimeaveragedContinuousQuantum2025}.
Instead, it resolves smoothed versions of the integrated records $Y_x(t)=\int_0^t J_x(s)\,{\rm d}s$,
where the smoothening approximates the white-noise limit better as the strength of the local oscillator increases.
As the strength of the local oscillator approaches the diffusive limit, the distinguishing power quantified by the instrument norm of the associated instrument diverges and so does the energy required to power such a measurement.
We thus see, once more, that the distinguishing power of an
instrument in our framework is tied to the energy required to implement it.

\section{Gelfand triple}

From the Dyson series~\cref{ssec:cpts in the Fock space} we have seen that any process that has an associated (finite) Hamiltonian generator is a PCR. 
We have also seen in~\cref{sec:instruments} that the instrument norm is physically connected to the energy needed to perform the measurement and we have seen that the $\Gamma$-norm quantifies the statistical distinguishability of cPTs under bounded energy instruments, rendering it a natural norm on cPTs.
As the span of PCRs is not closed under the $\Gamma$-norm, we would like to take the closure under the $\Gamma$-norm to obtain a mathematically well-defined theory.
Theorem 1 in the main text shows that this is precisely what the cPT definition ensures.

But this places the wavefunctions of a cPT in the space of $L^2$ functions, which renders probing them pointwise meaningless.
This reflects the physical reality that nothing is truly instantaneous in nature and real physical measurements take a finite amount of time.
However, in many practical situations it is a very good simplifying approximation to think of processes happening instantaneously: for instance, we often talk of measuring a POVM at a given time $\tau\in[0,T]$.
We would like to include such approximations into our theory.

The correct mathematical structure, unlocked by Theorem 1 of the main text, that allows us to do this is that of a rigged Hilbert space or a Gelfand triple~\cite{gelfandGeneralizedFunctionsApplications2014,wlokaPartialDifferentialEquations1987,madridRoleRiggedHilbert2005}.
This is the same structure that allows one to rigorously treat the eigenfunctions of the position and momentum operators when treating a quantum mechanical particle in space.
Concretely, the Gelfand triple introduces the space of regular cPTs by closing the span of PCRs under a different norm: the Sobolev Fock norm on $\Gamma(H^1([0,T]))$.
This imposes (weak) differentiability of the wavefunctions, placing the instantaneous instruments into the dual space of this closure.

We now state and prove a rigorous version of Corollary 1 in the main text.

\setcounter{corollary}{\numexpr\value{savedcorno}-1\relax}

\begin{corollary}[Process Gelfand triple] \label{lem:gelfand}
    Let $\|\cdot \|_{H^1}$ be the norm on $\Gamma(H^1([0, T]))$, the Fock space of the Sobolev space $H^1([0, T])$.
        Let $\Lambda_{\|\cdot\|_{H^1}}$ and $\Lambda_{\|\cdot\|_{\Gamma}}$ be the closures of $\Span(\mathsf{PCR})$ under $\|\cdot \|_{H^1}$ and $\|\cdot \|_\Gamma$ respectively.
    Now, $\Lambda_{\|\cdot\|_{H^1}} \subsetneq \Lambda_{\|\cdot\|_\Gamma}$ and the inclusion map $\imath: \Lambda_{\|\cdot\|_{H^1}} \rightarrow \Lambda_{\|\cdot\|_\Gamma}$ defines a Gelfand triple
    \begin{equation}\label{eq:process gelfand triple}
        \Lambda_{\|\cdot\|_{H^1}} \hookrightarrow \Lambda_{\|\cdot\|_\Gamma} \hookrightarrow \Lambda_{\|\cdot\|_{H^1}}^*,
    \end{equation}
    where the notation $\hookrightarrow$ means that there exists a continuous injective map with a dense range between the two spaces.
    The space $\Lambda_{\|\cdot\|_{H^1}}^*$ includes the effects of the instrument measuring a POVM at some $t\in[0,T]$.
\end{corollary}

\begin{proof}
The Sobolev Fock norm of a vector with wavefunctions $\phi^{(n)}_{\bm{\nu}}(\bm{t})$ is given by
\begin{equation}
    \|\ket{\phi}\|_{H^1}^2 = \sum_{n=0}^\infty \sum_{\bm{\nu}\in[d_{\rm S}^4 - 1]^n} \sum_{\bm{\alpha}\in\{0,1\}^n} \int_0^T {\rm D}^n \bm{t} \ |\partial_{\bm{\alpha}} \phi^{(n)}_{\bm{\nu}}(\bm{t})|^2,
\end{equation}
where $\partial_{\bm{\alpha}}$ takes the derivative with respect to all $t_i$ with $\alpha_i = 1$.
If $\ket{\phi}$ is in the PCR, we can write
\begin{equation}
    \partial_{\bm{\alpha}} \phi_{\bm{\nu}}^{(n)}(\bm{t}) = \bbra{\id} U(T, t_n) X_{\alpha_n,\nu_n}(t_n) U(t_n, t_{n-1}) \dots X_{\alpha_1,\nu_1}(t_1) U(t_1, 0) \kket{\psi},
\end{equation}
where $U(t', t) \coloneqq \mathcal{T}\e^{\int_t^{t'}{\rm d}s \mathbb{H}(s)}$ and $X_{0,\nu}(t) = \mathbb{P}_\nu$ while $X_{1,\nu}(t)=[\mathbb P_\nu,\mathbb H(t)]$, such that we can bound
\begin{equation}
    |\partial_{\bm{\alpha}} \phi^{(n)}_{\bm{\nu}}(\bm{t})| \le d_{\rm S}d_{\rm E} \frac{(2\|\mathbb H\|_\infty)^{|\bm\alpha|}}{d_{\rm S}^n},
\end{equation}
where $|\bm{\alpha}|$ is the number of non-zero elements of $\bm{\alpha}$.
Therefore, for a PCR
\begin{align}
    \|\ket{\phi}\|_{H^1}^2 &\le \sum_{n=0}^\infty \sum_{\bm{\nu}}\sum_{\bm{\alpha}} \int_0^T {\rm D}^n\bm{t} \ d_{\rm S}^2 d_{\rm E}^2 \frac{(2\|\mathbb{H}\|_\infty)^{2|\bm{\alpha}|}}{d_{\rm S}^{2n}} \nonumber\\
    &= \sum_{n=0}^\infty \sum_{\bm{\nu}}\sum_{k=0}^n \binom{n}{k} \frac{T^n}{n!} d_{\rm S}^2 d_{\rm E}^2 \frac{(2\|\mathbb{H}\|_\infty)^{2k}}{d_{\rm S}^{2n}} \nonumber\\
    &= d_{\rm S}^{2}d_{\rm E}^{2} \sum_{n=0}^\infty \sum_{\bm{\nu}} \frac{T^n(4\|\mathbb{H}\|_\infty^2 + 1)^n}{n! d_{\rm S}^{2n}} \nonumber\\
    &= d_{\rm S}^2d_{\rm E}^2 \sum_{n=0}^\infty \frac{T^n}{n!}(4\|\mathbb{H}\|_\infty^2 + 1)^n \left(d_{\rm S}^2 - \frac1{d_{\rm S}^2}\right)^n \nonumber\\
    &= d_{\rm S}^2d_{\rm E}^2 \e^{T\left(d_{\rm S}^2 - \frac1{d_{\rm S}^2}\right)(4\|\mathbb{H}\|_\infty^2 + 1)}
\end{align}
and we find that the Sobolev norm increases with the operator norm of the Hamiltonian, but remains finite for all PCRs.
Hence $\Span(\mathsf{PCR})$ can be closed under $\|\cdot\|_{H^1}$.
Since the topology induced by $\|\cdot\|_{H^1}$ is strictly stronger than the one induced by $\|\cdot\|_\Gamma$, and from the above calculation we know that there are sequences of PCRs with no uniform bound on $\|\cdot\|_{H^1}$, we know that $\Lambda_{\|\cdot\|_{H^1}} \subsetneq \Lambda_{\|\cdot\|_\Gamma}$.

Now, consider the inclusion map $\imath:\Lambda_{\|\cdot\|_{H^1}} \to \Lambda_{\|\cdot\|_\Gamma}$ and its adjoint $\imath^*:\Lambda_{\|\cdot\|_\Gamma}^* \to \Lambda_{\|\cdot\|_{H^1}}^*$.
The map $\imath$ has dense range because its range contains $\Span(\mathsf{PCR})$.
Furthermore, $\Lambda_{\|\cdot\|_{H^1}}$ is a closed subspace of $\Gamma(H^1([0,T]))$ and therefore it is a reflexive Banach space.
Hence, a standard result (see, e.g., section 17.1 in~\cite{wlokaPartialDifferentialEquations1987}) implies that the map $\imath^*$ has a dense image, which allows us to obtain the entire $\Lambda_{\|\cdot\|_{H^1}}^*$ by the closure of the image.
This completes the construction of the process Gelfand triple.

Finally, as explained in the companion work~\cite{wassnerOperationalContinuumLimit2026}, the effects of the instrument measuring a POVM $\{E_x\}_x$ at time $\tau$ are $\bra{\mathcal{I}(\{x\})} = \sum_{\nu=0}^{d_{\rm S}^4-1} \chi_\nu(x) \bra{\Omega} \xi_\nu(\tau),$ where $\chi_\nu(x) = \tr \mathbb{P}_\nu E_x$.
These effects have only vacuum and one-particle components. Point evaluation at $\tau$ is a bounded functional on $H^1([0,T])$, so each effect defines a continuous functional on the Sobolev Fock space and restricts to an element of $\Lambda_{\|\cdot\|_{H^1}}^*$.
\end{proof}

\section{Proof of Theorem 1 in the main text}

\begin{proof}[\unskip\nopunct]
Consider the intervals $I_k \coloneqq \left[(k-1)\frac{T}{N}, k\frac{T}{N}\right)$ with $k \in [N]$ and define the indicator function
\begin{equation}
    \bm{1}_{I_k}(t) \coloneqq \begin{cases}
        1 & t \in I_k\\
        0 & t \notin I_k.
    \end{cases}
\end{equation}
Note that $\|\bm{1}_{I_k}\|_2 = \sqrt{\frac{T}{N}}$.
Consider the annihilation operators on $\Gamma$ corresponding to the normalised indicator functions
\begin{equation}
    a_\nu(k) \coloneqq \sqrt{\frac{N}{T}} \int_{I_k} {\rm d}t \ \psi_\nu(t),
\end{equation}
which satisfy the commutation relations
\begin{equation}
    [a_\nu(k), a_{\nu'}^\dag(k')] = \delta_{\nu, \nu'}\delta_{k, k'}.
\end{equation}
We introduce the projectors 
\begin{align}
    P_N &\coloneqq  \sum_{\substack{\bm{m}\in \mathbb N_0^{N \times (d_{\rm S}^4-1)} }}\frac{1}{m_{(1,1)}!\cdots m_{(N,d_{\rm S}^4-1)}!} \prod_{\substack{(k,\nu)\\k\in[N] \\
    \nu \in [d_{\rm S}^4-1]}} \bigl(a_{\nu}^\dag(k)\bigr)^{m_{(k,\nu)}} \ket{\Omega}\!\bra{\Omega} \prod_{\substack{(k',\nu')\\k'\in[N] \\
    \nu' \in [d_{\rm S}^4-1]}} \bigl(a_{\nu'}(k')\bigr)^{m_{(k',\nu')}} \label{eq:projector_P_N} \\
    P_N^{(0,1)} &\coloneqq \sum_{n=0}^N \sum_{\bm{\nu} \in [d_{\rm S}^4 -1]^n} \sum_{1 \le \ell_1 < \dots < \ell_n\le N} \prod_{k\in[n]} a_{\nu_k}^\dag(\ell_k) \ket{\Omega}\!\bra{\Omega} \prod_{k' \in [n]} a_{\nu_{k'}}(\ell_{k'}) \label{eq:projector_P_N_01}. 
\end{align}
Now consider a positive and causal vector $\ket{\Psi} \in \Gamma_T^{(d_{\rm S}^4 - 1)}$. For all $ \bm{\tilde \nu}\in ([d_{\rm S}^4 -1] \cup \{0\} )^N$ we define
\begin{align}
    \Upsilon_{\bm{\tilde \nu}}&:= (\tfrac{N}{T})^{\|\bm{\tilde \nu}\|_0/2}\bra{\Omega}\prod_{i \in \text{supp}(\bm{\tilde \nu})} \bigl(a_{\tilde \nu_i}(i)\bigr)  P_N^{(0,1)} \ket{\Psi} \nonumber \\
    &= (\tfrac{N}{T})^{\|\bm{\tilde \nu}\|_0/2}\bra{\Omega}\prod_{i \in \text{supp}(\bm{\tilde \nu})} \bigl(a_{\tilde \nu_i}(i)\bigr) \sum_{n=0}^N \sum_{\bm{\nu} \in [d_{\rm S}^4 -1]^n} \sum_{1\le\ell_1<\dots<\ell_n\le N} \prod_{k\in[n]} a_{\nu_k}^\dag(\ell_k) \ket{\Omega}\!\bra{\Omega} \prod_{k' \in [n]} a_{\nu_{k'}}(\ell_{k'})\ket{\Psi} \nonumber \\
    &= (\tfrac{N}{T})^{\|\bm{\tilde \nu}\|_0/2}\bra{\Omega}\prod_{i \in \text{supp}(\bm{\tilde \nu})} \bigl(a_{\tilde \nu_i}(i)\bigr)\ket{\Psi} \nonumber \\
    &= (\tfrac{N}{T})^{\|\bm{\tilde \nu}\|_0} \Bigl(\prod_{i \in \text{supp}(\bm{\tilde \nu})} \int_{I_{i}} {\rm d}t_i \Bigr) \psi^{(\|\bm{\tilde \nu}\|_0)}_{\bm{\tilde \nu}\neq 0}(\bm{t}) \nonumber\\
    &= (\tfrac{N}{T})^{N} \Bigl(\prod_{i \in [N]} \int_{I_{i}} {\rm d}t_i \Bigr)  g_{\bm{\tilde \nu}}(\bm{t}) ,\label{eq:ave_dPT_elements}
\end{align}
where $\psi^{(\|\bm{\tilde \nu}\|_0)}_{\bm{\tilde \nu}\neq 0}(\bm{t})$ is the $\|\tilde{\bm{\nu}}\|_0$-particle amplitude of $\ket{\Psi}$ with the particle species given by the non-zero elements of the string $\bm{\tilde \nu}$ and
\begin{equation}
    g_{\bm{\tilde \nu}}(\bm{t}):= \bra{\Omega}\xi_{\tilde \nu_N}(t_N)\xi_{\tilde \nu_{N-1}}(t_{N-1})\cdots \xi_{\tilde \nu_1}(t_1)\ket{\Psi}.
\end{equation}

Because $\ket{\Psi}$ is positive and causal we have by definition that the set $\{\Upsilon_{\bm{ \nu}}\}_{\bm{\nu}}$ defined in \cref{eq:ave_dPT_elements} are the elements of the Choi matrix of an $N$-step discrete process tensor $\Upsilon_N$. 
Therefore, we can dilate it as described in \cite{milzQuantumStochasticProcesses2021} and obtain $N$ unitary channels of dimension $d_{\rm S}^{4 N}$ and a pure initial state $\kket{\psi}$, such that for any $\alpha,\alpha'\in \{0,\ldots,d_{\rm S}^2-1\}$
\begin{align}
  \sum_{\nu_N \in \{[d_{\rm S}^4-1]\cup \{0\}\}} c_{0, \alpha'}^{\nu_N}  \Upsilon_{\bm{\nu}}&= \bbra{\id_{\rm S}}P_0\rrangle \bbra{\id_{\rm E} \otimes P_{\alpha'}}\mathcal U_{N-1}\mathbb P_{\nu_{N-1}}\mathcal U_{N-2}\cdots \mathcal U_{1}\mathbb P_{\nu_1}\mathcal{U}_0\kket{\psi},
\end{align}
where $\mathcal U_k$ are the vectorisations of the unitary channels and $c_{\alpha, \alpha'}^\nu$ is the basis transformation such that $\kket{P_\alpha}\!\bbra{P_{\alpha'}} = \sum_\nu c_{\alpha, \alpha'}^\nu \mathbb{P}_\nu$. 
For each $\mathcal U_k$ there exists $\mathbb H_k = -i(H_k \otimes \id - \id \otimes H_k^\ast)$ with a Hermitian matrix $H_k$ such that $\mathcal U_k = \exp(\mathbb H_k)$. 
Due to the causality condition we can write that 
\begin{align}
 &\quad \sum_{\nu_N\in[d_{\rm S}^4-1]\cup\{0\}}c_{\alpha,\alpha'}^{\nu_N}\Upsilon_{\bm\nu}=   \bbra{\id_{\rm S}}P_\alpha \rrangle\bbra{\id_{\rm E} \otimes P_{\alpha'}}\mathcal U_{N-1}\mathbb P_{\nu_{N-1}}\mathcal U_{N-2}\cdots \mathcal U_{1}\mathbb P_{\nu_1}\mathcal{U}_0\kket{\psi} \nonumber\\
 \Leftrightarrow & \quad \sum_{\alpha,\alpha'}(c_{\alpha, \alpha'}^{\nu'_N})^\ast  \sum_{\nu_N \in \{[d_{\rm S}^4-1]\cup \{0\}\}} c_{\alpha, \alpha'}^{\nu_N}  \Upsilon_{\bm{\nu}} = \sum_{\alpha,\alpha'} (c_{\alpha, \alpha'}^{\nu'_N})^\ast\bbra{\id_{\rm S}}P_\alpha \rrangle\bbra{\id_{\rm E} \otimes P_{\alpha'}}\mathcal U_{N-1}\mathbb P_{\nu_{N-1}}\mathcal U_{N-2}\cdots \mathcal U_{1}\mathbb P_{\nu_1}\mathcal{U}_0\kket{\psi} \nonumber \\
 \Leftrightarrow & \quad   \sum_{\nu_N \in \{[d_{\rm S}^4-1]\cup \{0\}\}} \delta_{\nu_N,\nu'_N}  \Upsilon_{\bm{\nu}} = \bbra{\id_{\rm SE} } \mathbb P_{\nu'_{N}}\mathcal U_{N-1}\mathbb P_{\nu_{N-1}}\mathcal U_{N-2}\cdots \mathcal U_{1}\mathbb P_{\nu_1}\mathcal{U}_0\kket{\psi} \nonumber \\
 \Leftrightarrow & \quad    \Upsilon_{\bm{\nu}} = \bbra{\id_{ \rm SE} } \mathbb P_{\nu_{N}}\mathcal U_{N-1}\mathbb P_{\nu_{N-1}}\mathcal U_{N-2}\cdots \mathcal U_{1}\mathbb P_{\nu_1}\mathcal{U}_0\kket{\psi} \label{eq:discr_PT_N_steps_from_smoothing}. 
\end{align}
We now consider the formal cMPS
\begin{equation}
    \cmps{\sum_{k=0}^{N-1}\delta(t-\tfrac{k T}{N})\mathbb H_{k}}{\{\mathbb P_\nu\}_\nu}{\kketbra{\psi}{\id_{\rm SE}}}\equiv \ket{ \Phi^{\tilde{\mathbb H}}} \in \Gamma,
\label{eq:PCR_proof_cMPS_delta_Ham}
\end{equation}
which is up to the distribution-valued drift matrix, which we label $\tilde{\mathbb H}$, in the PCR.
Both $\tilde{\mathbb H}$ and $\psi$ depend on $N$ but for clarity we leave this dependence implicit. 
Writing
\begin{align}
    \tint &= \sum_{\substack{\bm{n}\in \mathbb N_0^{N}\\ \|\bm{n}\|_1=n}} \prod_{k \in [N]}  \int_{(k-1)\tfrac{T}{N}\leq t_{\sum_{i=1}^{k-1}n_i+1} \leq \cdots \leq   t_{\sum_{i=1}^k n_i} \leq k\tfrac{T}{N}}\mathrm{d}t_{\sum_{i=1}^{k-1}n_i+1}\cdots \mathrm{d}t_{\sum_{i=1}^{k}n_i} \nonumber \\
    &\equiv \sum_{\substack{\bm{n}\in \mathbb N_0^{N}\\ \|\bm{n}\|_1=n}} \prod_{k \in [N]}  \int_{I_{k}}\mathrm{D}^{n_k}t
\end{align}
and $M_{\tilde{\mathbb H}}(t, t') \coloneqq \mathcal{T} \e^{\int_{t'}^t {\rm d}s \tilde{\mathbb H}(s)}$, we have
\begin{multline}
    \ket{ \Phi^{\tilde{\mathbb H}}}= 
  \sum_{n=0}^\infty \sum_{\bm{\nu}}\sum_{\substack{\bm{n}\in \mathbb N_0^{N}\\ \|\bm{n}\|_1=n}} \prod_{k \in [N]}  \int_{I_{k}}\mathrm{D}^{n_k}t  \bbra{\id_{\mathrm{SE}}}M_{\tilde{\mathbb H}}(T,t_n)\mathbb P_{\nu_{n}}\cdots \mathbb P_{\nu_{1}}  M_{\tilde{\mathbb H}}(t_1,0)\kket{\psi}\psi^\dagger_{\nu_n}(t_n)\cdots\psi^\dagger_{\nu_1}(t_1) \ket{\Omega}  \\ =  \sum_{n=0}^\infty \sum_{\bm{\nu}}\sum_{\substack{\bm{n}\in \mathbb N_0^{N}\\ \|\bm{n}\|_1=n}}  \bbra{\id_{\rm SE} }\mathbb P_{\nu_{n}}\cdots\mathbb P_{\nu_{\sum_{i=1}^{N-1} n_i+1}}\Bigl(\prod_{j\in [N-1]} \mathcal U_{j}\mathbb P_{\nu_{\sum_{i=1}^j n_i}} \cdots \mathbb P_{\nu_{\sum_{i=1}^{j-1} n_i+1}}\Bigr) \mathcal{U}_0\kket{\psi}\times \\ \Bigl(\prod_{k \in [N]}  \int_{I_{k}}\mathrm{D}^{n_k}t\Bigr) \psi^\dagger_{\nu_n}(t_n)\cdots\psi^\dagger_{\nu_1}(t_1) \ket{\Omega}.
\end{multline}
We can rewrite the integrals as
\begin{equation}
      \int_{I_{k}
      }\mathrm{D}^{n_k}t \, \psi_{\nu_{\sum_{i=1}^k n_i}}^\dagger(t_{\sum_{i=1}^k n_i})\cdots \psi_{\nu_{\sum_{i=1}^{k-1} n_i+1}}^\dagger(t_{\sum_{i=1}^{k-1}n_i+1}) =  \frac{1}{n_k!}\Bigl(\frac{T}{N}\Bigr)^{n_k/2}
     \Bigl(\prod_{j=\sum_{i=1}^{k-1}n_i+1}^{\sum_{i=1}^k n_i}a_{\nu_j}^\dagger(k)\Bigr),
\end{equation}
and hence
\begin{multline}
    \ket{ \Phi^{\tilde{\mathbb H}}} =  
   \sum_{n=0}^\infty \sum_{\bm{\nu}}\sum_{\substack{\bm{n}\in \mathbb N_0^{N}\\ \|\bm{n}\|_1=n}} \Bigl(\frac{T}{N}\Bigr)^{\tfrac{n}{2}}
    \frac{\bbra{\id_{\rm SE} }\mathbb P_{\nu_{n}}\cdots\mathbb P_{\nu_{\sum_{i=1}^{N-1} n_i+1}}\Bigl(\prod_{j\in [N-1]} \mathcal U_{j}\mathbb P_{\nu_{\sum_{i=1}^j n_i}} \cdots \mathbb P_{\nu_{\sum_{i=1}^{j-1} n_i+1}}\Bigr) \mathcal{U}_0\kket{\psi}}{n_N!\cdots n_1!}\times \\ \times \Bigl(\prod_{k \in [N]}\prod_{j=\sum_{i=1}^{k-1}n_i+1}^{\sum_{i=1}^k n_i}a_{\nu_j}^\dagger(k)\Bigr) \ket{\Omega} .
    \label{eq:expl_form_delta_PCR}
\end{multline}
Therefore we see that $P_N \ket{ \Phi^{\tilde{\mathbb H}}}= \ket{ \Phi^{\tilde{\mathbb H}}}$. 

We proceed by establishing that this construction implies $P^{0,1}_N \ket{ \Phi^{\tilde{\mathbb H}}}= P^{0,1}_N\ket{\Psi}$.
The projector $ P^{0,1}_N$ annihilates all contributions in $\ket{ \Phi^{\tilde{\mathbb H}}}$ with more than one $a^\dagger(k)$ excitation in any interval $I_k$. We extend the summation index $\bm{\nu}\in [d_{\rm S}^4-1]^n$ in Eq.~(\ref{eq:expl_form_delta_PCR}) to a string $\bm{\tilde \nu}\in  \{[d_{\rm S}^4-1]\cup \{0\}\}^N$ by inserting the value zero $\tilde \nu_k = 0$ at those positions $k\in [N]$ where $n_k=0$.
For a given $\bm{n}$ the sum over $\tilde{\bm{\nu}}$ then varies all the non-zero elements of $\tilde{\bm{\nu}}$ in $[d_{\rm S}^4 - 1]$. 
Now
\begin{align}
    P^{0,1}_N \ket{ \Phi^{\tilde{\mathbb H}}}&= \sum_{n=0}^N\sum_{\substack{\bm{n}\in \{0,1\}^{N}\\ \|\bm{n}\|_1=n}}\sum_{\bm{\tilde \nu}} \Bigl(\frac{T}{N}\Bigr)^{\tfrac n 2}
    \bbra{\id_{\rm SE} }\mathbb P_{\tilde \nu_{N}}\Bigl(\prod_{j\in [N-1]} \mathcal U_{j}\mathbb P_{\tilde \nu_{j}}\Bigr) \mathcal U_0 \kket{\psi} \Bigl(\prod_{k \in \text{supp}(\bm{\tilde \nu})}a_{\tilde \nu_k}^\dagger(k)\Bigr) \ket{\Omega} \nonumber \\
     &= \sum_{n=0}^N\sum_{\substack{\bm{n}\in \{0,1\}^{N}\\ \|\bm{n}\|_1=n}}\sum_{\bm{\tilde \nu}} \Bigl(\frac{T}{N}\Bigr)^{\tfrac n 2}   \Upsilon_{\bm{ \tilde\nu}} \Bigl(\prod_{k \in \text{supp}(\bm{\tilde \nu})}a_{\tilde \nu_k}^\dagger(k)\Bigr) \ket{\Omega} \nonumber \\
     &= \underbrace{\sum_{n=0}^N\sum_{\substack{\bm{n}\in \{0,1\}^{N}\\ \|\bm{n}\|_1=n}}\sum_{\bm{\tilde \nu}} 
     \bigl(\prod_{k \in \text{supp}(\bm{\tilde \nu})}a_{\tilde \nu_k}^\dagger(k)\bigr) \ket{\Omega}\!\bra{\Omega}\bigl(\prod_{i \in \text{supp}(\bm{\tilde \nu})} a_{\tilde\nu_i}(i)\bigr)}_{=P_N^{(0,1)}}  P_N^{(0,1)} \ket{\Psi} \nonumber \\
 &= P_N^{(0,1)} \ket{\Psi},
\end{align}
where we used Eq.~(\ref{eq:discr_PT_N_steps_from_smoothing}) and the definition of $\Upsilon_{\bm \nu}$.

To obtain a Fock space vector in process-canonical representation  with a bounded generator, we consider a smoothed version of the distribution-valued drift matrix $\tilde{\mathbb H}(t) = \sum_{k=0}^{N-1}\delta(t-\tfrac{k T}{N})\mathbb H_{k} $. Given any smoothening parameter $\eta> 0$, we define
\begin{equation}
    \mathbb H_\eta(t) = \frac{1}{2\eta}\left(\bm1_{[0,2\eta]}(t)\mathbb H_0+\sum_{k=1}^{N-1}\bm1_{[\tfrac{kT}{N}-\eta,\tfrac{kT}{N}+\eta]}(t)\mathbb H_k\right)  ,
\end{equation}
and the corresponding PCR $\ket{\Phi^{\mathbb H_\eta}}:= \ket{\Phi[\mathbb H_\eta,\{\mathbb P_\nu\}_\nu,\kketbra{\psi}{\id_{SE}} ]}$. 
The distance of $\ket{\Phi^{\mathbb H_\eta}}$ to the positive and causal Fock space vector $\ket{\Psi}$ can now be bounded by a series of triangle inequalities and using the previously shown properties:
\begin{align}
    \|\ket{\Psi} - \ket{\Phi^{\mathbb H_\eta}}\| &\leq  \|\ket{\Psi} - P_N\ket{ \Phi^{\tilde{\mathbb H}}}\|+ \underbrace{ \|P_N\ket{ \Phi^{\tilde{\mathbb H}}}-  \ket{ \Phi^{\tilde{\mathbb H}}}\|}_{= 0}+ \|\ket{ \Phi^{\tilde{\mathbb H}}}- \ket{\Phi^{\mathbb H_\eta}}\| \nonumber \\
    &\leq   \|\ket{\Psi}- P_N^{(0,1)}\ket{\Psi} \|+ \underbrace{\|P_N^{(0,1)}(\ket{\Psi}- \ket{ \Phi^{\tilde{\mathbb H}}})\|}_{=0} +\|P_N^{(0,1)}\ket{ \Phi^{\tilde{\mathbb H}}}  - P_N\ket{ \Phi^{\tilde{\mathbb H}}}\|+  \|\ket{ \Phi^{\tilde{\mathbb H}}}- \ket{\Phi^{\mathbb H_\eta}}\| \nonumber \\
    &=  \|\ket{\Psi}- P_N^{(0,1)}\ket{\Psi} \| +\|(P_N^{(0,1)}   - P_N)\ket{ \Phi^{\tilde{\mathbb H}}}\|+  \|\ket{ \Phi^{\tilde{\mathbb H}}}- \ket{\Phi^{\mathbb H_\eta}}\| \label{eq:contr_bound_dist_PCR}. 
\end{align}
The following three lemmata provide a bound on each term in the final expression above, which can be made arbitrarily small by making $N$ and $\eta$ sufficiently large and small respectively.
The lemmata are proven in Appendix~\cref{app:lemmata_proofs}.

The second term in~\cref{eq:contr_bound_dist_PCR} is controlled by the following.
\begin{lemma}[Arbitrarily small projector difference]
\label{lem:diff_proj_cMPS}
   For an arbitrary Fock vector $\ket{\psi}$ it holds that for any $\epsilon > 0$ there exists an $N_0 \in \mathbb N$ such that for all $N > N_0$ 
   \begin{equation}
        \|(P_N^{(0,1)}   - P_N)\ket{\psi}\| < \frac \epsilon 3.
   \end{equation}
\end{lemma}
The first term in Eq.~(\ref{eq:contr_bound_dist_PCR}) is controlled by the convergence of the projector $P_N^{(0,1)}$ to the identity.
\begin{lemma}[Convergence of projector $P_N^{(0,1)}$]
\label{lem:conv_PN01_proj}
    The sequence of projectors  $\bigl(P_N^{(0,1)}\bigr)_{N\in \mathbb N}$ converges in strong operator topology to the identity: For any vector $\ket{\psi} \in \Gamma$ it holds that
\begin{equation}
    \lim_{N \to \infty} \|P_N^{(0,1)}\ket{\psi}- \ket{\psi}\| = 0.
\end{equation} 
\end{lemma}
Lemma 2 is in fact a corollary of Lemma 1.

Finally, the last term is dealt with by the following.
\begin{lemma}[Closeness of relevant PCRs]
\label{lem:smooth_PCR}
    For a fixed  $N \in \mathbb N$ and any $\epsilon > 0$, there exists an $\eta_{\epsilon,N} > 0$ such that 
    \begin{equation}
        \|\ket{ \Phi^{\tilde{\mathbb H}}}- \ket{\Phi^{\mathbb H_{\eta_{\epsilon,N}}}}\| < \frac \epsilon 3.
    \end{equation}
\end{lemma}
By the above lemmata, for any $\epsilon > 0$, there exists a $N_\epsilon \in \mathbb N$ and an $\eta_{\epsilon,N_{\epsilon}} > 0$ such that 
\begin{equation}
     \|\ket{\Psi} - \ket{\Phi^{\mathbb H_{\eta_{\epsilon,N_{\epsilon}}}}}\| < \tfrac{\epsilon}{3} + \tfrac{\epsilon}{3} + \tfrac{\epsilon}{3} = \epsilon.
\end{equation}
This proves Theorem 1 in the main text. 
\end{proof}

\section{Proofs of lemmata from PCR theorem proof}
\label{app:lemmata_proofs}

\subsection{Proof of Lemma~\ref{lem:diff_proj_cMPS} and Lemma~\ref{lem:conv_PN01_proj}}
\begin{proof}[\unskip\nopunct]
As we will show in more detail below, Lemma~\ref{lem:conv_PN01_proj} essentially follows from two intuitions (and their generalisation): that (i) the Lebesgue measure on $[0,T]^2$ of the set $\{(t,t') \in  [0,T]^2 : t = t'\}$ is zero, and (ii) any continuous function on $[0,T]$ can be arbitrarily well approximated by piecewise constant functions. Fact (i) will prove Lemma~\ref{lem:diff_proj_cMPS}. 
We will be using the canonical unitary identification of the single-particle Hilbert space $\mathcal H = \oplus_{j=1}^q L^2([0,T])$, where $q = d_{\rm S}^4-1$, with the tensor product
\begin{align}
    \mathcal H  &\simeq L^2([0,T])\otimes \mathbb C^q \nonumber \\
    (f_\nu)_{\nu = 1}^q &\mapsto \sum_{\nu = 1}^q f_\nu \otimes v_\nu,
\end{align}
where $\{v_\nu\}_\nu$ is an orthonormal basis of $\mathbb C^q$. We denote
\begin{equation}
    e_k = \sqrt{\tfrac{N}{T}} \bm{1}_{I_k}, \quad k \in [N]
\end{equation}
such that $\{e_k\}_{k \in [N]}$ is an orthonormal set in $L^2([0,T])$. Let
\begin{equation}
    \Pi_N := \sum_{k \in [N]} \ket{e_k}\bra{e_k} \otimes \id_{\mathbb C^q}
\end{equation}
be the single-particle orthogonal projector onto the step functions constant on the cells $I_k$. We will call $\ket{e_k}$ tensored with the particle species basis states the \textit{lattice basis states}. The Fock space projector $P_N$ in \cref{eq:projector_P_N} is the second quantisation of $\Pi_N$ (see \cite[Section X.7]{reedsimon_vol2}). This means that $P_N$ acts on the $n$-particle subspace of $\Gamma_T^{(q)}$ as $\Pi_N^{\otimes n}$.  

Because $P_N^{(0,1)}$ is an orthogonal projector for every $N$, it is a bounded operator and it satisfies the uniform bound $\|P_N^{(0,1)}\|=1$ on its operator norm.
Therefore, to prove the lemma it is enough to show $\lim_{N \to \infty}\|P_N^{(0,1)}\ket{\phi} - \ket{\phi}\|= 0$ for $\ket{\phi} \in \mathcal D$ being an element in a dense subset $\mathcal D \subset \Gamma$.
Let $\mathcal D$ be the set of finite-particle vectors, i.e. those vectors for which there exists $n_0$ such that all $n$-particle amplitudes with $n>n_0$ vanish.
This set is dense in $\Gamma$.
As $P_N^{(0,1)}$ does not mix the particle sectors, it is therefore enough to show $\lim_{N \to \infty}\|P_N^{(0,1)}\ket{\phi_n} - \ket{\phi_n}\|= 0$ for all elements $\ket{\phi_n}$ of the $n$-particle sector.
In fact, this is the case also for the orthogonal projectors $P_N$ and $P_N - P_N^{(0,1)}$ (this is an orthogonal projector since $P_N^{(0,1)} \le P_N$), as will be useful below.

We will denote by  $\phi_n \in L^2([0,T]^n)\otimes (\mathbb C^{q})^{\otimes n} $ the symmetric $n$-particle wavefunction of $\ket{\phi_n}$.
By the triangle inequality we get that
\begin{equation}
    \|P_N^{(0,1)}\ket{\phi_n} - \ket{\phi_n}\| \leq \underbrace{\|(P_N -P_N^{(0,1)})\ket{\phi_n}\|}_{(i)} + \underbrace{\|\ket{\phi_n} - P_N\ket{\phi_n}\|}_{(ii)}.
    \label{eq:proof_lem3_tot_exp}
\end{equation}
(i) We define the collision set 
\begin{equation}
    D_N := \bigl\{(t_1,\dots,t_n)\in[0,T]^n \;:\; \exists\, i<j,\ \text{ such that } t_i,t_j\in I_k \text{ for some } k\in[N] \bigr\}. 
\end{equation}
The lattice basis states 
with $n$ excitations are labelled by $(\ell_1,\cdots,\ell_n)$ with $1 \leq \ell_1,\cdots,\ell_n \leq N $. 
The wavefunctions of the lattice basis states having all $\ell_j$'s distinct are supported on the complement $D_N^c$, while if some labels are repeated then the wavefunction is supported on $D_N$. Therefore, when expanding $\Pi_N^{\otimes n}\phi_n$ in the lattice basis a multiplication by the indicator function $\bm{1}_{D_N}$ precisely deletes the singly-occupied components in the expansion and preserves the ones which are multiply-occupied. But this is exactly the action of $P_N -P_N^{(0,1)}$ in every finite-particle sector. Hence, 
\begin{equation}
    \|(P_N -P_N^{(0,1)})\ket{\phi_n}\|^2 = \|\bm{1}_{D_N}\Pi_N^{\otimes n}\phi_n\|^2 . 
\end{equation}
Both operators on the right hand side act as the identity on $(\mathbb C^q)^{\otimes n}$, the particle species part of the $n$-particle subspace. 
Furthermore, the two operators $\bm{1}_{D_N}$ and $\Pi_N^{\otimes n}$ commute.
To see this, consider their action on the box decomposition $L^2([0,T]^n)=\bigoplus_{\boldsymbol \ell\in[N]^n}L^2(B_{\bm{\ell}})$, where $B_{\bm{\ell}}:=I_{\ell_1}\times\cdots\times I_{\ell_n}$: the operator $\Pi_N^{\otimes n}$
preserves each summand, acting on it as the rank-one projection onto the constant functions on the block, while
$\mathbf 1_{D_N}$ acts on each summand as a \emph{scalar} ($1$ if $\bm{\ell}$ has a repeated
index, $0$ otherwise). A block-scalar
operator commutes with any block-preserving one, so
$\mathbf 1_{D_N}\Pi_N^{\otimes n}=\Pi_N^{\otimes n}\mathbf 1_{D_N}$.
Now, since $\Pi_N^{\otimes n}$ is a projector, we get
\begin{align}
     \|(P_N -P_N^{(0,1)})\ket{\phi_n}\|^2 &= \|\Pi_N^{\otimes n}\bm{1}_{D_N}\phi_n\|^2 \nonumber \\
     &\leq \|\Pi_N^{\otimes n}\| \|\bm{1}_{D_N}\phi_n\|^2 \nonumber \\
     &= \frac{1}{n!}  \int_{D_N}\mathrm{d}^n\bm{t}\|\phi_n(t_1,\cdots,t_n)\|_2^2 ,
\end{align}
where $\|\cdot\|_2$ refers to the Frobenius norm on the particle species part of the $n$-particle subspace $(\mathbb C^q)^{\otimes n}$.
As $N \to \infty$, the Lebesgue measure of the collision set vanishes. To see this, consider 
\begin{align}
    \int_{D_N}\mathrm{d}^n\bm{t} &= \underbrace{\bigl(N^n -\frac{N!}{(N-n)!}\bigr) }_{\text{number of }\bm{\ell}\in [N]^n \text{with repeated indices}}\left(\frac{T}{N}\right)^n \nonumber  \\
   & = T^n \bigl(1 - \prod_{j=0}^{n-1}(1-\tfrac{j}{N})\bigr) \nonumber \\
   &= T^n\left(\sum_{j=0}^{n-1}\tfrac{j}{N}+ \mathcal O(N^{-2})\right).
\end{align}
Therefore, there exists an $N_0$ such that for all $N \geq N_0$ it holds that
\begin{equation}
    \int_{D_N}\mathrm{d}^n\bm{t} \leq \binom{n}{2}\tfrac{T^n}{N},
\end{equation} 
which means that the measure vanishes in the $N \to \infty$ limit. Now, since $\|\phi_n(\cdot)\|_2^2 \in L^1([0,T]^n)$ and $\|\phi_n(t_1,\cdots,t_n)\|_2^2 \ge 0$, absolute continuity of the Lebesgue integral together with $\int_{D_N}\mathrm{d}^n\bm{t} \to 0$ gives that 
\begin{equation}
    \|(P_N -P_N^{(0,1)})\ket{\phi_n}\| \xrightarrow[N\to\infty]{}\; 0.
      \label{eq:proof_lem3_conv_i}
\end{equation}
Because $P_N -P_N^{(0,1)}$ is an orthogonal projector which does not mix particle spaces this proves Lemma~\ref{lem:diff_proj_cMPS}.

(ii) Let $E_N \coloneqq \sum_{k\in[N]} \ket{e_k}\!\bra{e_k}$.
We first show that $E_N \to \id$ strongly on $L^2([0,T])$. We have for any $t \in [0,T]$ that 
\begin{equation}
    |(E_N f)(t)-f(t)| \leq \sup_{s \in I_k}|f(s)-f(t)| 
\end{equation}
where $k$ is such that $t \in I_k$. The right hand side is known as the \textit{modulus of continuity}. For a continuous function $f \in C([0,T])$ it holds that $ |(E_N f)(t)-f(t)| \to 0$ for any $t$ and hence also $\|E_N f -f \| \to 0$ in $L^2$ norm. Since $E_N$ is a projector, its operator norm satisfies $\|E_N\| = 1$ for all $N$.
From $C([0,T])$ being dense in $L^2([0,T])$ it hence follows that $E_N \to \id $ strongly on all of $L^2([0,T])$. 
Consequently, also $\Pi_N \to \id$ in strong operator topology on $\mathcal H$.
Since $\|\ket{\phi_n}-P_N \ket{\phi_n}\| =\|\phi_n - \Pi_N^{\otimes n}\phi_n\|$, it remains to show that this convergence is lifted to the $n$-fold tensor power of $\mathcal H$. 
Again, we do this by showing that $\Pi_N^{\otimes n} \to \id^{\otimes n}$ strongly on a dense subspace of $\mathcal H^{\otimes n}$, which we choose to be the set of finite linear combinations of product vectors $f_1 \otimes \cdots \otimes f_n$.
For such a product vector we obtain, using a telescoping sum
\begin{align}
    \|(\Pi_N^{\otimes n} - \id^{\otimes n})f_1 \otimes \cdots \otimes f_n\| &= \| \sum_{j=1}^{n}\bigl(\id^{\otimes(j-1)}\otimes(\Pi_N-\id)\otimes\Pi_N^{\otimes(n-j)}\bigr)f_1 \otimes \cdots \otimes f_n\| \nonumber \\
    &\leq \sum_{j=1}^{n} \Bigl(\prod_{i\neq j}\|f_i\|\Bigr)\,\|(\Pi_N-\id)f_j\|
  \xrightarrow[N\to\infty]{}\;0,
\end{align}
where the convergence follows from $\|f_j\|< \infty$. As before, this convergence extends to all of $\mathcal H^{\otimes n}$ and we get
\begin{equation}
    \|\ket{\phi_n}-P_N \ket{\phi_n}\|\xrightarrow[N\to\infty]{}\;0.
    \label{eq:proof_lem3_conv_ii}
\end{equation}

Combining \cref{eq:proof_lem3_tot_exp,eq:proof_lem3_conv_i,eq:proof_lem3_conv_ii} concludes the proof of Lemma~\ref{lem:conv_PN01_proj}. 
\end{proof}

\subsection{Proof of Lemma~\ref{lem:smooth_PCR}}
\begin{proof}[\unskip\nopunct]
We will show that $\|\ket{ \Phi^{\tilde{\mathbb H}}}- \ket{\Phi^{\mathbb H_{\eta}}}\|$ for $\eta \to 0^+$ is $\mathcal O(\sqrt{\eta})$. We fix $N$, write $W_k^{\eta}:= [\tfrac{kT}{N}-\eta,\tfrac{kT}{N}+\eta]$ with $k \in [N-1]$ and $W_0^{\eta} = [0,2\eta]$,  and assume that $\eta < \tfrac{T}{3N}$ such that the small windows are mutually disjoint. With
$\tilde{\mathbb H}(t)=\sum_{k=0}^{N-1}\delta \bigl(t-\tfrac{kT}{N}\bigr)\mathbb H_k$ and
$\mathbb H_\eta(t)=\tfrac{1}{2\eta}\sum_{k=0}^{N-1}\bm{1}_{W_k^\eta}(t)\,\mathbb H_k$, the associated cMPSs are
\begin{equation}
     \ket{\Phi^{H}}
  =
  \llangle \id_{\rm SE}|\;
  \mathcal T\exp\Bigl(\int_0^T\!\mathrm{d} t\;\bigl[H(t)+\textstyle\sum_{\nu}\mathbb P_\nu\,\psi_\nu^\dagger(t)\bigr]\Bigr)
  \;|\psi\rrangle\ket{\Omega},
  \qquad H\in\{\tilde{\mathbb H},\,\mathbb H_\eta\}.
\end{equation}
The particle wavefunctions of the two cMPSs are
\begin{equation}
    \varphi^{H,(n)}_{\bm{\nu}}(\bm{t})
  =\bbra{\id_{\rm SE}} M_H(T,t_n)\mathbb P_{\nu_n}M_H(t_n,t_{n-1})\cdots\mathbb P_{\nu_1}\,M_H(t_1,0)\,
  \kket{\psi},
  \qquad H\in\{\tilde{\mathbb H},\mathbb H_\eta\},
\end{equation}
where $M_H(t,t')= \mathcal T \exp[\int_{t'}^t \mathrm{d}s H(s)]$. Therefore,
\begin{align}
     \|\ket{\Phi^{\mathbb H_\eta}}-\ket{\Phi^{\tilde{\mathbb H}}}\|^2 &= \sum_{n=0}^\infty \sum_{\bm{\nu}} \tint \bigl| \varphi^{\mathbb H_\eta,(n)}_{\bm{\nu}}(\bm{t})- \varphi^{\tilde{\mathbb H},(n)}_{\bm{\nu}}(\bm{t}) \bigr|^2.
 \end{align}
 Since the drift matrix of both cMPSs is supported only in $\bigcup_k W_k^\eta$ and the jump matrices are the same for both cMPSs, the integrand on the right-hand side in the equation above has support only on the set $D_\eta^{(n)}:= \bigl\{\bm{t} : \exists j \in [n], t_j \in \bigcup_k W_k^\eta \bigr\}$. Then,
 \begin{align}
     \|\ket{\Phi^{\mathbb H_\eta}}-\ket{\Phi^{\tilde{\mathbb H}}}\|^2 &= \sum_{n=0}^\infty \sum_{\bm{\nu}} \int_{D_\eta^{(n)} \cap \{0<t_1 < \cdots t_n<T\}} \mathrm{d}^n \bm{t} \bigl| \varphi^{\mathbb H_\eta,(n)}_{\bm{\nu}}(\bm{t})- \varphi^{\tilde{\mathbb H},(n)}_{\bm{\nu}}(\bm{t}) \bigr|^2 \nonumber \\
     &\leq \sum_{n=0}^\infty \sum_{\bm{\nu}} \bigl( \sup_{\bm{t}\in D_\eta^{(n)} \cap \{0<t_1 < \cdots t_n<T\}}\bigl| \varphi^{\mathbb H_\eta,(n)}_{\bm{\nu}}(\bm{t})- \varphi^{\tilde{\mathbb H},(n)}_{\bm{\nu}}(\bm{t}) \bigr|^2\bigr) \int_{D_\eta^{(n)} \cap \{0<t_1 < \cdots t_n<T\}} \mathrm{d}^n \bm{t}. 
 \end{align}
The pointwise difference in the amplitude functions is bounded by 
\begin{align}
    \sup_{\bm{t}\in D_\eta^{(n)} \cap \{0<t_1 < \cdots t_n<T\}}\bigl| \varphi^{\mathbb H_\eta,(n)}_{\bm{\nu}}(\bm{t})- \varphi^{\tilde{\mathbb H},(n)}_{\bm{\nu}}(\bm{t}) \bigr|^2 &\leq \bigl(2 \|\kket{\psi}\bbra{\id_{\rm SE}}\|_1 \|\mathbb P_{\nu_n}\|_\infty \cdots \|\mathbb P_{\nu_1}\|_\infty \bigr)^2 \nonumber \\
    &= 4 \, \|\kket{\psi}\| ^2 \, \|\bbra{\id_{\rm SE}}\|^2 \, d_{\rm S}^{-2n} \nonumber \\
    &= 4 d_{\rm S}^{2N-2n} ,
\end{align}
where we used that both $\mathbb H_\eta$ and $\tilde{\mathbb H}$ are generators of unitary evolution and the normalisation of $\kket{\psi}$. For the remaining integral we use that $D_\eta^{(n)} = \bigcup_{j=1}^{n} \{\bm{t}: t_j \in \bigcup_k W_k^\eta \}$ and a union bound then gives
\begin{align}
    \int_{D_\eta^{(n)} \cap \{0<t_1 < \cdots t_n<T\}} \mathrm{d}^n \bm{t} &= \frac{1}{n!}\int_{D_\eta^{(n)}} \mathrm{d}^n \bm{t} \nonumber  \\
    &\leq \frac{1}{n!} \sum_{j=1}^n \int_{\{\bm{t}: t_j \in \bigcup_k W_k^\eta \}}  \mathrm{d}^n \bm{t} \nonumber \\
    &= \frac{1}{n!} n 2\eta NT^{n-1}
\end{align}
Overall, we get 
\begin{align}
    \|\ket{\Phi^{\mathbb H_\eta}}-\ket{\Phi^{\tilde{\mathbb H}}}\|^2 &\leq \sum_{n=0}^\infty 4 (d_{\rm S}^4-1)^n d_{\rm S}^{2N-2n}\frac{1}{n!} n 2\eta N T^{n-1} \nonumber \\
    &\leq \eta \, 8N d_{\rm S}^{2+2N} e^{T d_{\rm S}^2}
\end{align}
which proves the lemma.
\end{proof}

\section{Choi duality in the continuum}

In this appendix we present the details on the Choi duality in the continuum and the proof of Theorem 2 in the main text. We recall that the difficulty in formulating a Choi duality in the continuum is that the Fock space 
\begin{equation}
    \Gamma = \Gamma \bigl(L^2([0, T])^{\oplus (d_{\rm S}^4 - 1)}\bigr) = \Gamma\bigl(L^2([0, T])\bigr)^{\otimes (d_{\rm S}^4 - 1)}\label{eq:rel_Fock_space_det}
\end{equation}
does not naturally factor into forward- and backward-contour Fock spaces because it also contains cross species. In particular, labelling the particle species by $\nu = (\mu, \mu')$, where $\nu = 1, \dots, d_{\rm S}^4 - 1$ and $\mu, \mu' = 0, \dots, d_{\rm S}^2 - 1$ with the restriction $(\mu, \mu') \neq (0, 0)$, the part of the Fock space associated to the \textit{cross} species is the problematic part.  
To remedy this issue we have to treat the effect $\ket{\mathcal I}$ of the continuous instrument and the cPT $\ket{\Upsilon}$ in the contraction $\braket{\mathcal I|\Upsilon}$ differently.
We seek two distinct maps $\mathcal P: \ket{\Upsilon} \to \hat\Upsilon$ and $\mathcal V: \ket{\mathcal I} \to \check{\mathcal I} $ which act on cPTs and instrument effects, respectively. 
The aim is to transfer the contraction $\braket{\mathcal I|\Upsilon}$ to the operator space in a sensible way: 
\begin{equation}
   \braket{\mathcal I | \Upsilon}=  \braket{\check{\mathcal{I}},\hat \Upsilon_T}_{\rm ker}
 \label{eq:Choi_dual_inner_prod}.
\end{equation}

The rest of this Appendix contains the following. In the first section, we define the map $\mathcal{P}$ and a recovery map $\mathcal{R}$, showing that it inverts the action of $\mathcal{P}$ on all cPTs.
This is the content of the first part of Theorem 2 in the main text. 
In the second section we define the map $\mathcal V$ and elaborate how, together with $\mathcal P$ and $\mathcal R$, it provides a sensible notion of the Choi duality for cPTs, giving a precise meaning to $\braket{\cdot, \cdot}_{\text{ker}}$ in~\cref{eq:Choi_dual_inner_prod}.
We will find that when $\ket{\mathcal I}$ satisfies certain conditions, $\check{\mathcal{I}}$ is a bounded operator on $\Gamma_T^{(d_{\rm S}^2-1)} \equiv \gamma$ and then~\cref{eq:Choi_dual_inner_prod} becomes the Hilbert--Schmidt inner product.
If this is not the case, $\check{\mathcal{I}}$ is a sesquilinear form on $\gamma$.

\subsection{Choi operator of a cPT}

We specify the map $\mathcal{P}: \Gamma \rightarrow \text{HS}(\gamma)$, where $\text{HS}$ is the space of Hilbert--Schmidt operators, by defining the resulting operator in terms of its \textit{integration kernel}.
Given an integration kernel $f^{(n; n')}_{\bm{\mu}; \bm{\mu'}}(\bm{t}; \bm{t'})$ the associated operator is defined via
\begin{equation}
    \sum_{n,n'=0}^\infty \sum_{\bm{\mu} \in [d_{\rm S}^2 - 1]^n} \sum_{\bm{\mu'} \in [d_{\rm S}^2 - 1]^{n'}} \int D^n\bm{t} \int D^{n'}\bm{t'} f^{(n; n')}_{\bm{\mu}; \bm{\mu'}}(\bm{t}; \bm{t'}) \phi_{\mu_1}^\dag(t_1) \dots \phi_{\mu_n}^\dag(t_n) \ket{\tilde \Omega}\!\bra{\tilde{\Omega}} \phi_{\mu'_1}(t_1') \dots \phi_{\mu'_{n'}}(t_{n'}').
\end{equation}
Here, we denote with $\ket{\tilde \Omega}$ and $\phi_{\mu}(t)$ the vacuum state and the field operators in the Fock space $\gamma$, respectively.
In general, integration kernels can be distribution-valued: for instance, the integration kernel of the identity operator is
\begin{equation}
    f^{(n;n')}_{\bm{\mu};\bm{\mu'}}(\bm{t}; \bm{t'}) \coloneqq 
    \begin{cases}
        \delta_{\bm{\mu}, \bm{\mu'}} \prod_{i=1}^n \delta(t_i - t'_i) & n=n'\\
        0 & \text{otherwise}.
    \end{cases}
\end{equation}
However, Hilbert--Schmidt operators have $L^2$-function-valued integration kernels.
We define the integration kernel of $\mathcal{P}(\ket{\Phi})$ in terms of wavefunctions of $\ket{\Phi}$ by
\begin{equation}
    \mathcal{P}(\ket{\Phi})^{(n; n')}_{\bm{\mu}; \bm{\mu'}}(\bm{t}; \bm{t'}) \coloneqq \Phi^{n + n'}_{(\bm{\mu}, \bm{0}), (\bm{0}, \bm{\mu'})}(\bm{t}, \bm{t'}),
    \label{eq:def_P_map_CHoi}
\end{equation}
for all $n, n' \in \mathbb N_0$ and $\bm{\mu}\in [d_{\rm S}^2-1]^n, \bm{\mu'}\in [d_{\rm S}^2-1]^{n'}$. This places  the output of $\mathcal P$ into $\text{HS}(\gamma)$. Note that this definition discards the wavefunctions corresponding to cross species particles in $\Gamma$. 

Even though $\mathcal{P}$ is clearly not invertible, we will define a \textit{recovery} operator 
\begin{equation}
    \label{eq:def_P_tilde_choi}
    \mathcal{R}: \mathcal S \rightarrow \Gamma,
\end{equation}
where $\mathcal S$ is a subspace (specified in a moment) of the space of trace-class operators $\text{TC}(\gamma)$, which will invert the action of $\mathcal{P}$ on all cPTs.
To define this operator, we will need the singular value decomposition (SVD):
the kernel of any trace-class operator $A$ can be written as
\begin{equation}
    A^{(n;n')}_{\bm{\mu};\bm{\mu'}}(\bm{t};\bm{t'}) = \sum_k\sigma_k u^{(k,n)}_{\bm\mu}(\bm t)\left(v^{(k,n')}_{\bm{\mu'}}(\bm{t'})\right)^*,
\end{equation}
where $v^{(k,n)}_{\bm{\mu}} \in L^2([0,T]^n)$, $u^{(k,n')}_{\bm{\mu'}} \in L^2([0,T]^{n'})$ and $\sum_k\sigma_k <\infty$.
By the Cauchy--Schwarz inequality, the summability of $(\sigma_k)_k$ implies that the sum converges absolutely almost everywhere.
Using the SVD, we define the wavefunctions of $\mathcal{R}(A)$ for any trace class $A$ by
\begin{equation}\label{eq:P tilde def}
    \mathcal{R}(A)^{(n + n^\ell + n^r)}_{(\bm{\mu}, \bm{\mu'}), (\bm{\mu^\ell}, \bm{0}), (\bm{0}, \bm{\mu^r})}(\bm{t}, \bm{t^\ell}, \bm{t^r}) \coloneqq \sum_k\sigma_k u_{\bm\mu,\bm{\mu^\ell}}^{(k,n+n^\ell)}(\bm t,\bm{t^\ell})\left(v_{\bm{\mu'},\bm{\mu^r}}^{(k,n+n^r)}(\bm t,\bm{t^r})\right)^*,
\end{equation}
where $n, n^\ell, n^r \in \mathbb{N}_0$, $\bm{\mu}, \bm{\mu'} \in [d_{\rm S}^2 - 1]^n$, $\bm{\mu^\ell}\in[d_{\rm S}^2 - 1]^{n^\ell}$ and $\bm{\mu^r}\in[d_{\rm S}^2 - 1]^{n^r}$ as well as $\bm{t}\in[0,T]^n$, $\bm{t^\ell}\in[0,T]^{n^\ell}$ and $\bm{t^r} \in [0, T]^{n^r}$.
As the LHS is not guaranteed to be square-integrable in general for any $A \in  \text{TC}(\gamma)$, we define the domain $\mathcal S$ of $\mathcal{R}$ as exactly those trace-class operators for which~\cref{eq:P tilde def} are the wavefunctions of a vector in $\Gamma$.

Suppose that a trace-class operator $A$ happens to have continuous function-valued integration kernel.
For such an operator~\cref{eq:P tilde def} simplifies
\begin{equation}\label{eq:P tilde cont}
    \mathcal{R}(A)^{(n + n^\ell + n^r)}_{(\bm{\mu}, \bm{\mu'}), (\bm{\mu^\ell}, \bm{0}), (\bm{0}, \bm{\mu^r})}(\bm{t}, \bm{t^\ell}, \bm{t^r}) = A^{(n+n^\ell; n+ n^r)}_{\bm{\mu}, \bm{\mu^\ell}; \bm{\mu'}, \bm{\mu^r}}(\bm{t}, \bm{t^\ell}; \bm{t}, \bm{t^r}).
\end{equation}
The reason why we needed the SVD to define the action of $\mathcal{R}$ on a general trace class $A$ is that the RHS of~\cref{eq:P tilde cont} would not be well defined under $L^2$ equivalence of the kernel functions.

Now, we will see that $\mathcal{R}$ inverts the action of $\mathcal{P}$ on all PCR $\ket{\Upsilon}\in\Gamma$.
First, we compute their trace.
\begin{lemma}\label{lem:trace of PCR choi}
    Any $\ket{\Upsilon}\in\Gamma$ in the PCR satisfies $\tr[\mathcal{P}(\ket{\Upsilon})] = \e^{T\left(d_{\rm S} - \frac1{d_{\rm S}}\right)}$.
\end{lemma}
\begin{proof}
    Since the Paulis are an orthonormal basis we have that
    \begin{equation}
        \sum_{\mu \in [d_{\rm S}^2 - 1]} P_\mu \otimes P_\mu^* = \kket{\id_{\rm S}}\!\bbra{\id_{\rm S}} - \frac1{d_{\rm S}} \id_{\rm S}.
    \end{equation}
    Hence, we get
    \begin{align}
        \sum_{\bm{\mu}\in[d_{\rm S}^2 - 1]^n} \Upsilon^{(n)}_{(\bm{\mu}, \bm{\mu})}(\bm{t}) &= \bbra{\id_{\rm SE}} \mathcal{T} \e^{\int_{t_n}^T \mathcal{H} {\rm d}t} \left(\kket{\id_{\rm S}}\!\bbra{\id_{\rm S}} - \frac1{d_{\rm S}} \id_{\rm S}\right) \mathcal{T} \e^{\int_{t_{n-1}}^{t_n} \mathcal{H} {\rm d}t} \dots \left(\kket{\id_{\rm S}}\!\bbra{\id_{\rm S}} - \frac1{d_{\rm S}} \id_{\rm S} \right)\e^{\int_0^{t_1}\mathcal{H}{\rm d}t} \kket{\psi} \nonumber\\
                                                                                            &= \sum_{k=0}^n \binom{n}{k} d_{\rm S}^k \left(-\frac1{d_{\rm S}}\right)^{n-k} \nonumber\\
                                                                                            &= \left(d_{\rm S} - \frac1{d_{\rm S}}\right)^n,
    \end{align}
    which follows because $\mathcal{T} \e^{\int_{t_i}^{t_j} \mathcal{H} {\rm d}t}$, $\frac1{d_{\rm S}} \kket{\id_{\rm S}}\!\bbra{\id_{\rm S}}$ and $\id_{\rm S}$ are CPTP maps.
    Finally, because the PCR has continuous wavefunctions, we can write the trace as
    \begin{align}
        \tr[\mathcal{P}(\ket{\Upsilon})] &= \sum_{n=0}^\infty \int D^n\bm{t} \sum_{\bm{\mu} \in [d_{\rm S}^2 - 1]^n} \Upsilon^{(n)}_{(\bm{\mu}, \bm{\mu})}(\bm{t}) \nonumber\\
                                         &= \sum_{n=0}^\infty \int D^n\bm{t} \left(d_{\rm S} - \frac1{d_{\rm S}}\right)^n \nonumber\\
                                         &= \sum_{n=0}^\infty \frac{T^n}{n!} \left(d_{\rm S} - \frac1{d_{\rm S}}\right)^n \nonumber\\
                                         &= \e^{T\left(d_{\rm S} - \frac1{d_{\rm S}}\right)}.
    \end{align}
\end{proof}
Since a cPT $\ket{\Upsilon}$ in the PCR has continuous wavefunctions, the kernel of $\mathcal{P}(\ket{\Upsilon})$ is also continuous (in fact it is a \emph{matrix product operator}~\cite{tjoaContinuousMatrixproductOperators2026}) and we can evaluate the wavefunctions of $\ket{\tilde{\Upsilon}} \coloneqq \mathcal{R}\circ\mathcal{P}(\ket{\Upsilon})$ using~\cref{eq:P tilde cont}, obtaining
\begin{equation}
    \tilde{\Upsilon}^{(n + n^\ell + n^r)}_{(\bm{\mu}, \bm{\mu'}), (\bm{\mu^\ell}, \bm{0}), (\bm{0}, \bm{\mu^r})}(\bm{t}, \bm{t^\ell}, \bm{t^r}) = \Upsilon^{(2n + n^\ell + n^r)}_{(\bm{\mu}, \bm{0}), (\bm{0}, \bm{\mu'}), (\bm{\mu^\ell}, \bm{0}), (\bm{0}, \bm{\mu^r})}(\bm{t}, \bm{t}, \bm{t^\ell}, \bm{t^r}).
\end{equation}
By evaluating the wavefunctions of a PCR and using $\mathbb{P}_{(\mu, \mu')} = \mathbb{P}_{(\mu, 0)} \mathbb{P}_{(0, \mu')}$ we notice that this in fact implies that $\ket{\tilde\Upsilon} = \ket{\Upsilon}$, showing that any PCR $\ket{\Upsilon}$ satisfies
\begin{equation}
\label{eq:conditions_Choi_dual_PPtilde}
    \mathcal{R}\circ\mathcal{P}(\ket{\Upsilon}) = \ket{\Upsilon}.
\end{equation}
Furthermore, $\mathcal{P}(\ket{\Upsilon})$ of a PCR is positive semidefinite, as follows from the fact that its integration kernel can be decomposed as
\begin{align}\label{eq:positive decomp PCR}
    \mathcal{P}(\ket{\Upsilon})^{(n;n')}_{\bm{\mu}; \bm{\mu'}}(\bm{t};\bm{t'}) &= \bbra{\id} \mathcal{T}e^{\int_{\tilde{t}_{n+n'}}^T\mathbb{H}{\rm d}t} \mathbb{P}_{\nu_{n+n'}} \dots \mathbb{P}_{\nu_1} \mathcal{T}e^{\int_0^{\tilde t_1} \mathbb{H}{\rm d}t} \kket{\psi}\\
                                                                               &= \sum_{k} \bra{k} \mathcal{T}e^{\int_{t_n}^T -i H {\rm d}t} P_{\mu_n} \dots P_{\mu_1} \mathcal{T}e^{\int_0^{t_1} -iH {\rm d}t} \ket{\psi} \bra{\psi} \mathcal{T}_- e^{\int_0^{t_1'} iH {\rm d}t} P_{\mu'_1} \dots \mathcal{T}_- P_{\mu'_{n'}} e^{\int_{t'_{n'}}^T iH {\rm d}t} \ket{k} \nonumber\\
                                                                               &= \sum_{k} f^{(k,n)}_{\bm{\mu}}(\bm{t}) \left(f^{(k,n')}_{\bm{\mu'}}(\bm{t'})\right)^*, \nonumber
\end{align}
where $\bm{\tilde{t}}$ is the ordered concatenation of $\bm{t}$ and $\bm{t'}$ and $\bm{\nu}$ is the concatenation of $(\bm{\mu}, \bm{0})$ and $(\bm{0}, \bm{\mu'})$ ordered according to the vector $\bm{\tilde{t}}$.

We will use these observations along with Theorem 1 to show that the property \cref{eq:conditions_Choi_dual_PPtilde} is satisfied by all cPTs. 
\begin{proposition}
    All continuous process tensors $\ket{\Upsilon}\in\Gamma$ satisfy $\mathcal{R}\circ\mathcal{P}(\ket{\Upsilon}) = \ket{\Upsilon}$.
    Furthermore, their Choi dual $\hat{\Upsilon} = \mathcal{P}(\ket{\Upsilon})$ is positive semidefinite and trace class.
    \label{prop:part_1_theorem2_main}
\end{proposition}
\begin{proof}
    By Theorem 1 in the main text, the space $\text{span}(\mathsf{PCR})$ spanned by the PCR Fock vectors (all finite linear combinations) is dense in the closed space $\text{span}(\mathsf{cPT})$.
    Since $\mathcal P: \Gamma \to \text{HS}(\gamma)$ is a bounded linear operator in the induced norm $\|\cdot\|_\Gamma \to \|\cdot\|_{\rm HS}$, the set $\mathcal P(\text{span}(\mathsf{PCR}))$ is dense in $\mathcal P(\text{span}(\mathsf{cPT})\,)$ in the $\text{HS}$ norm.
    We have just seen that for any PCR Fock vector $\ket{\Upsilon'}$ it holds that $\mathcal{R}\circ\mathcal{P}(\ket{\Upsilon'}) = \ket{\Upsilon'}$.
    Therefore, by linearity, the operator $\mathcal{R}$ has unit operator norm $\|\cdot\|_{\rm HS} \to \|\cdot\|_\Gamma$ on $\mathcal P(\text{span}(\mathsf{PCR}))$.
    Consider the restriction $\widetilde{\mathcal{R}}$ of $\mathcal{R}$ to the domain  $\mathcal{P}(\Span(\mathsf{PCR}))$.
    By the bounded linear transformation theorem there is a unique linear extension $\hat{\mathcal{R}}$ of $\widetilde{\mathcal{R}}$ to $\text{span}(\mathsf{cPT})$, the action of which on $\hat \Upsilon \in \Span(\mathsf{cPT})$ is defined by any sequence $(\hat{\Upsilon}_k \in \mathcal{P}(\Span(\mathsf{PCR})))_k$ converging to $\hat{\Upsilon}$ as
    \begin{equation}
        \hat{\mathcal{R}}(\hat{\Upsilon})
        \coloneqq \lim_{k\rightarrow\infty} \widetilde{\mathcal{R}}(\hat{\Upsilon}_k).
    \end{equation}
    As $\mathcal{R}$ is a linear extension of $\widetilde{\mathcal{R}}$ to $\Span(\mathsf{cPT})$ and $\hat{\mathcal{R}}$ is unique, we have $\mathcal{R} = \hat{\mathcal{R}}$.
    By Theorem 1 in the main text, take for any cPT $\ket{\Upsilon} \in \Gamma$ a sequence $(\ket{\Upsilon_k})_k$ where each $\ket{\Upsilon_k}\in\mathsf{PCR}$, such that $\lim_{k\rightarrow \infty} \ket{\Upsilon_k} = \ket{\Upsilon}$.
    Hence
    \begin{equation}
        \mathcal{R}\circ\mathcal{P}(\ket{\Upsilon}) = \hat{\mathcal{R}}\circ\mathcal{P}(\ket{\Upsilon}) \coloneqq \lim_{k\rightarrow\infty} \widetilde{\mathcal{R}}\circ\mathcal{P}(\ket{\Upsilon_k}) = \lim_{k\rightarrow\infty} \ket{\Upsilon_k} = \ket{\Upsilon}.
    \end{equation}
    As we have already seen that $\mathcal{P}(\ket{\Upsilon_k})$ is positive semidefinite for all $k$, the Choi matrix $\hat{\Upsilon} = \mathcal{P}(\ket{\Upsilon})$ is too (as convergence in HS-norm implies convergence in the weak operator topology).
    Furthermore, for any orthonormal basis $\ket{e_i}\in\gamma$ we get
    \begin{align}
        \tr[\hat{\Upsilon}] &= \sum_i \braket{e_i| \mathcal{P}(\ket{\Upsilon}) e_i} \nonumber\\
                    &= \sum_i \lim_{k \rightarrow \infty} \braket{e_i| \mathcal{P}(\ket{\Upsilon_k})e_i} \nonumber\\
                                         &\le \liminf_{k\rightarrow\infty} \sum_i \braket{e_i| \mathcal{P}(\ket{\Upsilon_k})e_i} \nonumber\\
                                         &= \e^{T\left(d_{\rm S} - \frac1{d_{\rm S}}\right)},
    \end{align}
    where the third line follows by Fatou's lemma as $\mathcal{P}(\ket{\Upsilon_k})\ge0$ and the last line from~\cref{lem:trace of PCR choi}.
    Therefore, $\mathcal{P}(\ket{\Upsilon})$ is trace class. 
    \end{proof}
\cref{prop:part_1_theorem2_main} proves the first part of Theorem 2 in the main text. Before turning to its second part in the next subsection we briefly discuss the causality condition.

 A finite-dimensional operator on $2k$ copies of $\mathcal H_{\rm S}$, \textit{i.e.}, on $\bigotimes_{j =1}^k \mathcal H^{(\mathfrak{i}_j)}_{\rm S}\otimes \mathcal H^{(\mathfrak{o}_j)}_{\rm S} $ is said to be \textit{causal} if when tracing over any $\ell \leq k$ last output spaces $\bigotimes_{j = k+1-\ell}^k \mathcal H^{(\mathfrak{o}_j)}_{\rm S}$, the operator becomes the identity tensor factor on the corresponding input spaces $\bigotimes_{j = k+1-\ell}^k \mathcal H^{(\mathfrak{i}_j)}_{\rm S}$. 
 If for all $\ell$ the operators on the remaining part $j\le k-\ell$  are additionally positive-semidefinite operators then the operator is the \textit{Choi matrix} of a $k$-step process tensor \cite{milzQuantumStochasticProcesses2021}. 
 In Definition 2 in the main text a Fock space vector was defined to be causal if all of its smoothed marginals are causal. 
 In the same vein, with the map $\mathcal P$ at hand this causality condition can be directly transferred to the space of HS-operators on $\gamma$. 
 By construction, the discrete marginals of a cPT $\ket{\Upsilon_T}$ are the values of the integral kernel of the operator $\cPTchoi = \mathcal{P}(\cPT)$.
Let $\zeta_\mu(t) = \phi_\mu(t)$ for $\mu \in [d_{\rm S}^2-1]$ and $\zeta_0(t) = \id$ and let $\ket{\tilde{\Omega}}$ be the vacuum in $\gamma$.  
 Then,
 \begin{equation}
     \Upsilon_{\bm{\nu}}(\bm{t}) = \Upsilon_{(\bm{\mu}, \bm{\mu'})}(\bm{t}) = \braket{\tilde\Omega | \zeta_{\mu_1}(t_1) \dots \zeta_{\mu_n}(t_n) \cPTchoi \zeta_{\mu'_1}^\dag(t_1) \dots \zeta_{\mu'_n}^\dag(t_n) | \tilde{\Omega}} = \left(\cPTchoi\right)^{(n^\ell; n^r)}_{\bm{\mu^\ell};\bm{(\mu')^r}}(\bm{t^\ell};\bm{t^r}),
     \label{eq:causal_cond_choi_oprerator}
 \end{equation}
 where $\bm{\mu^\ell}$ ($(\bm{\mu'})^r$) are all the non-zero elements of $\bm{\mu}$ ($\bm{\mu'}$), $\bm{t^\ell}$ ($\bm{t^r}$) is $\bm{t}$ with times removed wherever $\bm{\mu}$ ($\bm{\mu'}$) is zero, and $n^\ell$ ($n^r$) are the number of elements in $\bm{\mu^\ell}$ ($(\bm{\mu'})^r$), respectively.
Since we know that for a cPT the Choi operator is trace class, we can write this in terms of the SVD as
 \begin{equation}
     \Upsilon_{\bm{\nu}}(\bm{t}) = \sum_{k=0}^\infty\sigma_k u_{k,\bm{\mu^\ell}}^{n^\ell}(\bm{t^\ell})\left(v_{k,(\bm{\mu'})^r}^{n^r}(\bm{t^r})\right)^*.
 \end{equation}
 The causality and positivity condition from \cref{def:cpt} can then be expressed as that for all non-negative $f \in L^2(\Delta_n)$ with $\|f\|_1 = 1$
 \begin{equation}
     \Upsilon^f_{\bm{\nu}} = \int_0^T {\rm D}^n\bm{t} \ f(\bm{t}) \Upsilon_{\bm{\nu}}(\bm{t}) = \int_0^T {\rm D}^n \bm{t} \ f(\bm{t}) \sum_{k=0}^\infty\sigma_k u_{k,\bm{\mu^\ell}}^{n^\ell}(\bm{t^\ell})\left(v_{k,(\bm{\mu'})^r}^{n^r}(\bm{t^r})\right)^*
 \end{equation}
 are elements of a completely positive and causal matrix, i.e. elements of a discrete process tensor.
This leads us to the following conditions on a general HS-operator $\hat\Phi\in \text{HS}(\gamma)$ with integral kernel
\begin{equation}
    \hat \Phi^{(n^\ell; n^r)}_{\bm{\mu^\ell};\bm{(\mu')^r}}(\bm{t^\ell};\bm{t^r}) = \sum_{k=0}^\infty\tilde \sigma_k \tilde u_{k,\bm{\mu^\ell}}^{n^\ell}(\bm{t^\ell})\left(\tilde v_{k,(\bm{\mu'})^r}^{n^r}(\bm{t^r})\right)^*.
\end{equation}
The operator $\hat \Phi$ is the Choi operator of a cPT if
\begin{itemize}
    \item it is trace-class,
    \item for all non-negative $f \in L^2(\Delta_n)$ with $\|f\|_1 = 1$ the matrix with elements
    \begin{equation}
        \Phi^{f}_{(\bm \mu, \bm \mu')} = \int_0^T {\rm D}^n \bm{t} \ f(\bm{t}) \sum_{k=0}^\infty\tilde \sigma_k \tilde u_{k,\bm{\mu^\ell}}^{n^\ell}(\bm{t^\ell})\left(\tilde v_{k,(\bm{\mu'})^r}^{n^r}(\bm{t^r})\right)^*,
    \end{equation}
    where $\bm \mu, \bm \mu' \in \{\{0\}\cup [d_{\rm S}^2]\}^n$ and $\bm{\mu^\ell},\bm{\mu^r},\bm{t^\ell},\bm{t^r}$ as in \cref{eq:causal_cond_choi_oprerator}, is a discrete $n$-step process tensor. 
    In particular, we demand that its matrix elements are finite.
\end{itemize}

\subsection{Choi duality inner product}

We have seen that the map $\mathcal{P}$ outputs HS-operators which are nicely behaved.
It does so by forgetting about some of the information contained in the input vector.
By contrast, the action of $\mathcal{V}$ should recover this information in order to obtain a sensible duality in terms of the trace-like pairing of integral kernels, as stated in the second part of Theorem 2 in the main text and in~\cref{eq:Choi_dual_inner_prod}.
This makes the map $\mathcal{V}$ much more poorly behaved than $\mathcal{P}$: it maps vectors $\ket{\Phi}\in\Gamma$ to sesquilinear forms on some $\mathcal Q_\Phi \times \mathcal Q_\Phi$, where $\mathcal Q_\Phi$ is a suitable subspace of $\gamma$.
The specifics of this subspace do not matter for us since we are here only interested in formalising \cref{eq:Choi_dual_inner_prod}.

The sesquilinear form $\mathcal{V}(\ket{\Phi})$ is defined by its integration kernel as
\begin{equation}
    \mathcal{V}(\ket{\Phi})^{(n; n')}_{\bm{\mu}; \bm{\mu'}}(\bm{t}; \bm{t'}) \coloneqq \!\!\!\sum_{\substack{\bm{s} \in \{0, 1\}^{n} \\ \|\bm{s}\| \le \min(n, n')}} \sum_{\substack{\bm{s'} \in \{0, 1\}^{n'} \\ \|\bm{s'}\| = \|\bm{s}\|}} \prod_{\substack{i, j: s_i = s_j'=1\\\|\bm{s}_{\le i}\|=\|\bm{s'}_{\le j}\|}} \delta(t_i - t'_j) \Phi^{(n + n' - \|\bm{s}\|)}_{(\bm{\mu}_{\bm{s}=1}, \bm{\mu'}_{\bm{s'}=1}), (\bm{\mu}_{\bm{s}=0}, \bm{0}), (\bm{0}, \bm{\mu'}_{\bm{s'}=0})}(\bm{t}_{\bm{s}=1}, \bm{t}_{\bm{s}=0}, \bm{t'}_{\bm{s'}=0}),
\end{equation}
where $\|\bm{s}_{\le i}\| \coloneqq \sum_{k=1}^i s_k$, $\|\bm{s}\| \coloneqq \sum_{k=1}^n s_k$, and $\bm{\mu}_{\bm{s}=0}$ and similar mean that we take the elements of $\bm{\mu}$ at positions where $\bm{s}$ has the zero value.
Note that on the RHS we have broken the time ordering constraint and extended the wavefunctions of $\ket{\Phi}$ to be symmetric under the simultaneous exchange of the times and the indices.
The distribution-valued integration kernel reflects the fact that the output of $\mathcal{V}$ is generally not a Hilbert--Schmidt operator.
In fact, for general $\ket{\Phi}$ it is not even bounded, as an operator on $\gamma$.

We now define the sesquilinear form between integration kernels in~\cref{eq:Choi_dual_inner_prod} by
\begin{equation}
    \langle A, B\rangle_\text{ker} \coloneqq \sum_{n,n'=0}^\infty \sum_{\bm{\mu}\in[d_{\rm S}^2 - 1]^n} \sum_{\bm{\mu'} \in[d_{\rm S}^2 - 1]^{n'}} \int D^n \bm{t} \int D^{n'} \bm{t'} \left(A_{\bm{\mu'};\bm{\mu}}^{(n';n)}(\bm{t'};\bm{t})\right)^* B_{\bm{\mu'};\bm{\mu}}^{(n';n)}(\bm{t'};\bm{t}).
\end{equation}
If $A$ and $B$ are Hilbert--Schmidt operators, $\langle \cdot, \cdot\rangle_\text{ker}$ reduces to the Hilbert--Schmidt inner product.
With this definition we use~\cref{prop:part_1_theorem2_main} to show~\cref{eq:Choi_dual_inner_prod} for any cPT $\ket{\Upsilon}$.
Let $\mathcal{P}(\ket{\Upsilon}) = \sum_k \sigma_k \ket{u^{(k)}}\!\bra{v^{(k)}}$ and consider for any $\ket{\mathcal{I}} \in \Gamma$
\begin{align}
    \begin{split}
        \langle\mathcal V(\ket{\mathcal I}), \;\mathcal P(\ket{\Upsilon})\rangle_{\text{ker}}   &=\sum_{n,n'=0}^\infty \sum_{\bm{\mu}\in[d_{\rm S}^2 - 1]^n}\sum_{\bm{\mu'}\in[d_{\rm S}^2 - 1]^{n'}} \int D^n \bm{t} \int D^{n'}\bm{t'} \sum_{\substack{\bm{s} \in \{0, 1\}^{n} \\ \|\bm{s}\| \le \min(n, n')}} \sum_{\substack{\bm{s'} \in \{0, 1\}^{n'} \\ \|\bm{s'}\| = \|\bm{s}\|}} \prod_{\substack{i, j: s_i = s_j'=1\\\|\bm{s}_{\le i}\|=\|\bm{s'}_{\le j}\|}} \delta(t_i - t'_j) \\
         &\hspace{1cm}\sum_{k} \sigma_k \left(v^{(k, n)}_{\bm{\mu}}(\bm{t})\right)^* \left(\mathcal{I}^{(n + n' - \|\bm{s}\|)}_{(\bm{\mu'}_{\bm{s'}=1}, \bm{\mu}_{\bm{s}=1}), (\bm{\mu'}_{\bm{s'}=0}, \bm{0}), (\bm{0}, \bm{\mu}_{\bm{s}=0})}(\bm{t}_{\bm{s}=1}, \bm{t'}_{\bm{s'}=0}, \bm{t}_{\bm{s}=0})\right)^* u^{(k, n')}_{\bm{\mu'}}(\bm{t'})
\\
    &= \sum_{n,n'=0}^\infty \sum_{\bm{\mu}\in[d_{\rm S}^2 - 1]^n}\sum_{\bm{\mu'}\in[d_{\rm S}^2 - 1]^{n'}} \sum_{\substack{\bm{s} \in \{0, 1\}^{n} \\ \|\bm{s}\| \le \min(n, n')}} \sum_{\substack{\bm{s'} \in \{0, 1\}^{n'} \\ \|\bm{s'}\| = \|\bm{s}\|}} \int D^n \bm{t} \int D^{n' - \|\bm{s}\|}\bm{t'} \sum_{k} \sigma_k \left(v^{(k, n)}_{\bm{\mu}}(\bm{t})\right)^*\\   
    &\hspace{1cm}\left(\mathcal{I}^{(n + n' - \|\bm{s}\|)}_{(\bm{\mu'}_{\bm{s'}=1}, \bm{\mu}_{\bm{s}=1}), (\bm{\mu'}_{\bm{s'}=0}, \bm{0}), (\bm{0}, \bm{\mu}_{\bm{s}=0})}(\bm{t}_{\bm{s}=1}, \bm{t'}_{\bm{s'}=0}, \bm{t}_{\bm{s}=0})\right)^* u^{(k, n')}_{\bm{\mu'}_{\bm{s'}=1}, \bm{\mu'}_{\bm{s'}=0}}(\bm{t}_{\bm{s}=1}, \bm{t'}).
    \end{split}
    \label{eq:first_step_pairing_duality}
\end{align}
Above we used Fubini-Tonelli to exchange the $\sum_k$ and the integrals where absolute integrability follows from $\mathcal{P}(\ket{\Upsilon}) = \sum_k \sigma_k \ket{u^{(k)}}\!\bra{v^{(k)}}$ being trace-class.
Then, reordering the individual integration measures followed by an application of \cref{eq:conditions_Choi_dual_PPtilde} to $\ket{\Upsilon}$ we get that
\begin{align}
    \begin{split}
        \langle\mathcal V(\ket{\mathcal I}), \;\mathcal P(\ket{\Upsilon})\rangle_{\text{ker}}   &= \sum_{n,n^\ell,n^r=0}^\infty  \sum_{\bm{\mu},\bm{\mu'}\in[d_{\rm S}^2 - 1]^n} \sum_{\bm{\mu^\ell}\in[d_{\rm S}^2 - 1]^{n^\ell}} \sum_{\bm{\mu^r}\in[d_{\rm S}^2 - 1]^{n^r}} \int D^n \bm{t} \int D^{n^\ell}\bm{t^\ell} \int D^{n^r} \bm{t^r} \\
    &\hspace{2cm}\sum_{k} \sigma_k  \left(v^{(k, n+n^r)}_{\bm{\mu'},\bm{\mu^r}}(\bm{t},\bm{t^r})\right)^*
   \left(\mathcal{I}^{(n + n^\ell + n^r)}_{(\bm{\mu}, \bm{\mu'}), (\bm{\mu^\ell}, \bm{0}), (\bm{0}, \bm{\mu^r})}(\bm{t}, \bm{t^\ell}, \bm{t^r})\right)^* u^{(k, n + n^\ell)}_{\bm{\mu}, \bm{\mu^\ell}}(\bm{t}, \bm{t^\ell})\\
    &= \sum_{n,n^\ell,n^r=0}^\infty \sum_{\bm{\mu},\bm{\mu'}\in[d_{\rm S}^2 - 1]^n} \sum_{\bm{\mu^\ell}\in[d_{\rm S}^2 - 1]^{n^\ell}} \sum_{\bm{\mu^r}\in[d_{\rm S}^2 - 1]^{n^r}} \int D^n \bm{t} \int D^{n^\ell}\bm{t^\ell} \int D^{n^r} \bm{t^r} \\
    &\hspace{4cm}\left(\mathcal{I}^{(n + n^\ell + n^r)}_{(\bm{\mu}, \bm{\mu'}), (\bm{\mu^\ell}, \bm{0}), (\bm{0}, \bm{\mu^r})}(\bm{t}, \bm{t^\ell}, \bm{t^r})\right)^* \Upsilon^{(n + n^\ell + n^r)}_{(\bm{\mu},\bm{\mu'}),(\bm{\mu^\ell},\bm{0}),(\bm{0},\bm{\mu^r})}(\bm{t},\bm{t}^\ell,\bm{t}^r)
  \\
    &= \braket{\mathcal I | \Upsilon}.
    \end{split}
\end{align}
Intuitively, one can understand the two maps $\mathcal V$ and $\mathcal P$ as inducing a Gelfand-triple structure $\mathcal Q\subset \text{HS}(\gamma) \subset \mathcal Q^\ast$.
The map $\mathcal P$ is lossy but regularising: it outputs only Hilbert--Schmidt operators and for cPTs even trace-class operators.
On the other side, $\mathcal V$ is faithful but rough: it takes values in the continuous maps $\mathcal Q \to \mathcal Q^\ast$ rather than in operators on $\gamma$.
In the pairing of $\mathcal V$ and $\mathcal P$ by \cref{eq:Choi_dual_inner_prod}, the roughness of $\mathcal V$ is absorbed by the regularity of $\mathcal P$ whenever $\mathcal{P}$ acts on a cPT such that~\cref{eq:conditions_Choi_dual_PPtilde} holds.
Whenever $\ket{\mathcal I}$ has finite particle number and bounded amplitudes, $ \mathcal V(\ket{\mathcal I})$ is a bounded operator. In this case, $\mathcal V(\ket{\mathcal I})^\dagger \mathcal P(\ket{\Upsilon})$ is of trace class and the pairing coincides with the Hilbert--Schmidt inner product: $ \bigl\langle\mathcal V(\ket{\mathcal I}), \;\mathcal P(\ket{\Upsilon})\bigr\rangle_{\text{ker}} \equiv \tr(\mathcal V(\ket{\mathcal I})^\dagger \mathcal P(\ket{\Upsilon}))$.

Concluding, even though there is no genuine isomorphism between vectors in $\Gamma$ and Hilbert--Schmidt operators on $\gamma$, the operators $\mathcal{P}$ and $\mathcal{V}$ define a sensible notion of a Choi duality in the continuum.
This duality is considerably more complex than its discrete counterpart, where especially $\mathcal{V}(\ket{\mathcal{I}})$ is a difficult object to work with.
However, the Choi operator $\mathcal P(\ket{\Upsilon})$ actually contains less information than $\ket{\Upsilon}$ and hence could prove more useful in some applications.
We would like to emphasise that we have already made use of it in calculating distance measures on cPTs (see the accompanying work \cite{wassnerOperationalContinuumLimit2026}).
